\documentclass[onecolumn]{IEEEtran} 
\usepackage{amsthm,amsmath,amssymb,amsfonts,bm}
\usepackage{tikz}
\usepackage{booktabs,multirow,nicematrix}
\usepackage[table,xcdraw]{xcolor}
\usepackage{adjustbox}
\usepackage[shortlabels]{enumitem}
\usepackage[square, comma, sort&compress, numbers]{natbib}
\usepackage{float,graphicx}
\usepackage[colorlinks,linkcolor=blue]{hyperref}

\renewcommand{\eqref}[1]{\textcolor{blue}{(\ref{#1})}}

\newtheorem{theorem}{Theorem}[section]
\newtheorem{lemma}[theorem]{Lemma}
\newtheorem{proposition}[theorem]{Proposition}
\newtheorem{corollary}[theorem]{Corollary}
\newtheorem{definition}[theorem]{Definition}

\newtheorem{remark}[theorem]{Remark}

\makeatletter
\@addtoreset{equation}{section}
\makeatother

\begin{document}
	\title{The Hermitian Hull Dimensions for a Class of $(\mathcal{L},\mathcal{P})$-Twisted Generalized
		Reed–Solomon Codes}
	\author{Chenlu Jia, Zhonghao Liang, Yue Huang and Qunying Liao
		\thanks{Corresponding author: Qunying Liao.  Emails:3120193984@qq.com; liangzhongh0807@163.com; 2907245673@qq.com; qunyingliao@sicnu.edu.cn.}
		\thanks{Chenlu Jia, Zhonghao Liang, Yue Huang and Qunying Liao are with College of Mathematical Sciences, Sichuan Normal University, Chengdu 610066, China.} 
		\thanks{This paper is supported by National Natural Science Foundation of China (12471494) and Natural Science Foundation of Sichuan Province (2024NSFSC2051).}
	}
	\maketitle
	
	\begin{abstract}
		Determining the hull of linear codes has long been an important topic in coding theory. 
		Recently, non-generalized Reed–Solomon (in short, non-GRS) codes have attracted extensive research interest. The 
		$(\mathcal{L},\mathcal{P})$-twisted generalized Reed–Solomon (in short, $(\mathcal{L},\mathcal{P})$-TGRS) code, which is an extension of the generalized Reed-Solomon (GRS) code, constitutes a well-studied calss of non-GRS codes. 
		There are numerous works focusing on the Euclidean hull of $(\mathcal{L},\mathcal{P})$-TGRS codes, while only a few results on the Hermitian hull of $(\mathcal{L},\mathcal{P})$-TGRS codes.
		In this paper, we focus on a class of $(\mathcal{L},\mathcal{P})$-TGRS codes $\mathcal{C}_{k}(\boldsymbol{\alpha})$. By taking a special class of the vector $\boldsymbol{\alpha}$ with length $i(q-1)$, and analyze the parity of  $i$ and the relation between $i$ and $q+1$, we divide three cases to fully determine the Hermitian hull dimension of $\mathcal{C}_{k}(\boldsymbol{\alpha})$. As an application, we construct two classes of entanglement-assisted quantum error-correcting codes.
		

	\end{abstract} 
	
	\section{Introduction} 
	Let $\mathbb{F}_q$ be the finite field with $q$ elements, where $q$ is a prime power. An $[n,k,d]_{q^2}$ linear code $\mathcal{C}$ is a $k$-dimensional subspace of $\mathbb{F}_{q^2}^n$ with minimum distance $d$. The Hermitian inner product of two vectors $\mathbf{a} = (a_1,\dots,a_n)$ and $ \mathbf{b} = (b_1,\dots,b_n) $ over $ \mathbb{F}_{q^2}^n$ is defined by 
	$\langle \mathbf{a},\mathbf{b} \rangle_\mathrm{H} = \sum\limits_{i=1}^n a_i b_i^q$. 
	The Hermitian dual of $\mathcal{C}$ is
	$$
	\mathcal{C}^{\perp_\mathrm{H}} = \{\mathbf{x} \in \mathbb{F}_{q^2}^n: \langle \mathbf{x},\mathbf{c} \rangle_\mathrm{H} = 0,\text{ for all } \mathbf{c} \in \mathcal{C}\}.
	$$

	For a linear code $\mathcal{C}$, the hull  $\mathrm{Hull}(\mathcal{C})$  is defined as the intersection of $\mathcal{C}$ and its dual code. 
	It is well-known that the
	value of $\dim(\mathrm{Hull}(\mathcal{C}))$ plays a critical role in determining the computational complexity of algorithms to check the permutation equivalence of two linear codes\cite{Hull}, computing the automorphism group of a linear code\cite{Hull 6}, calculating the number of shared pairs that required to construct an entanglement-assisted quantum error-correcting code (in short, EAQECC)\cite{EAQECC}. And so, it is very important to determine  $\dim(\operatorname{Hull}(\mathcal{C}))$\citep{Hull 1,Hull 2,Hull 3,Hull 4,Hull 5,Hull TGRS 2,RS}.

	
	
	In recent years, the construction of non-GRS type linear codes has attracted considerable attention due to that they can effectively resist the Sidelnikov-Shestakov attack and the Wi-eschebrink attack. 
	So far, there are extensive study on the properties and constructions of non-GRS codes\citep{non-GRS liang 1,non-GRS  5,non-GRS liang 2,non-GRS liang 3,non-GRS liang 4,non-GRS  6,non-GRS  9,non-GRS 10,non-GRS 12,non-GRS 13,non-GRS 14,non-GRS 15}. In particular, in 2017, Beelen et al.  \cite{TGRS} firstly introduced the twisted generalized Reed-Solomon (in short, TGRS) code. 
	Subsequently, many scholars studied the TGRS code, including the NMDS properties, self-dual properties, self-orthogonal properties, and so on\citep{Hull TGRS 3,TGRS property 2,Hull TGRS 8,TGRS property 4,TGRS property 5,TGRS property 6,TGRS property 7,TGRS property 8,TGRS property 9,TGRS property 10}. In 2025, Zhao et al.\cite{A-TGRS} generalized the definition of the TGRS code to be the arbitrary twisted generalized Reed-Solomon (in short, A-TGRS) code. And then they constructed several classes of Hermitian self-dual A-TGRS codes\cite{A-TGRS 2}. 
	Recently, Hu et al.\cite{L P-TGRS} generalized TGRS codes to be the most general form, namely, $(\mathcal{L},\mathcal{P})$-twisted generalized Reed-Solomon (in short, $(\mathcal{L},\mathcal{P})$-TGRS) codes, and presented an in-depth and comprehensive investigation. So far, there are many study focusing on some special $(\mathcal{L},\mathcal{P})$-TGRS codes\citep{L P-TGRS,L P liang 1,L P liang 2,L P 3,L P 4,Guo}.
	
	
	
	To date, there are numerous works focused on Euclidean hulls of $(\mathcal{L},\mathcal{P})$-TGRS codes\citep{Hull TGRS 1,Hull TGRS 2,Hull TGRS 3,Hull TGRS 4,Hull TGRS 5,sok,Hull TGRS 7,Hull TGRS 8,Hull TGRS 9,Hull TGRS 10,Hull TGRS 11}. However, there exist only a few results on Hermitian hulls of $(\mathcal{L},\mathcal{P})$-TGRS codes, as listed below. 
	
	\begin{itemize}
		\item In 2021, Wu et al.\cite{Hull TGRS 10}  constructively proved that there exist  $(\mathcal{L},\mathcal{P})$-TGRS codes $\mathcal{C}(\mathcal{L},\mathcal{P},\boldsymbol{B}_1)$ with zero-dimensional Hermitian hull, where 
		$$
		\boldsymbol{B}_1 =
		\begin{pmatrix}
			\boldsymbol{0}_{h-1\times t-1}&\boldsymbol{0}_{h-1\times 1}&\boldsymbol{0}_{h-1\times n-k-t}\\
			\boldsymbol{0}_{1\times t-1} &d_{h,t}& \boldsymbol{0}_{1\times n-k-t} \\
			\boldsymbol{0}_{k-h\times t-1}&\boldsymbol{0}_{k-h\times 1}&\boldsymbol{0}_{k-h\times n-k-t}
		\end{pmatrix}_{k\times(n-k)}
		\left(0 \le h \le k-1, 0 \le t \le n-k-1 \right).
		$$	
		
		\item In 2022, Lin Sok \cite{sok}  constructively proved that there exist $(\mathcal{L},\mathcal{P})$-TGRS codes $\mathcal{C}(\mathcal{L},\mathcal{P},\boldsymbol{B}_2)$ with arbitrary Hermitian hull dimension, where 
		$$
		\boldsymbol{B}_2 =
		\begin{pmatrix}
			\boldsymbol{0}_{k-1\times 1}&\boldsymbol{0}_{k-1\times n-k-1}\\
			b_{k-1,0} &\boldsymbol{0}_{1\times n-k-1} \\
		\end{pmatrix}_{k\times(n-k)}.
		$$	
		
		\item In 2022, Luo et al. \cite{Luo} constructively proved that there exist $(\mathcal{L},\mathcal{P})$-TGRS codes $\mathcal{C}(\mathcal{L},\mathcal{P},\boldsymbol{B}_3)$ with $\mathrm{dim}(\mathrm{Hull}_{H}(\mathcal{C}_k))=k$, where 
		$$
		\boldsymbol{B}_3 =
		\begin{pmatrix}
			d_{0,0}&\boldsymbol{0}_{1\times n-k-1}\\
			\boldsymbol{0}_{k-1\times 1} &\boldsymbol{0}_{k-1\times n-k-1} \\
		\end{pmatrix}_{k\times(n-k)}.
		$$	
		

		\item In 2025, Gao et al. \cite{H Gao 1} determined the  Hermitian hull dimension for the $(\mathcal{L},\mathcal{P})$-TGRS code $\mathcal{C}(\mathcal{L},\mathcal{P},\boldsymbol{B}_4)$ with 
		$$
		\boldsymbol{B}_4 =
		\begin{pmatrix}
			\boldsymbol{0}_{k-1\times 1}&\boldsymbol{0}_{k-1\times n-k-1}\\
			1 &\boldsymbol{0}_{1\times n-k-1} \\
		\end{pmatrix}_{k\times(n-k)}.
		$$

		\item In 2026, for a special class of $(\mathcal{L},\mathcal{P})$-TGRS codes  $\mathcal{C}(\mathcal{L},\mathcal{P},\boldsymbol{B}_5)$, Lao et al. \cite{Bound} gave a  upper bound and a lower bound for $\mathrm{dim}\left(\mathrm{Hull}_{H}\left(\mathcal{C}(\mathcal{L},\mathcal{P},\boldsymbol{B}_5)\right)\right)$ and some sufficient conditions for that the code $\mathcal{C}(\mathcal{L},\mathcal{P},\boldsymbol{B}_5)$ has a given  Hermitian hull dimensions. Here, 
		$\mathcal{P} \subseteq \{h_1,h_2,\ldots,h_\ell\}\subseteq \{0,1,\ldots,k-1\}$,  $\mathcal{L}\subseteq\{t_1,t_2,\ldots,t_\ell\}\subseteq \{1,2,\ldots,n-k\}$ with $h_1,h_2,\ldots,h_\ell$ are distinct and $t_1,t_2,\ldots,t_\ell$ are distinct, and for each integer $s$ with $1\le s \le \ell$, 
		$$b_{ij}=\begin{cases}
			\eta_s, &\text{if }i=h_s\text{ and }j=t_s-1;\\
			0,&\text{otherwise}. 
		\end{cases}$$
	\end{itemize}
	

	Motivated by the above works, in this paper, we focus on a class of $(\mathcal{L},\mathcal{P})$-TGRS codes  $\mathcal{C}(\mathcal{L},\mathcal{P},\boldsymbol{B})$ with 
	$$
	\boldsymbol{B} =
	\begin{pmatrix}
		\boldsymbol{0}_{k-2\times 1}&\boldsymbol{0}_{k-2\times 1}&\boldsymbol{0}_{k-2\times n-k-2}\\
		b_{k-2,0}&b_{k-2,1} &\boldsymbol{0}_{1\times n-k-2} \\
		b_{k-1,0}&b_{k-1,1} &\boldsymbol{0}_{1\times n-k-2} \\
	\end{pmatrix}_{k\times(n-k)},
	$$
	where $\left\{b_{k-2,0}, b_{k-2,1},  b_{k-1,0},b_{k-1,1}\right\} \subseteq \mathbb{F}_{q^2}$. By taking a special class of the vector $\boldsymbol{\alpha}$, we divide three cases to completely determine the corresponding Hermitian Hull dimensions.  
	
	The paper is organized as follows. In Section \ref{section 2}, we give the definition of the $(\mathcal{L},\mathcal{P})$-TGRS code and some necessary lemmas. In Section \ref{section 3}, we determine the Hermitian hull dimension for a class of $(\mathcal{L},\mathcal{P})$-TGRS codes, and then obtain two classes of EAQECCs. In Section \ref{section 4}, we give some corresponding examples. In Section \ref{section 5}, we conclude the whole paper.

	\section{Preliminaries}
	\label{section 2}
	
	Throughout this paper, for convenience, we fix some notations as  follows. 
	
	\begin{itemize}
		\item $q=p^m$ and $j=\frac{q-3}{2}$, where $p$ is an odd prime and $m$ is a positive integer.
		\item $\mathbb{F}_{q^2}$ is the finite field with $q^2$ elements and $\mathbb{F}_{q^2}^{*}=\mathbb{F}_{q^2}\setminus \{0\}=\left< \gamma \right>$.
		\item $\mathbb{Z}$ denotes the set of all integers.
		\item For any matrix $\boldsymbol{G}$, $\boldsymbol{G}^\dagger$ denotes the transpose of $\boldsymbol{G}$ with respect to the Hermitian inner product.
		\item For the linear code $\mathcal{C}$, $\dim(\mathrm{Hull}_{H}(\mathcal{C}))$ denotes the Hermitian hull dimension of $\mathcal{C}$. 
		\item For any integer $r$ with $1\le r\le q$, $c_{r}=\left( q-1 \right) \beta ^{r\left( q-1 \right) }$.
		\item For any integers $r$ and $i$ with $2\le i\le q$, $1\le r\le q$, $e_{r}=\left( q-1 \right) \sum\limits_{t=0}^{i-1}{\gamma ^{rt\left( q-1 \right)}}$. 
		\item $\varGamma=\Delta\left[1+\left(\Delta\cdot\gamma^{(i-1)}\right)^{q-1}\right]\left(1+\gamma ^{i\left( q-1 \right)}\right)\left(1+\gamma ^{2\left( q-1 \right)}\right)-2\left(\frac{b_{k-2,1}}{b_{k-2,0}}\right)^{q+1}\gamma^{i(q-1)}\cdot\left(1+\gamma ^{\left( q-1 \right)}\right)$, 
		where 
		
		$\Delta=b_{k-1,1}-\frac{b_{k-2,1}b_{k-1,0}}{b_{k-2,0}}$.
		\item $\varGamma_{1}=b_{k-1,1}\left(\gamma^{i(q-1)}+1\right)\left(1+\gamma^{q-1}\right)\left(1+b_{k-1,1}^{q-1}\gamma^{(i-1)(q-1)}\right)-b_{k-1,0}^{q+1}\left(\gamma^{(i-1)(q-1)}-1\right)\left(\gamma^{(i+1)(q-1)}-1\right)$.
		\item $\varGamma_{2}=2b_{k-1,1}\left(1+\gamma^{1-q}\right)\left(1+b_{k-1,1}^{q-1}\gamma^{(q-1)(i-1)}\right)+b_{k-1,0}^{q+1}\left(1+\gamma^{i(q-1)}\right)\left(1+\gamma^{2(1-q)}\right)$. 
		
	\end{itemize}


	In this section, we recall the definition of the $(\mathcal{L},\mathcal{P})$-twisted generalized Reed-Solomon code, and give some necessary lemmas. 
	
	The definition of the $(\mathcal{L},\mathcal{P})$-twisted generalized Reed-Solomon code is given in the following 
	
	\begin{definition} $($\textup{\cite{L P-TGRS}}, \text{Definition $2$}$)$ 
		\label{D2.1}
		Let $n$, $k$ and $\ell$ be integers with $0 < k \leq n$ and $0 \le \ell \le n-k$,
		$\boldsymbol{B} = (b_{i,j})_{k \times (n-k)}$, $\mathcal{L} \subseteq \{0,1,\dots,n-k-1\}$ and $\mathcal{P} \subseteq \{0,1,\dots,k-1\}$.
		Let $\boldsymbol{\alpha} = (\alpha_1,\dots,\alpha_n) \in \mathbb{F}_{q^2}^n$ with $\alpha_i \neq \alpha_j$ ($i \neq j$), $\boldsymbol{v} = (v_1,\dots,v_n) \in (\mathbb{F}_{q^2}^*)^n$.
		The $(\mathcal{L},\mathcal{P})$-twisted generalized Reed-Solomon (in short, $(\mathcal{L},\mathcal{P})$-TGRS) code is defined as
		$$
		\mathcal{C}(\mathcal{L},\mathcal{P},\boldsymbol{B}) \triangleq \bigl\{ (v_1 f(\alpha_1),\dots,v_n f(\alpha_n)) \,\big|\, f(x) \in \mathcal{F}(\mathcal{L},\mathcal{P},\boldsymbol{B}) \bigr\},
		$$
		where
		\[
		\mathcal{F}(\mathcal{L},\mathcal{P},\boldsymbol{B}) = \left\{ \sum_{i=0}^{k-1} f_i x^i + \sum_{i \in \mathcal{P}} f_i \sum_{j \in \mathcal{L}} b_{i,j} x^{k+j} \,\bigg|\, f_i \in \mathbb{F}_{q^2},\ 0 \leq i \leq k-1 \right\}. 
		\]
		Specifically, when $\boldsymbol{v} = (1,1, \ldots, 1)\in (\mathbb{F}_{q^2}^*)^n$, the linear code is called a  $(\mathcal{L},\mathcal{P})$-TRS code.
		
	\end{definition}

	In this paper, we consider a special class of $(\mathcal{L},\mathcal{P})$-TRS codes with
	$$
	\boldsymbol{B} =
	\begin{pmatrix}
		0&0&0&\cdots&0\\
		\vdots&\vdots&\vdots& &\vdots\\
		0&0&0&\cdots&0\\
		b_{k-2,0} & b_{k-2,1}&0&\cdots&0 \\
		b_{k-1,0} & b_{k-1,1}&0&\cdots&0 \\
	\end{pmatrix}_{k\times(n-k)},
	$$
	where $\left\{b_{k-2,0}, b_{k-2,1},  b_{k-1,0},b_{k-1,1}\right\} \subseteq \mathbb{F}_{q^2}$, 
	and briefly denote it as  $\mathcal{C}_{k}(\boldsymbol{\alpha})$.
	
	\begin{remark}
		\label{G k}
		By Definition \ref{D2.1}, it is easy to know that $\mathcal{C}_{k}(\boldsymbol{\alpha})$ has the generator matrix
		\begin{equation}
			\label{G}
			\boldsymbol{G}=
			\begin{pmatrix}
				1 & \cdots & 1 \\
				\alpha_1 & \cdots & \alpha_n \\
				\vdots & & \vdots \\
				\alpha_1^{k-3} & \cdots & \alpha_n^{k-3} \\[6pt]
				\displaystyle \alpha_1^{k-2} + \sum_{j=0}^{1} b_{k-2,j} \alpha_1^{k+j}
				& \cdots &
				\displaystyle \alpha_n^{k-2} + \sum_{j=0}^{1} b_{k-2,j} \alpha_n^{k+j} \\[8pt]
				\displaystyle \alpha_1^{k-1} + \sum_{j=0}^{1} b_{k-1,j} \alpha_1^{k+j}
				& \cdots &
				\displaystyle \alpha_n^{k-1} + \sum_{j=0}^{1} b_{k-1,j} \alpha_n^{k+j}
			\end{pmatrix}.
		\end{equation}	
	\end{remark}



The following Lemma \ref{L2.1 } is crucial for calculating the matrix $\boldsymbol{GG}^\dagger$. 

\begin{lemma}
	\label{L2.1 }
	$($\textup{\cite{ref2}}$)$
	Let $s$ be a positive integer with $s \mid q^2-1$, and 
	$\alpha_i = \gamma^{\frac{q^2-1}{s}i}$ for $1 \le i \le s$. Then for any integer $t$ and $\beta \in \mathbb{F}_{q^2}^*$, we have
	\[
	\sum_{i=1}^{s} (\beta \alpha_i)^t =
	\begin{cases}
		\beta^t s, & \text{if } s \mid t; \\
		0, & \text{otherwise}.
	\end{cases}
	\]
\end{lemma}

\begin{remark}
	\label{u+v-2}
	By taking $s=q-1$ in Lemma \ref{L2.1 }, it is easy to know that for any integers $u$ and $v$, we have 
	$$
	\sum_{i=1}^{q-1}{\left( \alpha _i\beta \right) ^{\left( u-1 \right) +q\left( v-1 \right)}}
	= \begin{array}{c}
		\begin{cases}
			\left( q-1 \right) \beta ^{\left( u-1 \right) +q\left( v-1 \right)} , &\text{if } \left( q-1 \right) \mid \left( u+v-2 \right);\\
			0 , &\text{otherwise}.\\
		\end{cases}
	\end{array}$$
	
\end{remark}

For an $[n,k,d]_{q^2}$ linear code $\mathcal{C}$, the following Lemmas \ref{L2.16 }-\ref{L2.17} provide a method for calculating  $\dim(\mathrm{Hull}_H(\mathcal{C}))$ and constructing EAQECCs. 

\begin{lemma}
	\label{L2.16 }
	$($\textup{\cite{EAQECC}}, \text{Proposition $3.2$}$)$ Let $\mathcal{C}$ be a classical $[n,k,d]_{q^2}$ code with parity-check matrix $\boldsymbol{H}$ and generator matrix $\boldsymbol{G}$. Then  
	$$
	\mathrm{rank}(\boldsymbol{HH}^\dagger) = n - k - \dim(\mathrm{Hull}_H(\mathcal{C})) $$
	and $$
	\mathrm{rank}\left( \boldsymbol{GG}^\dagger \right) = k - \dim(\mathrm{Hull}_H(\mathcal{C})) .$$
\end{lemma}

\begin{lemma}
	$($\textup{\cite{EAQECC}}, \text{Corollary $3.2$}$)$
	\label{L2.17}
	Let $\mathcal{C}$ and $\mathcal{C}^{\perp_H}$ be a classical linear code and its Hermitian dual with the parameters $[n,k,d]_{q^2}$ and $[n,k,d^{\perp_H}]_{q^2}$, respectively. Then there exist two EAQECCs with the parameters
	\[
	\bigl[\bigl[n,\,k - \dim\bigl(\mathrm{Hull}_H(\mathcal{C})\bigr),\,d,\,n - k - \dim\bigl(\mathrm{Hull}_H(\mathcal{C})\bigr)\bigr]\bigr]_q
	\]
	and
	\[
	\bigl[\bigl[n,\,n - k - \dim\bigl(\mathrm{Hull}_H(\mathcal{C})\bigr),\,d^{\perp_H},\,k - \dim\bigl(\mathrm{Hull}_H(\mathcal{C})\bigr)\bigr]\bigr]_q,
	\]
	respectively. Moreover, if $\mathcal{C}$ is MDS, then the above two EAQECCs are also MDS.
\end{lemma}

The following Lemmas \ref{L2.2 }-\ref{L2.14 } are crucial for proving our main results. And the proofs of Lemmas \ref{L2.6 1}-\ref{L2.14 } are given in the Appendix \ref{appendex A}.

\begin{lemma}
	\label{L2.2 }
	For any integers $i$ and $r$ with $2\le i\le q$ and $ 1\le r\le q$, then 
	$e_{r}=0$ if and only if $ri \equiv 0 \pmod{q+1}$.
\end{lemma}
\textbf{Proof.} For any $2\le i\le q$ and $1 \le r \le q$, it is easy to know that   $\gamma^{r(q-1)}-1\in \mathbb{F}_{q^2}^{*}$, then $e_r=0$ if and only if 
$\left(\gamma^{r(q-1)}-1\right)\left( q-1 \right)\sum\limits_{t=0}^{i-1}{ \gamma ^{rt\left( q-1 \right)}}=0$,  i.e., $(q-1)\left(\gamma^{ri(q-1)}-1\right)=0$. 
Now by $p\nmid q-1$. we know that  $e_{r} = 0$ if and only if $\gamma^{ri(q-1)}=1$, i.e., $ri \equiv 0 \pmod{q+1}$.

\begin{remark}
	\label{C2.1}
	$(1)$ If  $r=1$ or $q$,  then $\gcd(r, q+1)=1$. Furthermore, by Lemma \ref{L2.2 }, we know that $e_r=0$ if and only if $i\equiv 0\pmod {q+1}$. Note that $2\le i\le q$, and so $i\not\equiv 0\pmod {q+1}$, thus $e_r\neq0$, i.e., $e_1\neq 0$ and $e_q\neq 0$. 
	
	$(2)$ If $\gcd(i, q+1) = 1$, then by Lemma \ref{L2.2 }, we know that $e_r = 0$ if and only if $r \equiv 0 \pmod{q+1}$. Furthermore, for any $1 \le r \le q$, $e_{r} \neq 0$. 
\end{remark}

\begin{lemma}
	\label{L2.3 }
	Let $k_r=\dfrac{rh}{q+1}$. If $\gcd\left( i,q+1 \right) = h>1$, then for any $1\le r\le q$, $e_{r}=0$ if and only if $k_r\in \mathbb{Z}$. 
\end{lemma}
\textbf{Proof.} By Lemma \ref{L2.2 } and $\gcd\left( i,q+1 \right) = h$, we know that $e_r = 0$ if and only if  $r\equiv 0\left(\bmod \frac{q+1}{h}\right) $, i.e., there exists some $k_r\in \mathbb{Z}$ such that  $r=\frac{k_r(q+1)}{h}$. Namely, $e_r=0$ if and only if  $k_r=\dfrac{rh}{q+1}\in\mathbb{Z}$.

\begin{lemma}
	\label{L2.4 } If $\gcd\left( i,q+1 \right) = h>1$ and $2\le i\le q$, then the following two statements are true. 
	
	$(1)$ For $1\le r\le q-1$, $e_r$ and $e_{r+1}$ are not zeros simultaneously.
	
	$(2)$ For $1\le r\le q$, $e_r$ and $e_{q+1-r}$ are both zeros or not simultaneously. 
\end{lemma}
\textbf{Proof.} \textbf{(1)} From $2\le i\le q$, we have $h=\gcd\left( i,q+1 \right)\le q<q+1$, then $\frac{h}{q+1}\notin \mathbb{Z}$. Note that 
$k_{r+1}=\frac{(r+1)h}{q+1}=\frac{h}{q+1}+k_r$, thus $k_r$ and $k_{r+1}$ are not integers simultaneously. Furthermore, by Lemma \ref{L2.3 }, $e_r$ and $e_{r+1}$ are not zeros simultaneously. 

\textbf{(2)} Note that $k_{q+1-r}=\frac{(q+1-r)h}{q+1}=1-k_r$, and so,  $k_r$ and $k_{q+1-r}$ are both integers or not simultaneously. Furthermore, by Lemma \ref{L2.3 }, we know that $e_r$ and $e_{q+1-r}$ are both zeros or not simultaneously. 

\begin{lemma}
	\label{L2.5 }
	For $2\le i\le q$, we have  $e_{j+2}=\begin{array}{c}
		\begin{cases}
			0, & \text{if }  2\mid i;\\
			q-1,& \text{if } 2\nmid i. \\
		\end{cases} 
	\end{array} $
\end{lemma}	
\textbf{Proof.} By $j=\frac{q-3}{2}$, we have $j+2 =\frac{q+1}{2}$, then 
$\gamma^{(j+2)(q-1)} = \gamma^{\frac{q^2-1}{2}} = -1$. Furthermore, $$e_{j+2}=\left( q-1 \right) \sum\limits_{t=0}^{i-1}{\gamma ^{(j+2)t\left( q-1 \right)}}= \left( q-1 \right) \sum\limits_{t=0}^{i-1}{(-1)^t}= \begin{array}{c}
	\begin{cases}
		0, & \text{if }  2\mid i;\\
		q-1,& \text{if } 2\nmid i. \\
	\end{cases} 
\end{array}$$

\begin{lemma}
	\label{L2.6 1}
	Let $2\le i\le q$. If $\gcd\left( i,q+1 \right)= h>1$ and $2\mid i$, then $e_j=e_{j+4}=0$ if and only if $i=\frac{q+1}{2}$. 
\end{lemma}

\begin{lemma}
	\label{L2.6 } 
	Let $2\le i\le q$. If $\gcd\left( i,q+1 \right)= h>1$, then we have 
	$$ e_{j+1}=e_{j+3}=0\Longleftrightarrow
	 2\nmid i \text{ and } i=\frac{q+1}{2}\Longleftrightarrow
	q \equiv 1 \pmod 4 \text{ and } \frac{q+1}{2}\mid i.$$ 
	
	
	
\end{lemma}

\begin{lemma}
	\label{L2.7 }
	Let $2\le i\le q$. If $\gcd\left( i,q+1 \right)= h>1$ and  $e_{j+2}\neq 0$, then we have 
	$$ e_{j+4}=e_{j}=0\Longleftrightarrow  
	 q\equiv-1\pmod 4\text{ and } i=\frac{q+1}{4}\text{ or }\frac{3(q+1)}{4}\Longleftrightarrow 
	 q\equiv-1\pmod 4\text{ and } \frac{q+1}{4}\mid i.$$
	
	
	
\end{lemma}

\begin{lemma}
	\label{L2.8 }
	Let $2\le i\le q$. If $\gcd (i, q+1)=h>1$, then for $r\in\left\{j, j+1, j+2, j+3, j+4\right\}$, $e_{r}\neq 0$ if and only if $2\nmid i$ and $\frac{q+1}{2}\nmid 2i$.
\end{lemma}





\begin{lemma}
	\label{L2.9 }
	Let $ 2 \le i\le q$ with $i\equiv 0\pmod p$ or $i^2\equiv 1\pmod p$. Then the matrix $\boldsymbol{D}=\begin{pmatrix}
		&i(q-1)&e_{q}\\
		&e_{1} &i(q-1)
	\end{pmatrix}$ over $\mathbb{F}_{q^2}$ has $\mathrm{rank}(\boldsymbol{D})=2$.  
\end{lemma}

\begin{lemma}
	\label{L2.10 }
	Let $2 \le i\le q$. If $2 \nmid i$, then the following statements are true. 
	
	$(1)$ If $i\neq \frac{q+1}{2}$, then $e_{j+2}^2=e_{j+3}e_{j+1}$ if and only if $i=q$. 
	
	$(2)$ If $\frac{q+1}{2}\nmid 2i$, then $e_{j+1}^2=e_{j+2}e_{j}$ if and only if $i=q$. 
	
	$(3)$ If $\frac{q+1}{2}\nmid 2i$, then  $e_{j+3}^2=e_{j+4}e_{j+2}$ if and only if $i=q$. 
\end{lemma}

\begin{lemma}\label{L2.11 }
	Let $2 \le i< q$. If $ 2\nmid i$ and $\frac{q+1}{2}\nmid 2i$, then for $b_{k-1,1}$, $b_{k-1,0}\in \mathbb{F}_{q^2}$, $\varGamma_{2}=0$ if and only if 
	\begin{equation}
		\label{E2.15}
		b_{k-1,1}\left(e_{j+1}-\frac{e_{j+2}}{e_{j+1}}e_j\right) +b_{k-1,1}^{q}\left(e_{j+3}-  \frac{e_{j+4}}{e_{j+3}}e_{j+2}\right)+b_{k-1,0}^{q+1}\left(e_{j+2}-\frac{e_{j+3}}{e_{j+2}}e_{j+1}\right)=0. 
	\end{equation}
	
\end{lemma}

By taking $b_{k-1,0}=0$ in Lemma \ref{L2.11 }, it is easy to obtain the following 
\begin{corollary}
	\label{C2.2}
	Let $2 \le i< q$. If $ 2\nmid i$ and $\frac{q+1}{2}\nmid 2i$, then for $b_{k-1,1}\in \mathbb{F}_{q^2}$, 
	$$b_{k-1,1}\left(e_{j+1}-\frac{e_{j+2}}{e_{j+1}}e_{j}\right)+b_{k-1,1}^{q}\left(e_{j+3}-\frac{e_{j+4}}{e_{j+3}}e_{j+2}\right)=0$$
	if and only if 
	$b_{k-1,1}=0$ or  $b_{k-1,1}^{q-1}\gamma^{(i-1)(q-1)}=-1$. 
\end{corollary}

\begin{lemma}\label{L2.12 }
	Let $2 \le i < q$. If $ 2\nmid i$ and $\left(\frac{q+1}{2}\right)\nmid 2i$, then for $b_{k-2,0}$, $b_{k-2,1}\in \mathbb{F}_{q^2}^{*}$ and
	$b_{k-1,1}$, $b_{k-1,0}\in \mathbb{F}_{q^2}$,  
	\begin{equation}\label{E2.22}
		\frac{b_{k-1,1}}{b_{k-2,1}}-\frac{b_{k-1,0}}{b_{k-2,0}}+\frac{\left(b_{k-1,1}^{q}-\frac{b_{k-1,0}^qb_{k-2,1}^q}{b_{k-2,0}^q}\right)\left(e_{j+3}-\frac{e_{j+4}}{e_{j+3}}e_{j+2}\right)}{b_{k-2,1}\left(e_{j+1}-\frac{e_{j+2}}{e_{j+1}}e_{j}\right)}=\frac{b_{k-2,1}^{q}\left(e_{j+3}-\frac{e_{j+4}}{e_{j+3}}e_{j+2}\right)}{b_{k-2,0}^{q+1}\left(e_{j+2}-\frac{e_{j+3}}{e_{j+2}}e_{j+1}\right)}
	\end{equation}
	if and only if $\varGamma=0$. 
\end{lemma}

By taking $b_{k-1,0}=0$ or $b_{k-1,1}=0$ in Lemma \ref{L2.12 }, it is easy to obtain the following Corollaries \ref{C2.3}-\ref{C2.4}. 
\begin{corollary}
	\label{C2.3}
	Let $2 \le i < q$. If $ 2\nmid i$ and $\left(\frac{q+1}{2}\right)\nmid 2i$, then for $b_{k-2,0}$, $b_{k-2,1}\in \mathbb{F}_{q^2}^{*}$ and $b_{k-1,1}\in \mathbb{F}_{q^2}$, 
	$$\frac{b_{k-2,1}^{q}\left(e_{j+3}-\frac{e_{j+4}}{e_{j+3}}e_{j+2}\right)}{b_{k-2,0}^{q+1}\left(e_{j+2}-\frac{e_{j+3}}{e_{j+2}}e_{j+1}\right)}=\frac{b_{k-1,1}}{b_{k-2,1}}+\frac{b_{k-1,1}^{q}\left(e_{j+3}-\frac{e_{j+4}}{e_{j+3}}e_{j+2}\right)}{b_{k-2,1}\left(e_{j+1}-\frac{e_{j+2}}{e_{j+1}}e_{j}\right)} $$
	if and only if 
	$$b_{k-1,1}\left(1+b_{k-1,1}^{q-1}\cdot\gamma^{(i-1)(q-1)} \right)\left(1+\gamma ^{i\left( q-1 \right)}\right)\left(1+\gamma ^{2\left( q-1 \right)}\right)
	=2\left(\frac{b_{k-2,1}}{b_{k-2,0}}\right)^{q+1}\gamma^{i(q-1)}\cdot\left(1+\gamma ^{\left( q-1 \right)}\right).$$
\end{corollary}

\begin{corollary}
	\label{C2.4}
	Let $2 \le i < q$ and $\Delta_{1}=-\frac{b_{k-2,1}b_{k-1,0}}{b_{k-2,0}}$. If $ 2\nmid i$ and $\left(\frac{q+1}{2}\right)\nmid 2i$, then for $b_{k-2,0}$, $b_{k-2,1}\in \mathbb{F}_{q^2}^{*}$ and $b_{k-1,0}\in \mathbb{F}_{q^2}$, 
	$$\frac{b_{k-2,1}^q\left(e_{j+3}-\frac{e_{j+4}}{e_{j+3}}e_{j+2}\right)}{b_{k-2,0}^{q+1}\left(e_{j+2}-\frac{e_{j+3}}{e_{j+2}}e_{j+1}\right)}+\frac{b_{k-1,0}}{b_{k-2,0}}=\frac{-b_{k-1,0}^qb_{k-2,1}^q\left(e_{j+3}-\frac{e_{j+4}}{e_{j+3}}e_{j+2}\right)}{b_{k-2,0}^qb_{k-2,1}\left(e_{j+1}-\frac{e_{j+2}}{e_{j+1}}e_{j}\right) } $$
	if and only if 
	$$\Delta_{1}\left[1+\left(\Delta_{1}\cdot\gamma^{(i-1)}\right)^{q-1} \right]\left(1+\gamma ^{i\left( q-1 \right)}\right)\left(1+\gamma ^{2\left( q-1 \right)}\right)=2\left(\frac{b_{k-2,1}}{b_{k-2,0}}\right)^{q+1}\gamma^{i(q-1)}\cdot\left(1+\gamma ^{\left( q-1 \right)}\right).$$
\end{corollary}

\begin{lemma}\label{L2.13 }
	Let $ 2 \le i\le q$. If $ 2\mid i$, then for $b_{k-1,1}\in \mathbb{F}_{q^2}$,  
	\begin{equation}
		\label{E2.31}
		b_{k-1,1}e_{j+1}+b_{k-1,1}^{q}e_{j+3}=0
	\end{equation}
	if and only if $b_{k-1,1}=0$ or 
	$b_{k-1,1}^{q-1}\gamma^{(i-1)(q-1)}=1$. 
\end{lemma}

\begin{lemma}
	\label{L 1}
	Let $ 2 \le i\le q$. If $ 2\mid i$, then for $b_{k-2,0}\in \mathbb{F}_{q^2}^{*}$ and $b_{k-2,1}, b_{k-1,0}, b_{k-1,1}\in \mathbb{F}_{q^2}$, 
	$$ \left(b_{k-1,1}-\frac{b_{k-2,1}b_{k-1,0}}{b_{k-2,0}}\right)e_{j+1}+ \left(b_{k-1,1}^q-\frac{b_{k-2,1}^qb_{k-1,0}^q}{b_{k-2,0}^q}\right)e_{j+3}= 0$$
	if and only if $b_{k-1,1}b_{k-2,0}=b_{k-2,1}b_{k-1,0}$ or $\left(b_{k-1,1}-\frac{b_{k-2,1}b_{k-1,0}}{b_{k-2,0}}\right)^{q-1}\gamma^{(i-1)(q-1)}=1$. 
\end{lemma}

\begin{lemma}
	\label{L2.14 }
	Let $2 \le i\le q$. If $2\nmid i$ and $i\neq\frac{q+1}{2}$, then for $b_{k-1,0}$, $b_{k-1,1}\in \mathbb{F}_{q^2}$, $\varGamma_{1}=0$ if and only if  
	\begin{equation}
		\label{E2.35}
		b_{k-1,1}e_{j+1}+b_{k-1,0}^{q+1}\left(e_{j+2}-\frac{e_{j+1}}{e_{j+2}}e_{j+3}\right)+b_{k-1,1}^{q}e_{j+3}=0. 
	\end{equation}
	
\end{lemma}

\section{The Hermitian Hull dimensions of a class of $(\mathcal{L},\mathcal{P})$-TRS codes}
\label{section 3}

Throughout this section, we fix $q\ge7$, $\mathbb{F}_{q^2}^{*}=\left< \gamma \right>$,   $2\le i\le q$ with $i\equiv 0\pmod p$ or $i^2\equiv 1\pmod p$,  $j=\frac{q-3}{2}$, $k=q+j$ and 
$$\boldsymbol{\alpha}=\left(\alpha_1,\alpha_2,\ldots \alpha_{q-1},\alpha_1\gamma,\alpha_2\gamma,\ldots,\alpha_{q-1}\gamma,\ldots \alpha_1\gamma^{i-1},\alpha_2\gamma^{i-1},\ldots ,\alpha_{q-1}\gamma^{i-1}\right),$$
where $\alpha_s=\gamma^{s(q+1)}$  $\left(1\le s\le q-1\right)$. 

In this section, by taking a special class of the vector $\boldsymbol{\alpha}$, we completely determine the corresponding $\mathrm{dim}\left(\mathrm{Hull}_H(\mathcal{C}_{q+j}(\boldsymbol{\alpha}))\right)$.

\subsection{Main results}

In this subsection, we present the value of $\mathrm{dim}\left(\mathrm{Hull}_H(\mathcal{C}_{q+j}(\boldsymbol{\alpha}))\right)$ by $3$ cases, and then obtain two classes of EAQECCs.  

\begin{theorem}
	\label{T3.1}
	If $\gcd(i, q+1)=1$, then 
	$$
	\mathrm{dim}\left(\mathrm{Hull}_H(\mathcal{C}_{q+j}(\boldsymbol{\alpha}))\right) = \begin{array}{c}
		\begin{cases}
			j,&\text{if } i=q\\
			&\quad\text{or } i\neq q \text{ and }  b_{k-2,0}=b_{k-2,1}=b_{k-1,0}=b_{k-1,1}=0\\
			&\quad\text{or } i\neq q ,b_{k-2,1}=b_{k-2,0}=0 , b_{k-1,1}\neq 0 \text{ and } \varGamma_{2}=0;\\
			j-2,&\text{if } i\neq q \text{ and }  b_{k-2,0}=0, b_{k-2,1}\neq 0\\
			&\quad\text{or } i\neq q \text{ and }  b_{k-1,0}=b_{k-1,1}=0, b_{k-2,0}, b_{k-2,1}\neq 0\\ 
			&\quad\text{or } i\neq q, b_{k-2,1}=0, b_{k-2,0},b_{k-1,1}\neq 0 \text{ and }  b_{k-1,1}^{q-1}\gamma^{(i-1)(q-1)}\neq -1\\
			&\quad\text{or } i\neq q , b_{k-2,1},b_{k-2,0}\neq 0,\left(b_{k-1,0},b_{k-1,1}\right)\neq (0,0) \text{ and }  \varGamma\neq 0;\\
			j-1,&\text{otherwise}.
		\end{cases}
	\end{array} . 
	$$ 
\end{theorem}


\begin{theorem}\label{T3.2}
	If $\gcd(i, q+1)=h>1$ and $2\mid i$, then

	$$
	\mathrm{dim}\left(\mathrm{Hull}_{H}{\mathcal{C}_{q+j}(\boldsymbol{\alpha})} \right) = \begin{array}{c}
		\begin{cases}
			j+1,&\text{ if }  i=\frac{q+1}{2}, b_{k-1,0}=b_{k-1,1}=0 \\
			&\quad\text{ or }  i=\frac{q+1}{2},  b_{k-1,0}= 0 \text{ and }  b_{k-1,1}^{q-1}\gamma^{(i-1)(q-1)}=1\\
			&\quad\text{ or } i\neq \frac{q+1}{2} ,  b_{k-2,0}=b_{k-2,1}=b_{k-1,0}=b_{k-1,1}=0\\ 
			&\quad\text{ or } i\neq \frac{q+1}{2} ,  b_{k-2,0}=b_{k-2,1}=b_{k-1,0}=0  \text{ and }  b_{k-1,1}^{q-1}\gamma^{(i-1)(q-1)}=1;\\ 
			
			j,  &\text{ if }  i=\frac{q+1}{2},  b_{k-1,0}=0 ,\quad b_{k-1,1}\neq 0 \text{ and } b_{k-1,1}^{q-1}\gamma^{(i-1)(q-1)}\neq 1\\
			&\quad\text{ or } i\neq \frac{q+1}{2} ,  b_{k-2,0}=b_{k-2,1}=b_{k-1,0}=0, b_{k-1,1}\neq 0 \text{ and }  b_{k-1,1}^{q-1}\gamma^{(i-1)(q-1)}\neq 1;\\ 
			j-2, &\text{ if }  i\neq \frac{q+1}{2} ,  b_{k-2,0}\neq 0,  b_{k-2,1}b_{k-1,0}\neq b_{k-1,1}b_{k-2,0} \\
			&\quad \text{and } \left(b_{k-1,1}-\frac{b_{k-2,1}b_{k-1,0}}{b_{k-2,0}}\right)^{q-1}\gamma^{(i-1)(q-1)}\neq 1;\\
			j-1, &\text{otherwise}. 						
		\end{cases}
	\end{array}  
	$$
\end{theorem}

\begin{theorem}
	\label{T3.3}
	If $\gcd(i, q+1)=h>1$ and $2\nmid i$, then 
	
	$$
	\mathrm{dim}\left(\mathrm{Hull}(\mathcal{C}_{q+j}(\boldsymbol{\alpha})) \right) = \begin{array}{c}
		\begin{cases}
			j, &\text{ if } i=\frac{q+1}{2} \text{ and }  b_{k-2,1}=b_{k-2,0}=0\\
			&\quad\text{ or } i=\frac{q+1}{4} \text{ or } \frac{3(q+1)}{4} \text{ and } \varGamma_{1}=0 \\
			&\quad\text{ or }\frac{q+1}{2}\nmid 2i \text{ and } b_{k-2,0}=b_{k-2,1}=b_{k-1,0}=b_{k-1,1}=0\\
			&\quad\text{ or } \frac{q+1}{2}\nmid 2i , b_{k-2,1}=b_{k-2,0}=0, b_{k-1,1}\neq 0 \text{ and } \varGamma_{2}=0;\\	
			j-2, &\text{ if } i = \frac{q+1}{2}  \text{ and } b_{k-2,1}\neq 0\\
			&\quad\text{ or } \frac{q+1}{2}\nmid 2i , b_{k-2,0}=0 \text{ and } b_{k-2,1}\neq 0\\
			&\quad\text{ or } \frac{q+1}{2}\nmid 2i, b_{k-1,0}=b_{k-1,1}=0 \text{ and } b_{k-2,0}, b_{k-2,1}\neq 0\\ 
			&\quad\text{ or } \frac{q+1}{2}\nmid 2i, b_{k-2,1}=0, b_{k-2,0},b_{k-1,1}\neq 0 \text{ and } b_{k-1,1}^{q-1}\gamma^{(i-1)(q-1)}\neq -1\\
			&\quad\text{ or }\frac{q+1}{2}\nmid 2i,  b_{k-2,1},b_{k-2,0}\neq 0,\left(b_{k-1,0},b_{k-1,1}\right)\neq (0,0)  \text{ and } \varGamma\neq 0;	\\
			j-1, &\text{otherwise}. 
		\end{cases}
	\end{array}  
	$$ 
	
\end{theorem}	

\begin{remark} 
	For the case $i\neq q$ and  $b_{k-2,0}=b_{k-2,1}=b_{k-1,0}=b_{k-1,1}=0$ in Theorems \ref{T3.1}-\ref{T3.3}, the derived conclusions are just the corresponding results in Theorems $3.1$-$3.2$ of Reference $\cite{RS}$ for $j=\frac{q-3}{2}$.

\end{remark}


Then by combining with Lemma \ref{L2.17} and Theorems  \ref{T3.1}-\ref{T3.3}, we can immediately obtain two classes of EAQECCs as follows.

\begin{theorem}
	\label{thm:4.15}
	Assume that $d$ is the minimum distance for the code  $\mathcal{C}_{q+j}(\boldsymbol{\alpha})$. Then 
	there exists a $q$-ary EAQECC with parameters $\bigl[\bigl[i(q-1),\,q-1+m,\,d,\,(i-2)(q-1)+m\bigr]\bigr]_q$ for $m=0,1,2,3$. 
\end{theorem}

\begin{theorem}
	\label{thm:4.16}
	Assume that $d^{\perp_H}$ is the minimum distance for the code $\mathcal{C}_{q+j}(\boldsymbol{\alpha})$. Then 
	there exists a $q$-ary EAQECC with parameters $\bigl[\bigl[i(q-1),\,(i-2)(q-1)+m,\,d^{\perp_H},\,q-1+m\bigr]\bigr]_q$ for $m=0,1,2,3$. 
	
\end{theorem}

\subsection{Three Crucial Propositions}

In this subsection, we present three important propositions and their proofs. 


\begin{proposition}
	\label{P3.1}
	Let $\boldsymbol{G}$ be the generator matrix of $\mathcal{C}_{q+j}(\boldsymbol{\alpha})$, then 
	\begin{equation}
		\label{GG H}
		\boldsymbol{G}\boldsymbol{G}^{\dagger}=\left( \begin{matrix}
			i(q-1)&		&		&		&		&		&		&		&		e_q&		&		&		\\
			\vdots&		&		&		&		&		&		&	\iddots	&		&		&		&		\\
			0&		&		&		&		&		&		e_{j+4}&		&		&		&		&		\\
			0&		&		&		&		&		e_{j+3}&		\cdots&		\cdots&		\cdots&			\cdots&		b_{k-2,1}^{q}e_{j+2}&		b_{k-1,1}^{q}e_{j+2}\\
			0&		&		&		&		e_{j+2}&		\cdots&		\cdots&		\cdots&		\cdots&		\cdots&		b_{k-2,0}^{q}e_{j+1}&		b_{k-1,0}^{q}e_{j+1}\\
			0&		&		&		e_{j+1}&		\vdots&		&		&		&		&		&		&		e_j\\
			0&		&		e_j&		\vdots&		\vdots&		&		&		&		&		&		e_{j-1}&		\\
			\vdots&		\iddots	&		&		\vdots&		\vdots&		&		&		&		&			\iddots	&		&		\\
			e_1&		&		&		\vdots&		\vdots&		&		&		&		i(q-1)&		&		&		\\
			\vdots&		&		&		\vdots&		\vdots&		&		&		\iddots&		&		&		&		&		\\
			0&		&		&		b_{k-2,1}e_{j+2}&		b_{k-2,0}e_{j+3}&		&		e_{j+5}&			&		&		&		B_1&		B_2\\
			0&		&		&		b_{k-1,1}e_{j+2}&		b_{k-1,0}e_{j+3}&		e_{j+4}&		&		&		&		&		B_3&		B_4\\
		\end{matrix} \right) , 
	\end{equation}
	where
	$
	B_1=b_{k-2,0}^{q+1}e_{j+2}$, $B_2=b_{k-2,1}e_{j+1}+b_{k-2,0}b_{k-1,0}^{q}e_{j+2}$, 
	$B_3=b_{k-2,0}^{q}b_{k-1,0}e_{j+2}+b_{k-2,1}^{q}e_{j+3}$, and  $$B_4=b_{k-1,1}e_{j+1}+b_{k-1,0}^{q+1}e_{j+2}+b_{k-1,1}^{q}e_{j+3}.$$ 
\end{proposition}

\textbf{Proof.} 
By Remark \ref{G k}, we have  
$$\boldsymbol{G}=\left(\boldsymbol{G}_1:\boldsymbol{G}_{\gamma}:\cdots:\boldsymbol{G}_{\gamma^{i-1}}\right),$$
where 
\begin{equation}
	\label{M3.1} 
	\boldsymbol{G}_{\beta}=\left( \begin{matrix}
		1&				\cdots&		1\\
		\alpha _1\beta&			\cdots&		\alpha _{q-1}\beta\\
		\vdots&			&		\vdots\\
		\left( \alpha _1\beta \right) ^{k-3}&				\cdots&		\left( \alpha _{q-1}\beta \right) ^{k-3}\\
		\left( \alpha _1\beta \right) ^{k-2}+\sum\limits_{\ell=0}^1{b_{k-2,\ell}\left( \alpha _1\beta \right) ^{k+\ell}}&		\cdots&		\left( \alpha _{q-1}\beta \right) ^{k-2}+\sum\limits_{\ell=0}^1{b_{k-2,\ell}\left( \alpha _{q-1}\beta \right) ^{k+\ell}}\\
		\left( \alpha _1\beta \right) ^{k-1}+\sum\limits_{\ell=0}^1{b_{k-1,\ell}\left( \alpha _1\beta \right) ^{k+\ell}}&		\cdots&		\left( \alpha _{q-1}\beta \right) ^{k-1}+\sum\limits_{\ell=0}^1{b_{k-1,\ell}\left( \alpha _{q-1}\beta \right) ^{k+\ell}}\\
	\end{matrix} \right) 
\end{equation}
with $\beta \in \left\{1,\gamma,\gamma^2,\ldots ,\gamma^{i-1}\right\}$. Hence, 
\begin{equation}
	\label{M3.2}
	\boldsymbol{G}\boldsymbol{G}^\dagger=\left(\boldsymbol{G}_1:\boldsymbol{G}_{\gamma}:\cdots:\boldsymbol{G}_{\gamma^{i-1}}\right)\left( \begin{array}{c}
		\boldsymbol{G}_1^\dagger\\
		\boldsymbol{G}_{\gamma}^\dagger\\
		\vdots\\
		\boldsymbol{G}_{\gamma ^{i-1}}^\dagger\\
	\end{array} \right) =\sum\limits_{t=0}^{i-1}\boldsymbol{G}_{\gamma^t}\boldsymbol{G}_{\gamma^t}^\dagger.
\end{equation}

Now for any $\beta \in \left\{1,\gamma,\gamma^2,\ldots ,\gamma^{i-1}\right\}$, by directly calculating, we have 
\begin{equation}
	\label{M3.3}
	\boldsymbol{G}_{\beta}\boldsymbol{G}_{\beta}^\dagger=\left(a_{uv}\right)_{k\times k},
\end{equation}
where 
\begin{equation}
	\label{a uv}
	a_{uv}=\begin{array}{c}
		\begin{cases}
			\sum\limits_{s=1}^{q-1}\left(\alpha_s\beta\right)^{(u-1)+q(v-1)}, & \text{if } u,v\in \left\{1, 2, \ldots , k-2 \right\};\\
			\sum\limits_{s=1}^{q-1}\left[\left( \alpha _s\beta \right) ^{k-2+\left( v-1 \right) q}+b_{k-2,0}\left( \alpha _s\beta \right) ^{k+\left( v-1 \right) q}+b_{k-2,1}\left( \alpha _s\beta \right) ^{k+1+\left( v-1 \right) q}\right],& \text{if } u=k-1,v\in \left\{1, 2, \ldots , k-2 \right\};\\
			\sum\limits_{s=1}^{q-1}\left[\left( \alpha _s\beta \right) ^{k-1+\left( v-1 \right) q}+b_{k-1,0}\left( \alpha _s\beta \right) ^{k+\left( v-1 \right) q}+b_{k-1,1}\left( \alpha _s\beta \right) ^{k+1+\left( v-1 \right) q}\right], & \text{if } u=k, v\in \left\{1, 2, \ldots , k-2 \right\};\\
			\sum\limits_{s=1}^{q-1}\left[\left( \alpha _s\beta \right) ^{q\left( k-2 \right) +u-1}+b_{k-2,0}^{q}\left( \alpha _s\beta \right) ^{qk+u-1}+b_{k-2,1}^{q}\left( \alpha _s\beta \right) ^{q\left( k+1 \right) +u-1}\right], & \text{if } v=k-1,u\in \left\{1, 2, \ldots , k-2 \right\};\\
			\sum\limits_{s=1}^{q-1}\left[\left( \alpha _s\beta \right) ^{q\left( k-1 \right) +u-1}+b_{k-1,0}^{q}\left( \alpha _s\beta \right) ^{qk+u-1}+b_{k-1,1}^{q}\left( \alpha _s\beta \right) ^{q\left( k+1 \right) +u-1}\right], & \text{if } v=k-1,u\in \left\{1, 2, \ldots , k-2 \right\};\\
			\sum\limits_{s=1}^{q-1}A_{11}A_{12},&\text{if } u=v=k-1;\\
			\sum\limits_{s=1}^{q-1}A_{11}A_{22},&\text{if } u=k-1,v=k;\\ \sum\limits_{s=1}^{q-1}A_{21}A_{12},&\text{if } u=k,v=k-1;\\ \sum\limits_{s=1}^{q-1}A_{21}A_{22},&\text{if } u=v=k, 
		\end{cases}
	\end{array}
\end{equation}
with $A_{11}=\left( \alpha _s\beta \right) ^{k-2}+b_{k-2,0}\left( \alpha _s\beta \right) ^k+b_{k-2,1}\left( \alpha _s\beta \right) ^{k+1}$, $A_{12}=\left( \alpha _s\beta \right) ^{q\left( k-2 \right)}+b_{k-2,0}^{q}\left( \alpha _s\beta \right) ^{qk}+b_{k-2,1}^{q}\left( \alpha _s\beta \right) ^{q\left( k+1 \right)}$, \\ $A_{21}=\left( \alpha _s\beta \right) ^{k-1}+b_{k-1,0}\left( \alpha _s\beta \right) ^k+b_{k-1,1}\left( \alpha _s\beta \right) ^{k+1}$ and $A_{22}=\left( \alpha _s\beta \right) ^{q\left( k-1 \right)}+b_{k-1,0}^{q}\left( \alpha _s\beta \right) ^{qk}+b_{k-1,1}^{q}\left( \alpha _s\beta \right) ^{q\left( k+1 \right)}.$ 


Next, according to the ranges of $u$ and $v$, we divide the following $6$ cases to determine the values of $a_{uv}$. By combining with Remark \ref{u+v-2} and the range of $u+v-2$, we deduce that only the following $a_{uv}$ are non-zero, and all other entries are zero. The detailed proofs are given in  the Appendix \ref{appendex B}.

\textbf{case 1}  For $u,v\in \{1,2,\dots,k-2\}$, we have $a_{11}=a_{qq}=q-1$ and 
$$
\begin{aligned}
	&a_{u(q-u+1)}=(q-1)\beta^{(q-1)(q-u+1)}=c_{q-u+1},\quad(1\le u \le q);\\
	&a_{u(2q-u)}=(q-1)\beta^{(q-1)(q-u)}=c_{q-u},\quad(j+5\le u\le q-1);\\
	&a_{u(2q-u)}=(q-1)\beta^{(q-1)(2q-u+1)}=c_{2q-u+1},\quad(q+1\le u \le k-2).\\
\end{aligned}
$$	

\textbf{case 2} For $u=k-1, v\in \left\{1, 2, \ldots , k-2 \right\}$, we have $ a_{( k-1 ) (j+1)}= b_{k-2,1}(q-1)\beta^{(q-1)(j+2)}=b_{k-2,1}c_{j+2}$ and 
$$
	 a_{( k-1 )(j+2)}= b_{k-2,0}(q-1)\beta^{(q-1)(j+3)}=b_{k-2,0}c_{j+3}, 
	 a_{\left( k-1 \right) (j+4)}= (q-1)\beta^{(q-1)(j+5)}=c_{j+5}. 
$$

\textbf{case 3} For $u=k, v\in \left\{1, 2, \ldots , k-2 \right\}$, we have $a_{k (j+1)}=b_{k-1,1}(q-1)\beta^{(q-1)(j+2)}=b_{k-1,1}c_{j+2}$ and  

$$
	 a_{k (j+2)}=b_{k-1,0}(q-1)\beta^{(q-1)(j+3)}=b_{k-1,0}c_{j+3}, 
	 a_{k (j+3)}=(q-1)\beta^{(q-1)(j+4)}=c_{j+4}.
$$

\textbf{case 4} For $u\in {1, 2, \ldots , k-2}, v=k-1$, we have $a_{(j+1)\left( k-1 \right) }= b_{k-2,1}^q(q-1)\beta^{(q-1)(j+2)}=b_{k-2,1}^qc_{j+2}$ and  

$$
a_{(j+2)\left( k-1 \right) }= b_{k-2,0}^q(q-1)\beta^{(q-1)(j+1)}=b_{k-2,0}^qc_{j+1},  a_{(j+4)\left( k-1 \right) }= (q-1)\beta^{(q-1)(j-1)}=c_{j-1}. 
$$

\textbf{case 5} For $u\in \left\{1, 2, \ldots , k-2 \right\}, v=k$, we have $a_{(j+1) k}= b_{k-1,1}^q(q-1)\beta^{(q-1)(j+2)}=b_{k-1,1}^qc_{j+2}$ and 

$$ a_{(j+2) k}= b_{k-1,0}^q(q-1)\beta^{(q-1)(j+1)}=b_{k-1,0}^qc_{j+1}, 
 a_{(j+3)k}=(q-1)\beta^{(q-1)j}=c_{j}. 
$$

\textbf{case 6} For  $u,v \in \{k-1,k\}$, we have $a_{(k-1)(k-1)}=\left( q-1 \right) b_{k-2,0}^{q+1}\beta ^{(q-1)(j+2)}=b_{k-2,0}^{q+1}c_{j+2}$ and 
$$
\begin{aligned}
	&a_{(k-1)k}=\left( q-1 \right) \left[ b_{k-2,0}b_{k-1,0}^{q}\beta ^{(q-1)(j+2)}+b_{k-2,1}\beta ^{(q-1)(j+1)} \right]= b_{k-2,0}b_{k-1,0}^{q}c_{j+2}+b_{k-2,1}c_{j+1};\\
	&a_{k(k-1)}=\left( q-1 \right) \left[ b_{k-2,1}^{q}\beta ^{(q-1)(j+3)}+b_{k-2,0}^{q}b_{k-1,0}\beta ^{(q-1)(j+2)} \right]=b_{k-2,1}^{q}c_{j+3}+b_{k-2,0}^{q}b_{k-1,0}c_{j+2};\\
	&a_{kk}=\left( q-1 \right) \left[ b_{k-1,1}^{q}\beta ^{(q-1)(j+3)}+b_{k-1,0}^{q+1}\beta ^{(q-1)(j+2)}+b_{k-1,1}\beta ^{(q-1)(j+1)} \right]=b_{k-1,1}^{q}c_{j+3}+b_{k-1,0}^{q+1}c_{j+2}+b_{k-1,1}c_{j+1}.
\end{aligned}
$$

Thus 	
\begin{equation}
	\label{GG beta}
	\boldsymbol{G}_{\beta}\boldsymbol{G}_{\beta}^\dagger
	=\left( \begin{matrix}
		q-1&		&		&		&		&		&		&		&		&		&		c_q&		&		&		&		\\
		0&		&		&		&		&		&		&		&		&		c_{q-1}&		&		&		&		&		\\
		\vdots&		&		&		&		&		&		&		&	\iddots	&		&		&		&		&		&		\\
		0&		&		&		&		&		&		&		c_{j+4}&		&		&		&		&		&		&		\\
		0&		&		&		&		&		&		c_{j+3}&		\cdot&		\cdots&		\cdot&		\cdot&		\cdot&		\cdots&		D_5&		D_6\\
		0&		&		&		&		&		c_{j+2}&		\cdot&		\cdot&		\cdots&		\cdot&		\cdot&		\cdot&		\cdots&		D_7&	D_8\\
		0&		&		&		&		c_{j+1}&		\cdot&		\cdot&		\cdot&		\cdots&		\cdot&		\cdot&		\cdot&		\cdots&		0&		c_j\\
		0&		&		&		c_j&		\cdot&		\cdot&		\cdot&		\cdot&		\cdots&		\cdot&		\cdot&		\cdot&		\cdots&		c_{j-1}&		\cdot\\
		\vdots&		&	\iddots	&		&		\vdots&		\vdots&		\vdots&		\vdots&		&		&		&		&	\iddots	&		\vdots&		\vdots\\
		0&		c_2&		\cdots&		\cdot&		\cdot&		\cdot&		\cdot&		\cdot&		\cdots&		\cdot&		\cdot&		c_1&		\cdots&		\cdot&		\cdot\\
		c_1&		\cdot&		\cdots&		\cdot&		\cdot&		\cdot&		\cdot&		\cdot&		\cdots&		\cdot&		q-1&		\cdot&		\cdots&		\cdot&		\cdot\\
		0&		\cdot&		\cdots&		\cdot&		\cdot&		\cdot&		\cdot&		\cdot&		\cdots&		c_q&		\cdot&		\cdot&		\cdots&		\cdot&		\cdot\\
		\vdots&		&		&		&		\vdots&		\vdots&		\vdots&		\vdots&	\iddots	&		&		&		&		&		\vdots&		\vdots\\
		0&		\cdot&		\cdots&		\cdot&		D_1&		D_2&		0&		c_{j+5}&		\cdots&		\cdot&		\cdot&		\cdot&		\cdots&		A_1&		A_2\\
		0&		\cdot&		\cdots&		\cdot&		D_3&		D_4&		c_{j+4}&		\cdot&		\cdots&		\cdot&		\cdot&		\cdot&		\cdots&		A_3&		A_4\\
	\end{matrix} \right),
\end{equation}
where $D_1=b_{k-2,1}c_{j+2}$, 
$D_2=b_{k-2,0}c_{j+3}$, 
$D_3=b_{k-1,1}c_{j+2}$, 
$D_4=b_{k-1,0}c_{j+3}$, 
$D_5=b_{k-2,1}^{q}c_{j+2}$, 
$D_6=b_{k-1,1}^{q}c_{j+2}$, \\
$D_7=b_{k-2,0}^{q}c_{j+1}$,  
$D_8=b_{k-1,0}^{q}c_{j+1}$, $A_1=b_{k-2,0}^{q+1}c_{j+2}$,  $A_2=b_{k-2,0}b_{k-1,0}^{q}c_{j+2}+b_{k-2,1}c_{j+1}$, 
$A_3=b_{k-2,1}^{q}c_{j+3}+b_{k-2,0}^{q}b_{k-1,0}c_{j+2}$ and $A_4=b_{k-1,1}^{q}c_{j+3}+b_{k-1,0}^{q+1}c_{j+2}+b_{k-1,1}c_{j+1}$. 

Furthermore, by substituting the Equation \eqref{GG beta} into the Equation \eqref{M3.2} and directly calculating, Proposition \ref{P3.1} is immediately. \\
\quad

From the following Proposition \ref{P3.2},  to determine  $\operatorname{rank}(\boldsymbol{GG}^\dagger)$ for $\gcd(i,q+1)=1$, it's enough to compute  $\operatorname{rank}\left(\boldsymbol{A}_{2\times 2}\right)$. 
\begin{proposition}
	\label{P3.2}
	Let $\boldsymbol{G}$ be the generator matrix of $\mathcal{C}_{q+j}(\boldsymbol{a})$. If $\gcd(i,q+1)=1$, then $\mathrm{rank}\left(\boldsymbol{GG}^\dagger\right)=q+\mathrm{rank}\left(\boldsymbol{A}_{2\times 2}\right)$, 
	where 	$$
	\boldsymbol{A}_{2\times 2}=\left(\begin{matrix}
		C_1 & C_2\\
		C_3 & C_4\\
	\end{matrix}\right),
	$$
	with 
	$$
	\begin{aligned}
		&C_1=b_{k-2,0}^{q+1}\left(e_{j+2}-\frac{e_{j+3}}{e_{j+2}}e_{j+1}\right),\\ &C_2=b_{k-2,1}\left(e_{j+1}-\frac{e_{j+2}}{e_{j+1}}e_{j}\right) +b_{k-2,0}b_{k-1,0}^{q}\left(e_{j+2}-\frac{e_{j+3}}{e_{j+2}}e_{j+1}\right), \\
		&C_3=b_{k-2,0}^{q}b_{k-1,0}\left(e_{j+2}-\frac{e_{j+3}}{e_{j+2}}e_{j+1}\right)+b_{k-2,1}^{q}\left(e_{j+3}-\frac{e_{j+4}}{e_{j+3}}e_{j+2}\right)\\
	\end{aligned}
	$$
	and
	$$C_4=b_{k-1,1}\left(e_{j+1}-\frac{e_{j+2}}{e_{j+1}}e_{j}\right)+b_{k-1,0}^{q+1}\left(e_{j+2}-\frac{e_{j+3}}{e_{j+2}}e_{j+1}\right)+b_{k-1,1}^{q}\left(e_{j+3}-\frac{e_{j+4}}{e_{j+3}}e_{j+2}\right).$$
\end{proposition}
\textbf{Proof.} By $\gcd (i, q+1)=1$ and Corollary \ref{C2.1}, we have $e_r \neq 0$ for  $1 \le r \le q$. For convenience, we denote  $\boldsymbol{a}_{u,1}$ be the $u$-th row of the matrix $\boldsymbol{GG}^{\dagger}$ given by  \eqref{GG H}, where $1 \leq u\leq k$. Now, for the matrix $\boldsymbol{GG}^{\dagger}$ given by \eqref{GG H}, we perform the following $4$ steps of elementary transformations.

\textbf{Step 1.} Replace $\boldsymbol{a}_{u+q-1,1}$ with $\boldsymbol{a}_{u+q-1,1}-\frac{e_{q+2-u}}{e_{q+1-u}}\boldsymbol{a}_{u,1}$ for each $u=2,3, \ldots , j+1$, we obtain the
matrix $\boldsymbol{A}^{(1)}$ and denote the vector $\boldsymbol{a}_{u,1}^{(1)}$ be the $u$-th row of $\boldsymbol{A}^{(1)}$, where $1 \leq u\leq k$;

\textbf{Step 2.} Replace $\boldsymbol{a}_{k-1,1}^{(1)}$ with $\boldsymbol{a}_{k-1,1}^{(1)}-\frac{b_{k-2,0}e_{j+3}}{e_{j+2}}\boldsymbol{a}_{j+2,1}^{(1)}-\frac{b_{k-2,1}e_{j+2}}{e_{j+1}}\boldsymbol{a}_{j+3,1}^{(1)}$, and replace $\boldsymbol{a}_{k,1}^{(1)}$ with $\boldsymbol{a}_{k,1}^{(1)}-\frac{b_{k-1,0}e_{j+3}}{e_{j+2}}\boldsymbol{a}_{j+2,1}^{(1)}-\frac{b_{k-1,1}e_{j+2}}{e_{j+1}}\boldsymbol{a}_{j+3,1}^{(1)}$, we obtain the
matrix $\boldsymbol{A}^{(2)}$ and denote the vector $\boldsymbol{b}_{1,v}^{(2)}$ be the $v$-th column of $\boldsymbol{A}^{(2)}$, where $1 \leq v\leq k$;

\textbf{Step 3.} Replace $\boldsymbol{b}_{1,v+q-1}^{(2)}$ with $\boldsymbol{b}_{1,v+q-1}^{(2)}-\frac{e_{v-1}}{e_{v}}\boldsymbol{b}_{1,v}^{(2)}$ for each $v=2,3, \ldots , j+1$, we obtain the matrix $\boldsymbol{A}^{(3)}$ and denote the vector $\boldsymbol{b}_{1,v}^{(3)}$ be the $v$-th column of $\boldsymbol{A}^{(3)}$, where $1 \leq v\leq k$;

\textbf{Step 4.} Replace $\boldsymbol{b}_{1,k-1}^{(3)}$ with $\boldsymbol{b}_{1,k-1}^{(3)}-\frac{b_{k-2,0}^qe_{j+1}}{e_{j+2}}\boldsymbol{b}_{1,j+2}^{(3)}-\frac{b_{k-2,1}^qe_{j+2}}{e_{j+3}}\boldsymbol{b}_{1,j+3}^{(3)}$, and replace $\boldsymbol{b}_{1,k}^{(3)}$ with $\boldsymbol{b}_{1,k}^{(3)}-\frac{b_{k-1,0}^qe_{j+1}}{e_{j+2}}\boldsymbol{b}_{1,j+2}^{(3)}-\frac{b_{k-1,1}^qe_{j+2}}{e_{j+3}}\boldsymbol{b}_{1,j+3}^{(3)}$, we obtain the
matrix



$$
\boldsymbol{A}^{(4)}=\left( \begin{matrix}
	i(q-1)&		&		&		&		&		&		&		&		e_q&		&		&		\\
	\vdots&		&		&		&		&		&		&	\iddots	&		&		&		&		\\
	0&		&		&		&		&		&		e_{j+4}&		&		&		&		&		\\
	0&		&		&		&		&		e_{j+3}&		\cdots&		\cdots&		\cdots&			\cdots&		0&		0\\
	0&		&		&		&		e_{j+2}&		\cdots&		\cdots&		\cdots&		\cdots&		\cdots&		0&		0\\
	0&		&		&		e_{j+1}&		\vdots&		&		&		&		&		&		&		0\\
	0&		&		e_j&		\vdots&		\vdots&		&		&		&		&		&		0&		\\
	\vdots&		\iddots	&		&		\vdots&		\vdots&		&		&		&		&			\iddots	&		&		\\
	e_1&		&		&		\vdots&		\vdots&		&		&		&		i(q-1)&		&		&		\\
	\vdots&		&		&		\vdots&		\vdots&		&		&		\iddots&		&		&		&		&		\\
	0&		&		&	0&		0&		&		0&			&		&		&		C_1 &		C_2 \\
	0&		&		&		0&		0&		0&		&		&		&		&		C_3 &		C_4 \\
	
\end{matrix} \right) 
$$
with
$$
\begin{aligned}
	&C_1=b_{k-2,0}^{q+1}\left(e_{j+2}-\frac{e_{j+3}}{e_{j+2}}e_{j+1}\right),\\ &C_2=b_{k-2,1}\left(e_{j+1}-\frac{e_{j+2}}{e_{j+1}}e_{j}\right) +b_{k-2,0}b_{k-1,0}^{q}\left(e_{j+2}-\frac{e_{j+3}}{e_{j+2}}e_{j+1}\right), \\
	&C_3=b_{k-2,0}^{q}b_{k-1,0}\left(e_{j+2}-\frac{e_{j+3}}{e_{j+2}}e_{j+1}\right)+b_{k-2,1}^{q}\left(e_{j+3}-\frac{e_{j+4}}{e_{j+3}}e_{j+2}\right)\\
\end{aligned}
$$
and
$$C_4=b_{k-1,1}\left(e_{j+1}-\frac{e_{j+2}}{e_{j+1}}e_{j}\right)+b_{k-1,0}^{q+1}\left(e_{j+2}-\frac{e_{j+3}}{e_{j+2}}e_{j+1}\right)+b_{k-1,1}^{q}\left(e_{j+3}-\frac{e_{j+4}}{e_{j+3}}e_{j+2}\right).$$
Now set 
$$
\boldsymbol{A}_{2\times 2}=\left(\begin{matrix}
	C_1 & C_2\\
	C_3 & C_4\\
\end{matrix}\right),
$$
and by combining with Lemma \ref{L2.9 },  we can get   $$\mathrm{rank}(\boldsymbol{G}\boldsymbol{G}^\dagger)=\mathrm{rank}(\boldsymbol{A}^{(4)})=q+\mathrm{rank}\left(\boldsymbol{A}_{2\times 2}\right).$$

From the following Proposition \ref{P3.3},  to determine  $\operatorname{rank}(\boldsymbol{GG}^\dagger)$ for $\gcd(i,q+1)=h>1$, it's enough to compute  $\operatorname{rank}\left(\boldsymbol{A}_{7\times 7}\right)$.

\begin{proposition}
	\label{P3.3}
	Let $\boldsymbol{G}$ be the generator matrix of $\mathcal{C}_{q+j}(\boldsymbol{a})$, if $\gcd(i,q+1)=h>1$, then $\mathrm{rank}\left( \boldsymbol{GG}^\dagger \right)=q-5+\mathrm{rank}\left(\boldsymbol{A}_{7\times 7}\right)$, 
	where 
	\begin{equation}
		\label{M3.20}
		\boldsymbol{A}_{7\times 7}=\left(\begin{matrix}
			&		&		&		&		e_{j+4}&		0	&		0	\\
			&		&		&		e_{j+3}&		&	b_{k-2,1}^{q}e_{j+2}&		b_{k-1,1}^{q}e_{j+2}	\\
			&		&		e_{j+2}&		&		&			b_{k-2,0}^{q}e_{j+1}&		b_{k-1,0}^{q}e_{j+1}	\\
			&		e_{j+1}&		&		&		&	0	&		e_j	\\
			e_j&		&		&		&		&	e_{j-1}	&	0	\\
			0	&	b_{k-2,1}e_{j+2}&		b_{k-2,0}e_{j+3}&	0	&	e_{j+5}	&	B_1	&	B_2	\\
			0	&		b_{k-1,1}e_{j+2}&		b_{k-1,0}e_{j+3}&		e_{j+4}&0	&	B_3	&	B_4	\\
		\end{matrix}\right)\ .
	\end{equation}
\end{proposition}
\textbf{Proof.} Firstly, we denote the matrix  $\boldsymbol{GG}^\dagger$ given in the Equation \eqref{GG H} by $\left(a_{uv}^{1}\right)_{k\times k}$, then by Lemma \ref{L2.4 }, we know that $e_r$ and $e_{r+1}$ are not zeros simultaneously, where $1\le r\le q-1$. Furthermore,  $a_{u(q+1-u)}^{1}$ and $a_{(q-1+u)(q+1-u)}^{1}$ are not zeros simultaneously, $a_{(q+1-v)v}^{1}$ and $a_{(q+1-v)(q-1+v)}^{1}$ are also not zeros simultaneously, where $2\le u, v\le j-1$. 

Secondly, for convenience, we denote the vectors $\boldsymbol{a}_{u,2}$ and $\boldsymbol{b}_{2,v}$ be the $u$-th row and the $v$-th column of the matrix  $\boldsymbol{GG}^\dagger$ given by the Equation \eqref{GG H}, respectively, where $1 \le u,v\le k$. Then for the matrix  $\boldsymbol{GG}^\dagger$ given by  \eqref{GG H},  we perform the following $2$ steps of elementary transformations.

\textbf{Step 1.} For each $u=2,3,\ldots,j-1$, if $a_{u(q+1-u)}^{1} = 0$, then swap $\boldsymbol{a}_{u,2}$ and $\boldsymbol{a}_{q-1+u,2}$; if $a_{u(q+1-u)}^{1} \neq 0$, then replace $\boldsymbol{a}_{q-1+u,2}$ with $\boldsymbol{a}_{q-1+u,2}-\frac{e_{q+2-u}}{e_{q+1-u}}\boldsymbol{a}_{u,2}$; 

\textbf{Step 2.} For each $v=2,3,\ldots,j-1$, if $a_{(q+1-v)v}^{1} = 0$, then swap $\boldsymbol{b}_{2,v}$ and $\boldsymbol{b}_{2,q-1+v}$; if $a_{(q+1-v)v}^{1} \neq 0$, then replace $\boldsymbol{b}_{2,q-1+v}$ with $\boldsymbol{b}_{2,q-1+v}-\frac{e_{v-1}}{e_{v}}\boldsymbol{b}_{2,v}$, and so, we obtain the matrix $\boldsymbol{A}_{1}^{(1)}$, and  denote $\boldsymbol{A}_{1}^{(1)}=\left(a_{uv}^{(1)}\right)_{k\times k}$ and $e_{u_{1}}^{1}=a_{q+1-u_{1},u_{1}}^{(1)}\in \mathbb{F}_{q^2}^{*}$, where $u_{1}\in \left\{2,3,\ldots,j-1,j+5,\ldots,q-1\right\}$. Thus, 
\begin{equation}
	\begin{aligned}
		&\boldsymbol{A}_{1}^{(1)}\\
		&=\left( \begin{matrix}
			i(q-1)& &  &	 &  &	&		&		&		&		&		&		&		e_q&	&	&	&	&		\\
			&     &  &	&	&	&		&		&		&		&		&	e_{q-1}^{1}	&	&	&	&	&	&	\\
			&	 & &   &	&	&		&		&		&		&	\iddots	&		&	&	&   	&   &	&		\\
			&	& 	&   &	&	&		&		&		&	e_{j+5}^{1}&	&		&		&	&	&	&    &	\\
			&	&	&   &	&	&		&		&	e_{j+4}	&		&		&		&		&	&	&       &	&	\\
			&	&	&   &	&	&		&e_{j+3}&	\cdots &	\cdots &\cdots &\cdots &\cdots & \cdots &\cdots &\cdots & D_5^{1}& D_6^{1}\\
			&	&	&   &	&	&e_{j+2}&\cdots &\cdots	 &\cdots	 &\cdots	 &	\cdots &	\cdots &\cdots  &\cdots &\cdots  &D_7^{1}&	D_8^{1}\\
			&	&	&   &	&e_{j+1}& \vdots&	&		&		&		&		&		&	& & &  &	e_j\\
			&	&	&   &e_j& \vdots& \vdots&		&		&		&		&		&	   & &  &  &	e_{j-1}&	\\
			&  &	&e_{j-1}^{1}&   &\vdots&\vdots&		&		&		&		&		&	 &  &  &  0&	&	\\
			&	& \iddots	&	&	&\vdots&\vdots&		&	    &   	&	    &    	&	& &  \iddots&   &  & &		\\
			&e_2^{1}&	&		& &\vdots &	\vdots	&		&		&       &	    &   	&	 & 0     &   &   &	&	\\
			e_1& &	&& &	 \vdots &	\vdots	&		&	&		&		&       &   i(q-1) &      &   &   &	&	\\
			&	&	&	& &\vdots&	\vdots	&	&  & 		&		&	0	&	 &     &       & &	&	\\
			&	&	&	& &\vdots&\vdots	&	&  & 	&	\iddots	&		&	 &     &       & &	&	\\
			&	&	&	& &\vdots&\vdots	&		&  & 0		&		&		&	 &     &       &	& &	\\
			&	&	&	& &D_1^{1}&	D_2^{1}&	&e_{j+5} &	&	&	&   & &   &  &	B_1&		B_2\\
			&	&	&	& &D_3^{1}&	D_4^{1} &e_{j+4}	&	&	&	&	& &	& &      &	B_3&		B_4\\
		\end{matrix} \right),
	\end{aligned} 
\end{equation}
where $D_1^{1}=b_{k-2,1}e_{j+2}$, 
$D_2^{1}=b_{k-2,0}e_{j+3}$, 
$D_3^{1}=b_{k-1,1}e_{j+2}$, 
$D_4^{1}=b_{k-1,0}e_{j+3}$, 
$D_5^{1}=b_{k-2,1}^{q}e_{j+2}$, 
$D_6^{1}=b_{k-1,1}^{q}e_{j+2}$, 
$D_7^{1}=b_{k-2,0}^{q}e_{j+1}$, 
$D_8^{1}=b_{k-1,0}^{q}e_{j+1}$, 
$B_1=b_{k-2,0}^{q+1}e_{j+2}$, $B_2=b_{k-2,1}e_{j+1}+b_{k-2,0}b_{k-1,0}^{q}e_{j+2}$, 
$B_3=b_{k-2,0}^{q}b_{k-1,0}e_{j+2}+b_{k-2,1}^{q}e_{j+3}$ and  $B_4=b_{k-1,1}e_{j+1}+b_{k-1,0}^{q+1}e_{j+2}+b_{k-1,1}^{q}e_{j+3}$.

Next, let $$
\boldsymbol{g}_i =\left(0,\dots,0,\underset{\text{the $i$-th component}}{1},0,\dots,0\right)^\mathrm{T}\in\mathbb{F}_{q^2}^{k},
$$ 

$$
\boldsymbol{P}=\big(\boldsymbol{g}_{1}\ \ldots\ \boldsymbol{g}_{j-1}\quad\boldsymbol{g}_{k-6}\quad
\boldsymbol{g}_{k-5}\quad\boldsymbol{g}_{k-4}\quad
\boldsymbol{g}_{k-3}\quad\boldsymbol{g}_{k-2}\quad
\boldsymbol{g}_{j}\ \ldots\ \boldsymbol{g}_{k-7}\quad \boldsymbol{g}_{k-1}\quad\boldsymbol{g}_{k}\big),  
$$
and
$$
\boldsymbol{Q}=\big(\boldsymbol{g}_{1}\ \ldots\ \boldsymbol{g}_{j-1}\quad\boldsymbol{g}_{j+5}\ \ldots\ \boldsymbol{g}_{k-2}\quad\boldsymbol{g}_{j}\quad \boldsymbol{g}_{j+1}\quad\boldsymbol{g}_{j+2}\quad \boldsymbol{g}_{j+3}\quad\boldsymbol{g}_{j+4}\quad \boldsymbol{g}_{k-1}\quad\boldsymbol{g}_{k}\big). 
$$

It's easy to know that the matrices $\boldsymbol{P}$ and $\boldsymbol{Q}$ are both non-singular, and then by directly calculating, we have

$$
\boldsymbol{P}\boldsymbol{A}_{1}^{(1)}\boldsymbol{Q}=\begin{pmatrix}
	\boldsymbol{B}_{q-5\times q-5}
	&\boldsymbol{0}_{q-5\times j-2} &\boldsymbol{0}_{q-5\times 7}\\
	\boldsymbol{0}_{j-2\times q-5} &\boldsymbol{0}_{j-2\times j-2} &\boldsymbol{0}_{j-2\times 7}\\
	\boldsymbol{0}_{7\times q-5} &\boldsymbol{0}_{7\times j-2} & \boldsymbol{A}_{7\times 7}
\end{pmatrix},
$$
where $$\boldsymbol{B}_{q-5\times q-5}=\begin{pmatrix}
	i(q-1) & & & & &  & & e_q\\
	& & & & & & e_{q-1}^{1}& \\
	& & & & & \iddots& & \\
	& & & & e_{j+5}^{1}& & & \\
	& & & e_{j-1}^{1}& & & & \\
	& & \iddots& & & & & \\
	& e_{2}^{1}& & & & & & \\
	e_{1}& & & & & & & i(q-1)\\
\end{pmatrix}$$
and
$$\boldsymbol{A}_{7\times 7}=\begin{pmatrix}
	0	&	0	&	0	&	0	&		e_{j+4}&			0&		0	\\
	0	&	0	&	0	&		e_{j+3}&		0&	b_{k-2,1}^{q}e_{j+2}&		b_{k-1,1}^{q}e_{j+2}	\\
	0	&	0	&		e_{j+2}&		0&		0&			b_{k-2,0}^{q}e_{j+1}&		b_{k-1,0}^{q}e_{j+1}	\\
	0	&		e_{j+1}&		0&		0&		0&	0	&		e_j	\\
	e_j&		0&		0&		0&		0&	e_{j-1}	&	0	\\
	0	&	b_{k-2,1}e_{j+2}&		b_{k-2,0}e_{j+3}&	0	&	e_{j+5}	&	B_1	&	B_2	\\
	0	&		b_{k-1,1}e_{j+2}&		b_{k-1,0}e_{j+3}&		e_{j+4}&0	&	B_3	&	B_4	\\
\end{pmatrix}$$
with 
$
B_1=b_{k-2,0}^{q+1}e_{j+2}$, $ B_2=b_{k-2,1}e_{j+1}+b_{k-2,0}b_{k-1,0}^{q}e_{j+2}
$, 
$
B_3=b_{k-2,0}^{q}b_{k-1,0}e_{j+2}+b_{k-2,1}^{q}e_{j+3}$ and \\ $B_4=b_{k-1,1}e_{j+1}+b_{k-1,0}^{q+1}e_{j+2}+b_{k-1,1}^{q}e_{j+3}
$. 

By combining with Lemma \ref{L2.9 }, we have  $\mathrm{rank}\left(\boldsymbol{B}_{q-5\times q-5}\right)=q-5$, furthermore,  $$\mathrm{rank}\left(\boldsymbol{GG}^\dagger\right)=\mathrm{rank}\left(\boldsymbol{P}\boldsymbol{A}_{1}^{(1)}\boldsymbol{Q}\right)=q-5+\mathrm{rank}\left(\boldsymbol{A}_{7\times 7}\right).$$

\subsection{The proofs of Theorems \ref{T3.1}-\ref{T3.3}}

By Lemma 2.5, we only need to focus on computing  $\operatorname{rank}\left(\boldsymbol{GG}^\dagger\right)$ to determine $\dim\left(\operatorname{Hull}_H\left(\mathcal{\mathcal{C}}\right)\right)$.

\textbf{The proof of Theorem \ref{T3.1}} 

By $\gcd(i,q+1)=1$ and Proposition \ref{P3.2}, we only need to compute  $\mathrm{rank}\left(\boldsymbol{A}_{2\times 2}\right)$. 
Next, depending on  $i=q$ or not, we divide the following 2 cases to determine the value of $\mathrm{rank}\left(\boldsymbol{A}_{2\times 2}\right)$.

\textbf{case 1.} If $i=q$, then by Lemma \ref{L2.10 }, we know that $e_{j+1}^2=e_{j+2}e_{j}$, $e_{j+2}^2=e_{j+3}e_{j+1}$ and $e_{j+3}^2=e_{j+4}e_{j+2}$, i.e., $$e_{j+1}-\frac{e_{j+2}}{e_{j+1}}e_{j}=e_{j+2}-\frac{e_{j+3}}{e_{j+2}}e_{j+1}=e_{j+3}-\frac{e_{j+4}}{e_{j+3}}e_{j+2}=0.$$ Furthermore, we can get  $\mathrm{rank}\left(\boldsymbol{A}_{2\times 2}\right)=0$, i.e.,  $\mathrm{rank}\left(\boldsymbol{G}\boldsymbol{G}^\dagger\right)=q$. 

\textbf{case 2.} If $i\neq q$, then by Lemma \ref{L2.10 }, we know that $e_{j+1}-\frac{e_{j+2}}{e_{j+1}}e_{j}$,  $e_{j+2}-\frac{e_{j+3}}{e_{j+2}}e_{j+1}$  and  $e_{j+3}-\frac{e_{j+4}}{e_{j+3}}e_{j+2}$ are all non-zero. 
Next, according to the number of zero entries among  $b_{k-2,0},b_{k-2,1},b_{k-1,0},b_{k-1,1}$, we divide the following $5$ cases to determine the value of  $\mathrm{rank}\left(\boldsymbol{A}_{2\times 2}\right)$. 

\textbf{case 2.1.} If $b_{k-2,0}$, $b_{k-2,1}$, $b_{k-1,0}$ and $b_{k-1,1}$ are all zero, it is easy to get $\mathrm{rank}\left(\boldsymbol{A}_{2\times2}\right)=0$, i.e., $\mathrm{rank}\left( \boldsymbol{G}\boldsymbol{G}^\dagger \right)=q$. 

\textbf{case 2.2.} If only one of $b_{k-2,0},b_{k-2,1},b_{k-1,0},b_{k-1,1}$ is non-zero, then we have the following $4$ cases. 

$(1)$ If $b_{k-2,0}\in \mathbb{F}_{q^2}^{*}$, then we have 
$$
\boldsymbol{A}_{2\times 2}=\left(\begin{matrix}
	b_{k-2,0}^{q+1}\left(e_{j+2}-\frac{e_{j+3}}{e_{j+2}}e_{j+1}\right) & 0\\
	0 & 0\\
\end{matrix}\right).
$$
It is easy to get $\mathrm{rank}\left(\boldsymbol{A}_{2\times 2}\right)=1$, i.e., $\mathrm{rank}\left( \boldsymbol{GG}^\dagger \right)=q+1$. 

$(2)$ If $b_{k-2,1}\in \mathbb{F}_{q^2}^{*}$, then we have

$$
\boldsymbol{A}_{2\times 2}=\left(\begin{matrix}
	0 & b_{k-2,1}\left(e_{j+1}-\frac{e_{j+2}}{e_{j+1}}e_{j}\right)\\
	b_{k-2,1}^{q}\left(e_{j+3}-\frac{e_{j+4}}{e_{j+3}}e_{j+2}\right)& 0\\
\end{matrix}\right). 
$$
It is easy to get $\mathrm{rank}\left(\boldsymbol{A}_{2\times 2}\right)=2$, i.e., $\mathrm{rank}\left( \boldsymbol{GG}^\dagger \right)=q+2$. 

	$(3)$ If $b_{k-1,0}\in \mathbb{F}_{q^2}^{*}$, then we have 
	
	$$
	\boldsymbol{A}_{2\times 2}=\left(\begin{matrix}
		0 & 0\\
		0 & b_{k-1,0}^{q+1}\left(e_{j+2}-\frac{e_{j+3}}{e_{j+2}}e_{j+1}\right)\\
	\end{matrix}\right). 
	$$
	It is easy to get $\mathrm{rank}\left(\boldsymbol{A}_{2\times 2}\right)=1$, i.e., $\mathrm{rank}\left( \boldsymbol{GG}^\dagger \right)=q+1$. 
	
	

	$(4)$ If $b_{k-1,1}\in \mathbb{F}_{q^2}^{*}$, then we have
	
	$$
	\boldsymbol{A}_{2\times 2}=\left(\begin{matrix}
		0 & 0\\
		0 & b_{k-1,1}\left(e_{j+1}-\frac{e_{j+2}}{e_{j+1}}e_{j}\right)+b_{k-1,1}^{q}\left(e_{j+3}-\frac{e_{j+4}}{e_{j+3}}e_{j+2}\right)\\
	\end{matrix}\right). 
	$$
	Thus,  $\mathrm{rank}\left(\boldsymbol{A}_{2\times 2}\right)=0$ if and only if $b_{k-1,1}\left(e_{j+1}-\frac{e_{j+2}}{e_{j+1}}e_{j}\right)+b_{k-1,1}^{q}\left(e_{j+3}-\frac{e_{j+4}}{e_{j+3}}e_{j+2}\right)=0$. And by Corollary \ref{C2.2}, we know that $\mathrm{rank}\left(\boldsymbol{A}_{2\times 2}\right)=0$ if and only if $b_{k-1,1}^{q-1}\gamma^{(i-1)(q-1)}=-1$. Thus 
	$$
	\mathrm{rank}\left( \boldsymbol{GG}^\dagger \right) = \begin{array}{c}
		\begin{cases}
			q, &\text{if } b_{k-1,1}^{q-1}\gamma^{(i-1)(q-1)}=-1 ;\\
			q+1, &\text{if }b_{k-1,1}^{q-1}\gamma^{(i-1)(q-1)}\neq -1.
		\end{cases}
	\end{array} 
	$$

	\textbf{case 2.3.} If two of $b_{k-2,0},b_{k-2,1},b_{k-1,0},b_{k-1,1}$ are exactly zero, then we have the following $6$ cases. 
	
	(1) If $b_{k-2,0}=b_{k-2,1}=0$, then we have 
	
	$$
	\boldsymbol{A}_{2\times 2}=\left(\begin{matrix}
		0 & 0\\
		0 & b_{k-1,1}\left(e_{j+1}-\frac{e_{j+2}}{e_{j+1}}e_{j}\right)+b_{k-1,0}^{q+1}\left(e_{j+2}-\frac{e_{j+3}}{e_{j+2}}e_{j+1}\right)+b_{k-1,1}^{q}\left(e_{j+3}-\frac{e_{j+4}}{e_{j+3}}e_{j+2}\right)\\
	\end{matrix}\right).
	$$
	Thus,  $\mathrm{rank}\left(\boldsymbol{A}_{2\times 2}\right)=0$ if and only if  $b_{k-1,1}\left(e_{j+1}-\frac{e_{j+2}}{e_{j+1}}e_{j}\right)+b_{k-1,0}^{q+1}\left(e_{j+2}-\frac{e_{j+3}}{e_{j+2}}e_{j+1}\right)+b_{k-1,1}^{q}\left(e_{j+3}-\frac{e_{j+4}}{e_{j+3}}e_{j+2}\right)=0$. And by Lemma \ref{L2.11 }, we know that $\mathrm{rank}\left(\boldsymbol{A}_{2\times 2}\right)=0$ if and only if $\varGamma_{2}=0$.  
		Thus 
		$$
		\mathrm{rank}\left( \boldsymbol{GG}^\dagger \right) = \begin{array}{c}
			\begin{cases}
				q, &\text{if } \varGamma_{2}=0;\\
				q+1, &\text{if }\varGamma_{2}\neq 0.
			\end{cases}
		\end{array}
		$$
		
		$(2)$ If $b_{k-2,0}=b_{k-1,0}=0$, then we have 
		
		$$
		\boldsymbol{A}_{2\times 2}=\left(\begin{matrix}
			0 & b_{k-2,1}\left(e_{j+1}-\frac{e_{j+2}}{e_{j+1}}e_{j}\right)\\
			b_{k-2,1}^{q}\left(e_{j+3}-\frac{e_{j+4}}{e_{j+3}}e_{j+2}\right) & b_{k-1,1}\left(e_{j+1}-\frac{e_{j+2}}{e_{j+1}}e_{j}\right)+b_{k-1,1}^{q}\left(e_{j+3}-\frac{e_{j+4}}{e_{j+3}}e_{j+2}\right)\\
		\end{matrix}\right).
		$$
		And so,  $\mathrm{rank}\left(\boldsymbol{A}_{2\times 2}\right)=2$, i.e., $\mathrm{rank}\left( \boldsymbol{GG}^\dagger \right)=q+2$.


		$(3)$ If $b_{k-2,0}=b_{k-1,1}=0$, then we have 
		
		$$
		\boldsymbol{A}_{2\times 2}=\left(\begin{matrix}
			0 & b_{k-2,1}\left(e_{j+1}-\frac{e_{j+2}}{e_{j+1}}e_{j}\right)\\
			b_{k-2,1}^{q}\left(e_{j+3}-\frac{e_{j+4}}{e_{j+3}}e_{j+2}\right) & b_{k-1,0}^{q+1}\left(e_{j+2}-\frac{e_{j+3}}{e_{j+2}}e_{j+1}\right)\\
		\end{matrix}\right).
		$$
		And so,  $\mathrm{rank}\left(\boldsymbol{A}_{2\times 2}\right)=2$, i.e., $\mathrm{rank}\left( \boldsymbol{GG}^\dagger \right)=q+2$.


		$(4)$ If $b_{k-1,0}=b_{k-2,1}=0$, then we have 
		
		$$
		\boldsymbol{A}_{2\times 2}=\left(\begin{matrix}
			b_{k-2,0}^{q+1}\left(e_{j+2}-\frac{e_{j+3}}{e_{j+2}}e_{j+1}\right) & 0\\
			0 & b_{k-1,1}\left(e_{j+1}-\frac{e_{j+2}}{e_{j+1}}e_{j}\right)+b_{k-1,1}^{q}\left(e_{j+3}-\frac{e_{j+4}}{e_{j+3}}e_{j+2}\right)\\
		\end{matrix}\right).
		$$
		Thus,  $\mathrm{rank}\left(\boldsymbol{A}_{2\times 2}\right)=1$ if and only if $b_{k-1,1}\left(e_{j+1}-\frac{e_{j+2}}{e_{j+1}}e_{j}\right)+b_{k-1,1}^{q}\left(e_{j+3}-\frac{e_{j+4}}{e_{j+3}}e_{j+2}\right)=0$. Now by Corollary \ref{C2.2}, we know that  $\mathrm{rank}\left(\boldsymbol{A}_{2\times 2}\right)=1$ if and only if $b_{k-1,1}^{q-1}\gamma^{(i-1)(q-1)}=-1$. Thus 
		$$
		\mathrm{rank}\left( \boldsymbol{GG}^\dagger \right) = \begin{array}{c}
			\begin{cases}
				q+1,&\text{if } b_{k-1,1}^{q-1}\gamma^{(i-1)(q-1)}=-1;\\
				q+2,&\text{if }  b_{k-1,1}^{q-1}\gamma^{(i-1)(q-1)}\neq -1. 
			\end{cases}
		\end{array} 
		$$


		$(5)$ If $b_{k-2,1}=b_{k-1,1}=0$, then we have
		
		$$
		\boldsymbol{A}_{2\times 2}=\left(\begin{matrix}
			b_{k-2,0}^{q+1}\left(e_{j+2}-\frac{e_{j+3}}{e_{j+2}}e_{j+1}\right) &b_{k-2,0}b_{k-1,0}^{q}\left(e_{j+2}-\frac{e_{j+3}}{e_{j+2}}e_{j+1}\right)\\
			b_{k-2,0}^{q}b_{k-1,0}\left(e_{j+2}-\frac{e_{j+3}}{e_{j+2}}e_{j+1}\right) &b_{k-1,0}^{q+1}\left(e_{j+2}-\frac{e_{j+3}}{e_{j+2}}e_{j+1}\right)\\
		\end{matrix}\right).
		$$
		Note that $$\frac{	b_{k-2,0}^{q}b_{k-1,0}\left(e_{j+2}-\frac{e_{j+3}}{e_{j+2}}e_{j+1}\right)}{b_{k-2,0}^{q+1}\left(e_{j+2}-\frac{e_{j+3}}{e_{j+2}}e_{j+1}\right)}=\frac{b_{k-1,0}^{q+1}\left(e_{j+2}-\frac{e_{j+3}}{e_{j+2}}e_{j+1}\right)}{b_{k-2,0}b_{k-1,0}^{q}\left(e_{j+2}-\frac{e_{j+3}}{e_{j+2}}e_{j+1}\right)}=\frac{b_{k-1,0}}{b_{k-2,0}},$$ and so, $\mathrm{rank}\left(\boldsymbol{A}_{2\times 2}\right)=1$, i.e., $\mathrm{rank}\left( \boldsymbol{GG}^\dagger \right)=q+1$.


		$(6)$ If $b_{k-1,0}=b_{k-1,1}=0$, then we have
		$$
		\boldsymbol{A}_{2\times 2}=\left(\begin{matrix}
			b_{k-2,0}^{q+1}\left(e_{j+2}-\frac{e_{j+3}}{e_{j+2}}e_{j+1}\right) & b_{k-2,1}\left(e_{j+1}-\frac{e_{j+2}}{e_{j+1}}e_{j}\right)\\
			b_{k-2,1}^{q}\left(e_{j+3}-\frac{e_{j+4}}{e_{j+3}}e_{j+2}\right) & 0\\
		\end{matrix}\right).
		$$
		And so,  $\mathrm{rank}\left(\boldsymbol{A}_{2\times 2}\right)=2$, i.e., $\mathrm{rank}\left( \boldsymbol{GG}^\dagger \right)=q+2$. 
		
		\textbf{case 2.4.} If only one of $b_{k-2,0},b_{k-2,1},b_{k-1,0},b_{k-1,1}$ is zero, then we have the following $4$ cases. 
		
		$(1)$ If $b_{k-2,0} = 0$, then we have
		
		$$
		\boldsymbol{A}_{2\times 2}=\left(\begin{matrix}
			0 & b_{k-2,1}\left(e_{j+1}-\frac{e_{j+2}}{e_{j+1}}e_{j}\right)\\
			b_{k-2,1}^{q}\left(e_{j+3}-\frac{e_{j+4}}{e_{j+3}}e_{j+2}\right) & b_{k-1,1}\left(e_{j+1}-\frac{e_{j+2}}{e_{j+1}}e_{j}\right)+b_{k-1,0}^{q+1}\left(e_{j+2}-\frac{e_{j+3}}{e_{j+2}}e_{j+1}\right)+b_{k-1,1}^{q}\left(e_{j+3}-\frac{e_{j+4}}{e_{j+3}}e_{j+2}\right)\\
		\end{matrix}\right).
		$$
		And so,  $\mathrm{rank}\left(\boldsymbol{A}_{2\times 2}\right)=2$, i.e., $\mathrm{rank}\left( \boldsymbol{GG}^\dagger \right)=q+2$. 
		
		$(2)$ If $b_{k-2,1} = 0$, then we have
		
		\begin{equation}
			\label{M3.17}
			\boldsymbol{A}_{2\times 2}=\left(\begin{matrix}
				b_{k-2,0}^{q+1}\left(e_{j+2}-\frac{e_{j+3}}{e_{j+2}}e_{j+1}\right)& b_{k-2,0}b_{k-1,0}^{q}\left(e_{j+2}-\frac{e_{j+3}}{e_{j+2}}e_{j+1}\right)\\
				b_{k-2,0}^{q}b_{k-1,0}\left(e_{j+2}-\frac{e_{j+3}}{e_{j+2}}e_{j+1}\right) & b_{k-1,1}\left(e_{j+1}-\frac{e_{j+2}}{e_{j+1}}e_{j}\right)+b_{k-1,0}^{q+1}\left(e_{j+2}-\frac{e_{j+3}}{e_{j+2}}e_{j+1}\right)+b_{k-1,1}^{q}\left(e_{j+3}-\frac{e_{j+4}}{e_{j+3}}e_{j+2}\right)\\
			\end{matrix}\right). 
		\end{equation}

		Next, for convenience, we denote  $\boldsymbol{a}_{u_1,3}$ be the $u_1$-th row of  the matrix  $\boldsymbol{A}_{2\times 2}$ given by the Equation \eqref{M3.17}, where $u_1=1,2$.  Then for the matrix  $\boldsymbol{A}_{2\times 2}$ given by  \eqref{M3.17}, we perform the following $2$ steps of elementary transformations.
		
		\textbf{Step 1.} Replace $\boldsymbol{a}_{2,3}$ with $\boldsymbol{a}_{2,3}-\frac{b_{k-1,0}}{b_{k-2,0}}\boldsymbol{a}_{1,3}$, we obtain the
		matrix $\boldsymbol{A}_{2\times 2}^{(1)}$ and denote the vector $\boldsymbol{b}_{3,v_1}^{(1)}$ be the $v_1$-th column of $\boldsymbol{A}_{2\times 2}^{(1)}$, where $v_1=1,2$;
		
		\textbf{Step 2.} Replace $\boldsymbol{b}_{3,2}^{(1)}$ with $\boldsymbol{b}_{3,2}^{(1)}-\left(\frac{b_{k-1,0}}{b_{k-2,0}}\right)^q\boldsymbol{b}_{3,1}^{(1)}$, we obtain the
		matrix 
		$$
		\boldsymbol{A}_{2\times 2}^1=\left(\begin{matrix}
			b_{k-2,0}^{q+1}\left(e_{j+2}-\frac{e_{j+3}}{e_{j+2}}e_{j+1}\right)& 0\\
			0 & b_{k-1,1}\left(e_{j+1}-\frac{e_{j+2}}{e_{j+1}}e_{j}\right)+b_{k-1,1}^{q}\left(e_{j+3}-\frac{e_{j+4}}{e_{j+3}}e_{j+2}\right)\\
		\end{matrix}\right). 
		$$
		Hence,  $\mathrm{rank}\left(\boldsymbol{A}_{2\times 2}\right)=\mathrm{rank}(\boldsymbol{A}_{2\times 2}^1)=1$ if and only if $$b_{k-1,1}\left(e_{j+1}-\frac{e_{j+2}}{e_{j+1}}e_{j}\right)+b_{k-1,1}^{q}\left(e_{j+3}-\frac{e_{j+4}}{e_{j+3}}e_{j+2}\right)=0 .$$
		Now by Corollary \ref{C2.2}, we know that $\mathrm{rank}\left(\boldsymbol{A}_{2\times 2}\right)=1$ if and only if $b_{k-1,1}^{q-1}\gamma^{(i-1)(q-1)}=-1$.  
		Thus 
		$$
		\mathrm{rank}\left( \boldsymbol{GG}^\dagger \right) =\begin{array}{c}
			\begin{cases}
				q+1,&\text{if } b_{k-1,1}^{q-1}\gamma^{(i-1)(q-1)}=-1;\\
				q+2,& \text{if  } b_{k-1,1}^{q-1}\gamma^{(i-1)(q-1)}\neq -1.  
			\end{cases}
		\end{array} 
		$$


		$(3)$ If $b_{k-1,0} = 0$, then we have
		
		$$
		\boldsymbol{A}_{2\times 2}=\left(\begin{matrix}
			b_{k-2,0}^{q+1}\left(e_{j+2}-\frac{e_{j+3}}{e_{j+2}}e_{j+1}\right)& b_{k-2,1}\left(e_{j+1}-\frac{e_{j+2}}{e_{j+1}}e_{j}\right)\\
			b_{k-2,1}^{q}\left(e_{j+3}-\frac{e_{j+4}}{e_{j+3}}e_{j+2}\right) & b_{k-1,1}\left(e_{j+1}-\frac{e_{j+2}}{e_{j+1}}e_{j}\right)+b_{k-1,1}^{q}\left(e_{j+3}-\frac{e_{j+4}}{e_{j+3}}e_{j+2}\right)\\
		\end{matrix}\right). 
		$$
		Hence,  $\mathrm{rank}\left(\boldsymbol{A}_{2\times 2}\right)=1$ if and only if $$\frac{b_{k-2,1}^{q}\left(e_{j+3}-\frac{e_{j+4}}{e_{j+3}}e_{j+2}\right)}{b_{k-2,0}^{q+1}\left(e_{j+2}-\frac{e_{j+3}}{e_{j+2}}e_{j+1}\right)}=\frac{b_{k-1,1}}{b_{k-2,1}}+\frac{b_{k-1,1}^{q}\left(e_{j+3}-\frac{e_{j+4}}{e_{j+3}}e_{j+2}\right)}{b_{k-2,1}\left(e_{j+1}-\frac{e_{j+2}}{e_{j+1}}e_{j}\right)}  .$$
		Now by Corollary \ref{C2.3}, we know that  $\mathrm{rank}\left(\boldsymbol{A}_{2\times 2}\right)=1$ if and only if  $$b_{k-1,1}\left(1+b_{k-1,1}^{q-1}\cdot\gamma^{(i-1)(q-1)} \right)\left(1+\gamma ^{i\left( q-1 \right)}\right)\left(1+\gamma ^{2\left( q-1 \right)}\right)
		=2\left(\frac{b_{k-2,1}}{b_{k-2,0}}\right)^{q+1}\gamma^{i(q-1)}\cdot\left(1+\gamma ^{\left( q-1 \right)}\right). $$
		Thus 
		$$
		\mathrm{rank}\left( \boldsymbol{GG}^\dagger \right) =\begin{array}{c}
			\begin{cases}
				q+1, &\text{if }  b_{k-1,1}\left(1+b_{k-1,1}^{q-1}\cdot\gamma^{(i-1)(q-1)} \right)\left(1+\gamma ^{i\left( q-1 \right)}\right)\left(1+\gamma ^{2\left( q-1 \right)}\right)\\
				&\quad=2\left(\frac{b_{k-2,1}}{b_{k-2,0}}\right)^{q+1}\gamma^{i(q-1)}\cdot\left(1+\gamma ^{\left( q-1 \right)}\right);\\
				q+2,&\text{if }  b_{k-1,1}\left(1+b_{k-1,1}^{q-1}\cdot\gamma^{(i-1)(q-1)} \right)\left(1+\gamma ^{i\left( q-1 \right)}\right)\left(1+\gamma ^{2\left( q-1 \right)}\right)\\
				&\quad\neq 2\left(\frac{b_{k-2,1}}{b_{k-2,0}}\right)^{q+1}\gamma^{i(q-1)}\cdot\left(1+\gamma ^{\left( q-1 \right)}\right).
			\end{cases}
		\end{array}  
		$$


			$(4)$ If $b_{k-1,1} = 0$, then we have
			
			\begin{equation}
				\label{M3.18}
				\begin{aligned}
					&\boldsymbol{A}_{2\times 2}\\
					&=\left(\begin{matrix}
						b_{k-2,0}^{q+1}\left(e_{j+2}-\frac{e_{j+3}}{e_{j+2}}e_{j+1}\right)& b_{k-2,1}\left(e_{j+1}-\frac{e_{j+2}}{e_{j+1}}e_{j}\right) +b_{k-2,0}b_{k-1,0}^{q}\left(e_{j+2}-\frac{e_{j+3}}{e_{j+2}}e_{j+1}\right)\\
						b_{k-2,0}^{q}b_{k-1,0}\left(e_{j+2}-\frac{e_{j+3}}{e_{j+2}}e_{j+1}\right)+b_{k-2,1}^{q}\left(e_{j+3}-\frac{e_{j+4}}{e_{j+3}}e_{j+2}\right) & b_{k-1,0}^{q+1}\left(e_{j+2}-\frac{e_{j+3}}{e_{j+2}}e_{j+1}\right)\\
					\end{matrix}\right).
				\end{aligned}
			\end{equation}
			Next, for convenience, we denote  $\boldsymbol{b}_{4,v_1}$ be the $v_1$-th column of the matrix $\boldsymbol{A}_{2\times 2}$ given by the Equation \eqref{M3.18}, where $v_1=1,2,$  then replace $\boldsymbol{b}_{4,2}$ with $\boldsymbol{b}_{4,2}-\left(\frac{b_{k-1,0}}{b_{k-2,0}}\right)^q\boldsymbol{b}_{4,1}$, we obtain the matrix  
			$$
			\boldsymbol{A}_{2\times 2}^2=\left(\begin{matrix}
				b_{k-2,0}^{q+1}\left(e_{j+2}-\frac{e_{j+3}}{e_{j+2}}e_{j+1}\right)& b_{k-2,1}\left(e_{j+1}-\frac{e_{j+2}}{e_{j+1}}e_{j}\right) \\
				b_{k-2,0}^{q}b_{k-1,0}\left(e_{j+2}-\frac{e_{j+3}}{e_{j+2}}e_{j+1}\right)+b_{k-2,1}^{q}\left(e_{j+3}-\frac{e_{j+4}}{e_{j+3}}e_{j+2}\right) & -\left(\frac{b
					_{10}b_{k-2,1}}{b_{k-2,0}}\right)^q\left(e_{j+3}-\frac{e_{j+4}}{e_{j+3}}e_{j+2}\right)\\
			\end{matrix}\right). 
			$$
			Hence,   $\mathrm{rank}\left(\boldsymbol{A}_{2\times 2}\right)=\mathrm{rank}(\boldsymbol{A}_{2\times 2}^2)=1$ if and only if
			
			$$
			\frac{b_{k-2,1}^q\left(e_{j+3}-\frac{e_{j+4}}{e_{j+3}}e_{j+2}\right)}{b_{k-2,0}^{q+1}\left(e_{j+2}-\frac{e_{j+3}}{e_{j+2}}e_{j+1}\right)}+\frac{b_{k-1,0}}{b_{k-2,0}}=\frac{-b_{k-1,0}^qb_{k-2,1}^q\left(e_{j+3}-\frac{e_{j+4}}{e_{j+3}}e_{j+2}\right)}{b_{k-2,0}^qb_{k-2,1}\left(e_{j+1}-\frac{e_{j+2}}{e_{j+1}}e_{j}\right) }. 
			$$
			Now by Corollary \ref{C2.4}, we know that $\mathrm{rank}\left(\boldsymbol{A}_{2\times 2}\right)=1$ if and only if
			
			$$\Delta_{1}\left[1+\left(\Delta_{1}\cdot\gamma^{(i-1)}\right)^{q-1} \right]\left(1+\gamma ^{i\left( q-1 \right)}\right)\left(1+\gamma ^{2\left( q-1 \right)}\right)=2\left(\frac{b_{k-2,1}}{b_{k-2,0}}\right)^{q+1}\gamma^{i(q-1)}\cdot\left(1+\gamma ^{\left( q-1 \right)}\right),$$
			where $\Delta_{1}=-\frac{b_{k-2,1}b_{k-1,0}}{b_{k-2,0}}$. 
			Thus 
			
			$$
			\mathrm{rank}\left( \boldsymbol{GG}^\dagger \right) =\begin{array}{c}
				\begin{cases}
					q+1, &\text{if }  \Delta_{1}\left[1+\left(\Delta_{1}\cdot\gamma^{(i-1)}\right)^{q-1} \right]\left(1+\gamma ^{i\left( q-1 \right)}\right)\left(1+\gamma ^{2\left( q-1 \right)}\right)\\
					&\quad=2\left(\frac{b_{k-2,1}}{b_{k-2,0}}\right)^{q+1}\gamma^{i(q-1)}\cdot\left(1+\gamma ^{\left( q-1 \right)}\right);\\
					q+2,&\text{if }  \Delta_{1}\left[1+\left(\Delta_{1}\cdot\gamma^{(i-1)}\right)^{q-1} \right]\left(1+\gamma ^{i\left( q-1 \right)}\right)\left(1+\gamma ^{2\left( q-1 \right)}\right)\\
					&\quad\neq 2\left(\frac{b_{k-2,1}}{b_{k-2,0}}\right)^{q+1}\gamma^{i(q-1)}\cdot\left(1+\gamma ^{\left( q-1 \right)}\right),
				\end{cases}
			\end{array} 
			$$	
			where $\Delta_{1}=-\frac{b_{k-2,1}b_{k-1,0}}{b_{k-2,0}}$.

			\textbf{case 2.5.} If $b_{k-2,0}$, $b_{k-2,1}$, $b_{k-1,0}$ and $b_{k-1,1}$ are all non-zero, then we have  
			\begin{equation}
				\label{M3.19}
				\boldsymbol{A}_{2\times 2}=\left(\begin{matrix}
					C_1 & C_2\\
					C_3 & C_4\\
				\end{matrix}\right),
			\end{equation}
			where $$
			\begin{aligned}
				&C_1=b_{k-2,0}^{q+1}\left(e_{j+2}-\frac{e_{j+3}}{e_{j+2}}e_{j+1}\right),\\ &C_2=b_{k-2,1}\left(e_{j+1}-\frac{e_{j+2}}{e_{j+1}}e_{j}\right) +b_{k-2,0}b_{k-1,0}^{q}\left(e_{j+2}-\frac{e_{j+3}}{e_{j+2}}e_{j+1}\right), \\
				&C_3=b_{k-2,0}^{q}b_{k-1,0}\left(e_{j+2}-\frac{e_{j+3}}{e_{j+2}}e_{j+1}\right)+b_{k-2,1}^{q}\left(e_{j+3}-\frac{e_{j+4}}{e_{j+3}}e_{j+2}\right)\\
			\end{aligned}
			$$
			and
			$$C_4=b_{k-1,1}\left(e_{j+1}-\frac{e_{j+2}}{e_{j+1}}e_{j}\right)+b_{k-1,0}^{q+1}\left(e_{j+2}-\frac{e_{j+3}}{e_{j+2}}e_{j+1}\right)+b_{k-1,1}^{q}\left(e_{j+3}-\frac{e_{j+4}}{e_{j+3}}e_{j+2}\right).$$
			
			Next, for convenience, we denote  $\boldsymbol{a}_{u_1,5}$ be the $u_1$-th row of the matrix  $\boldsymbol{A}_{2\times 2}$ given by the Equation \eqref{M3.19}, where $u_1=1,2$.  Then for the matrix  $\boldsymbol{A}_{2\times 2}$ given by \eqref{M3.19}, we  perform the following $2$ steps of elementary transformations.
			
			\textbf{Step 1.} Replace $\boldsymbol{a}_{2,5}$ with $\boldsymbol{a}_{2,5}-\frac{b_{k-1,0}}{b_{k-2,0}}\boldsymbol{a}_{1,5}$, we obtain the
			matrix $\boldsymbol{A}_{2\times 2}^{(2)}$ and denote the vector $\boldsymbol{b}_{5,v_1}^{(2)}$ be the $v_1$-th column of $\boldsymbol{A}_{2\times 2}^{(2)}$, where $v_1=1,2$;
			
			\textbf{Step 2.} Replace $\boldsymbol{b}_{5,2}^{(2)}$ with $\boldsymbol{b}_{5,2}^{(2)}-\left(\frac{b_{k-1,0}}{b_{k-2,0}}\right)^q\boldsymbol{b}_{5,1}^{(2)}$, we obtain the
			matrix 
			$$
			\boldsymbol{A}_{2\times 2}^{3}=\left(\begin{matrix}
				b_{k-2,0}^{q+1}\left(e_{j+2}-\frac{e_{j+3}}{e_{j+2}}e_{j+1}\right) & b_{k-2,1}\left(e_{j+1}-\frac{e_{j+2}}{e_{j+1}}e_{j}\right)\\
				b_{k-2,1}^{q}\left(e_{j+3}-\frac{e_{j+4}}{e_{j+3}}e_{j+2}\right) & C_4^{1}\\
			\end{matrix}\right),
			$$
			where
			$$
			C_4^{1}=\left(b_{k-1,1}-\frac{b_{k-2,1}b_{k-1,0}}{b_{k-2,0}}\right)\left(e_{j+1}-\frac{e_{j+2}}{e_{j+1}}e_{j}\right)+\left(b_{k-1,1}^{q}-\frac{b_{k-1,0}^qb_{k-2,1}^q}{b_{k-2,0}^q}\right)\left(e_{j+3}-\frac{e_{j+4}}{e_{j+3}}e_{j+2}\right).
			$$
			Hence, $\mathrm{rank}(A_{2\times 2})=\mathrm{rank}(A_{2\times 2}^3)=1$ if and only if
			
			$$
			\frac{b_{k-1,1}}{b_{k-2,1}}-\frac{b_{k-1,0}}{b_{k-2,0}}+\frac{\left(b_{k-1,1}^{q}-\frac{b_{k-1,0}^qb_{k-2,1}^q}{b_{k-2,0}^q}\right)\left(e_{j+3}-\frac{e_{j+4}}{e_{j+3}}e_{j+2}\right)}{b_{k-2,1}\left(e_{j+1}-\frac{e_{j+2}}{e_{j+1}}e_{j}\right)}=\frac{b_{k-2,1}^{q}\left(e_{j+3}-\frac{e_{j+4}}{e_{j+3}}e_{j+2}\right)}{b_{k-2,0}^{q+1}\left(e_{j+2}-\frac{e_{j+3}}{e_{j+2}}e_{j+1}\right)}. 
			$$
			And by Lemma \ref{L2.12 }, we know that $\mathrm{rank}(A_{2\times 2})=1$ if and only if $\varGamma=0$. 
			Thus 
			$$
			\mathrm{rank}\left( \boldsymbol{GG}^\dagger \right) =\begin{array}{c}
				\begin{cases}
					q+1,& \text{if } \varGamma=0;	\\
					q+2,&\text{if }  \varGamma\neq 0.
				\end{cases}
			\end{array} 
			$$	
			
			Now by combining with the discussions of the above $4$ cases, we conclude that
			$$
			\mathrm{rank}\left( \boldsymbol{GG}^\dagger \right) =\begin{array}{c}
				\begin{cases}
					q, & \text{if } i=q\\
					&\quad \text{or } i\neq q ,b_{k-2,1}=b_{k-2,0}=0 , b_{k-1,1}\neq 0 \text{ and }\\ 
					&\quad  2b_{k-1,1}\left(1+\gamma^{1-q}\right)\left(1+b_{k-1,1}^{q-1}\gamma^{(q-1)(i-1)}\right)=-b_{k-1,0}^{q+1}\left(1+\gamma^{i(q-1)}\right)\left(1+\gamma^{2(1-q)}\right);\\
					q+2,&\text{if } i\neq q  \text{ and  } b_{k-2,0}=0, b_{k-2,1}\neq 0\\
					&\quad \text{or } i\neq q \text{ and  } b_{k-1,0}=b_{k-1,1}=0, b_{k-2,0}, b_{k-2,1}\neq 0\\ 
					&\quad \text{or } i\neq q, b_{k-2,1}=0, b_{k-2,0},b_{k-1,1}\neq 0 \text{ and  } b_{k-1,1}^{q-1}\gamma^{(i-1)(q-1)}\neq -1\\
					&\quad \text{or } i\neq q , b_{k-2,1},b_{k-2,0}\neq 0,\left(b_{k-1,0},b_{k-1,1}\right)\neq (0,0) \text{ and  } \varGamma\neq 0;\\
					q+1, &\text{otherwise}.
				\end{cases}
			\end{array}  
			$$ 
			
			And so, by Lemma \ref{L2.16 }, we complete the proof of Theorem \ref{T3.1}. 
			
			\quad
			
			\textbf{The proofs of Theorems \ref{T3.2}-\ref{T3.3}}
			
			By $\gcd(i,q+1)=h>1$ and Proposition \ref{P3.3}, we only need to compute  $\mathrm{rank}\left(\boldsymbol{A}_{7\times 7}\right)$.

				
		
		\textbf{The proof of Theorem \ref{T3.2} } 
		
		By $2\mid i$ and Lemma \ref{L2.5 },  we have $e_{j+2}=0$. Then by Lemma \ref{L2.4 }, we have $e_{j+1}\in \mathbb{F}_{q^2}^{*}$ and $e_{j+3}\in \mathbb{F}_{q^2}^{*}$. Thus,  
		
		\begin{equation}
			\label{M3.21}
			\boldsymbol{A}_{7\times 7}=\left(\begin{matrix}
				&		&		&		&		e_{j+4}&		0	&		0	\\
				&		&		&		e_{j+3}&		&	0&		0	\\
				&		&		0&		&		&			b_{k-2,0}^{q}e_{j+1}&		b_{k-1,0}^{q}e_{j+1}	\\
				&		e_{j+1}&		&		&		&	0	&		e_j	\\
				e_j&		&		&		&		&	e_{j-1}	&	0	\\
				0	&	0&		b_{k-2,0}e_{j+3}&	0	&	e_{j+5}	&	0	&	b_{k-2,1}e_{j+1}	\\
				0	&		0&		b_{k-1,0}e_{j+3}&		e_{j+4}&0	&	b_{k-2,1}^{q}e_{j+3}	&	b_{k-1,1}e_{j+1}+b_{k-1,1}^{q}e_{j+3}	\\
			\end{matrix}\right)\ .
		\end{equation}
		Next, depending on  $i=\frac{q+1}{2}$ or not, we  divide the following 2 cases to determine the value of $\mathrm{rank}\left(\boldsymbol{A}_{7\times 7}\right)$. 
		
		\textbf{case 1.} If $i=\frac{q+1}{2}$, then by Lemma  \ref{L2.6 1} and Lemma \ref{L2.4 }, we have $e_{j+4}=e_{j}=0$, $e_{j-1}\in \mathbb{F}_{q^2}^{*}$ and  $e_{j+5}\in \mathbb{F}_{q^2}^{*}$. 
		Thus,   
		
		\begin{equation}
			\label{A11}
			\boldsymbol{A}_{7\times 7}=\left(\begin{matrix}
				&		&		&		&		0&		0	&		0	\\
				&		&		&		e_{j+3}&		&	0&		0	\\
				&		&		0&		&		&			b_{k-2,0}^{q}e_{j+1}&		b_{k-1,0}^{q}e_{j+1}	\\
				&		e_{j+1}&		&		&		&	0	&		0\\
				0&		&		&		&		&	e_{j-1}	&	0	\\
				0	&	0&		b_{k-2,0}e_{j+3}&	0	&	e_{j+5}	&	0&	b_{k-2,1}e_{j+1}	\\
				0	&		0&		b_{k-1,0}e_{j+3}&		0&0	&	b_{k-2,1}^{q}e_{j+3}	&	b_{k-1,1}e_{j+1}+b_{k-1,1}^{q}e_{j+3}\\
			\end{matrix}\right)\ . 
		\end{equation}
		Now, for convenience, we denote the vector $\boldsymbol{a}_{u_2,6}$  be the $u_2$-th row of the matrix  $\boldsymbol{A}_{7\times 7}$ given by the Equation \eqref{A11}, where $1 \le u_2\le 7$. Next, for the matrix $\boldsymbol{A}_{7\times 7}$ given by  \eqref{A11}, we perform the following $2$ steps of elementary transformations.

		\textbf{Step 1.} Replace $\boldsymbol{a}_{3,6}$ with $\boldsymbol{a}_{3,6}-\frac{b_{k-2,0}^{q}e_{j+1}}{e_{j-1}}\boldsymbol{a}_{5,6}$, and replace $\boldsymbol{a}_{7,6}$ with $\boldsymbol{a}_{7,6}-\frac{b_{k-2,1}^{q}e_{j+3}}{e_{j-1}}\boldsymbol{a}_{5,6}$,  we obtain the
		matrix $\boldsymbol{A}_{7\times 7}^{(1)}$ and denote the vector $\boldsymbol{b}_{6,v_{2}}^{(1)}$ be the $v_{2}$-th column of $\boldsymbol{A}_{7\times 7}^{(1)}$, where $1\le v_{2}\le 7$; 
		
		\textbf{Step 2.} Replace $\boldsymbol{b}_{6,3}^{(1)}$ with $\boldsymbol{b}_{6,3}^{(1)}-\frac{b_{k-2,0}e_{j+3}}{e_{j+5}}\boldsymbol{b}_{6,5}^{(1)}$, and replace $\boldsymbol{b}_{6,7}^{(1)}$ with $\boldsymbol{b}_{6,7}-\frac{b_{k-2,1}e_{j+1}}{e_{j+5}}\boldsymbol{b}_{6,5}^{(1)}$, we obtain the matrix 
		\begin{equation}
			\boldsymbol{A}_{7\times 7}^1=\left(\begin{matrix}
				&		&		&		&		0&		0	&		0	\\
				&		&		&		e_{j+3}&		&	0&		0	\\
				&		&		0&		&		&			0&		b_{k-1,0}^{q}e_{j+1}	\\
				&		e_{j+1}&		&		&		&	0	&		0\\
				0&		&		&		&		&	e_{j-1}	&	0	\\
				0	&	0&	0&	0	&	e_{j+5}	&	0& 0	\\
				0	&		0&		b_{k-1,0}e_{j+3}&		0&0	&	0	&	b_{k-1,1}e_{j+1}+b_{k-1,1}^{q}e_{j+3}\\
			\end{matrix}\right)\ , 
		\end{equation}
		and  $\mathrm{rank}\left(\boldsymbol{A}_{7\times 7}\right)=\mathrm{rank}\left(\boldsymbol{A}_{7\times 7}^{1}\right)$. Next, we divide the following $2$ cases basing on $b_{k-1,0}=0$ or not. 
		
		(1) If $b_{k-1,0}=0$, then we have
		
		\begin{equation}
			\boldsymbol{A}_{7\times 7}^{1}=\left(\begin{matrix}
				&		&		&		&		0&		0	&		0	\\
				&		&		&		e_{j+3}&		&	0&		0	\\
				&		&		0&		&		&			0&		0	\\
				&		e_{j+1}&		&		&		&	0	&		0\\
				0&		&		&		&		&	e_{j-1}	&	0	\\
				0	&	0&	0&	0	&	e_{j+5}	&	0& 0	\\
				0	&		0&		0&		0&0	&	0	&	b_{k-1,1}e_{j+1}+b_{k-1,1}^{q}e_{j+3}\\
			\end{matrix}\right)\ . 
		\end{equation}
		Thus,   $\mathrm{rank}\left(\boldsymbol{A}_{7\times 7}\right)=\mathrm{rank}\left(\boldsymbol{A}_{7\times 7}^{1}\right)=4$ if and only if
		$b_{k-1,1}e_{j+1}+b_{k-1,1}^qe_{j+3}=0$. And by Lemma \ref{L2.13 },  we know that  $\mathrm{rank}\left(\boldsymbol{A}_{7\times 7}\right)=4$ if and only if $b_{k-1,1}=0$ or $b_{k-1,1}^{q-1}\gamma^{(i-1)(q-1)}=1$. 
		Thus 
		
		$$\mathrm{rank}\left( \boldsymbol{GG}^\dagger \right)= \begin{array}{c}
			\begin{cases}
				q-1, &\text{if } b_{k-1,1}=0  \text{ or }  b_{k-1,1}^{q-1}\gamma^{(i-1)(q-1)}=1;\\
				q,& \text{if }  b_{k-1,1}\neq 0 \text{ and }  b_{k-1,1}^{q-1}\gamma^{(i-1)(q-1)}\neq 1.\\ 
			\end{cases}
		\end{array}  $$
		
		(2) If $b_{k-1,0}\neq 0$, then  $\mathrm{rank}(\boldsymbol{A}_{7\times 7}^{1})=6$, i.e. $\mathrm{rank}\left( \boldsymbol{GG}^\dagger \right)=q+1$. 
		
		And so, by combining with (1) and (2), we know that for $i=\frac{q+1}{2}$, 
		
		$$\mathrm{rank}\left( \boldsymbol{GG}^\dagger \right)= \begin{array}{c}
			\begin{cases}
				q-1, &\text{if }  b_{k-1,0}=b_{k-1,1}= 0\\
				&\quad \text{or }  b_{k-1,0}= 0 \text{ and }   b_{k-1,1}^{q-1}\gamma^{(i-1)(q-1)}=1;\\
				q,&\text{if }  b_{k-1,0}=0 , b_{k-1,1}\neq 0 \text{ and }  b_{k-1,1}^{q-1}\gamma^{(i-1)(q-1)}\neq 1;\\ 
				q+1,& \text{if }  b_{k-1,0}\neq 0.
			\end{cases}
		\end{array} $$
		
		\textbf{case 2. } If $i\neq\frac{q+1}{2}$, then by Lemma  \ref{L2.6 1}, we have  $e_{j}\in \mathbb{F}_{q^2}^{*}$ and $e_{j+4}\in \mathbb{F}_{q^2}^{*}$. 
		
		For convenience, we denote the vectors $\boldsymbol{a}_{u_2,7}$ and $\boldsymbol{b}_{7,v_2}$ be the $u_2$-th row and the $v_2$-th column of the matrix $\boldsymbol{A}_{7\times 7}$ given by the Equation \eqref{M3.21}, respectively, where $1 \le u_2,v_2\le 7$. Next, for the matrix $\boldsymbol{A}_{7\times 7}$ given by  \eqref{M3.21}, we perform the following $2$ steps of elementary transformations.
		
		\textbf{Step 1.} Replace $\boldsymbol{a}_{6,7}$ with $\boldsymbol{a}_{6,7}-\frac{e_{j+5}}{e_{j+4}}\boldsymbol{a}_{1,7}$, and replace $\boldsymbol{a}_{7,7}$ with $\boldsymbol{a}_{7,7}-\frac{e_{j+4}}{e_{j+3}}\boldsymbol{a}_{2,7}$; 
		
		\textbf{Step 2.} Replace $\boldsymbol{b}_{7,6}$ with $\boldsymbol{b}_{7,6}-\frac{e_{j-1}}{e_{j}}\boldsymbol{b}_{7,1}$, and replace $\boldsymbol{a}_{7,7}$ with $\boldsymbol{b}_{7,7}-\frac{e_{j}}{e_{j+1}}\boldsymbol{b}_{7,2}$, we obtain the matrix 
		\begin{equation}
			\label{A2}
			\boldsymbol{A}_{7\times 7}^2=\left(\begin{matrix}
				&		&		&		&		e_{j+4}&		0	&		0	\\
				&		&		&		e_{j+3}&		&	0&		0	\\
				&		&		0&		&		&			b_{k-2,0}^{q}e_{j+1}&		b_{k-1,0}^{q}e_{j+1}	\\
				&		e_{j+1}&		&		&		&	0	&	0\\
				e_j&		&		&		&		&	0	&	0	\\
				0	&	0&		b_{k-2,0}e_{j+3}&	0	&	0	&	0	&	b_{k-2,1}e_{j+1}	\\
				0	&		0&		b_{k-1,0}e_{j+3}&		0&0	&	b_{k-2,1}^{q}e_{j+3}	&b_{k-1,1}e_{j+1}+b_{k-1,1}^{q}e_{j+3}	\\
			\end{matrix}\right)\ ,  
		\end{equation}
		and  $\mathrm{rank}\left(\boldsymbol{A}_{7\times 7}\right)=\mathrm{rank}\left(\boldsymbol{A}_{7\times 7}^{2}\right)$. 
		
		Next, we divide the following $5$ cases basing on the values of $b_{k-2,0}$, $b_{k-2,1}$, $b_{k-1,0}$ and $b_{k-1,1}$ are zero or not.  
		
		(1) If $b_{k-2,0}=b_{k-1,0}=b_{k-2,1}=0$, then we have 
		$$
		\boldsymbol{A}_{7\times 7}^{2}=\left(\begin{matrix}
			&		&		&		&		e_{j+4}&		0	&		0	\\
			&		&		&		e_{j+3}&		&	0&		0	\\
			&		&		0&		&		&			0&	0	\\
			&		e_{j+1}&		&		&		&	0	&	0\\
			e_j&		&		&		&		&	0	&	0	\\
			0	&	0&		0&	0	&	0	&	0	&	0	\\
			0	&		0&	0&		0&0	&	0	&b_{k-1,1}e_{j+1}+b_{k-1,1}^{q}e_{j+3}	\\
		\end{matrix}\right)\ . 
		$$
		And so,  $\mathrm{rank}\left(\boldsymbol{A}_{7\times 7}\right)=\mathrm{rank}\left(\boldsymbol{A}_{7\times 7}^{2}\right)=4$ if and only if  $b_{k-1,1}e_{j+1}+b_{k-1,1}^qe_{j+3}=0$. Furthermore,  by Lemma \ref{L2.13 }, we know that $\mathrm{rank}\left(\boldsymbol{A}_{7\times 7}\right)=4$ if and only if $b_{k-1,1}=0$ or  $b_{k-1,1}^{q-1}\gamma^{(i-1)(q-1)}=1$. 
		Thus
		
		$$
		\mathrm{rank}\left( \boldsymbol{GG}^\dagger \right) =\begin{array}{c}
			\begin{cases}
				q-1 , &\text{if }  b_{k-1,1}=0 \text{ or } b_{k-1,1}^{q-1}\gamma^{(i-1)(q-1)}=1;\\ 
				q ,& \text{if }  b_{k-1,1}\neq 0 \text{ and } b_{k-1,1}^{q-1}\gamma^{(i-1)(q-1)}\neq 1. 
			\end{cases}
		\end{array} 
		$$ 
		
		(2) If $b_{k-2,0}=b_{k-1,0}=0$ and $b_{k-2,1}\neq 0$, then we have  
		$$
		\boldsymbol{A}_{7\times 7}^{2}=\left(\begin{matrix}
			&		&		&		&		e_{j+4}&		0	&		0	\\
			&		&		&		e_{j+3}&		&	0&		0	\\
			&		&		0&		&		&			0&		0	\\
			&		e_{j+1}&		&		&		&	0	&	0\\
			e_j&		&		&		&		&	0	&	0	\\
			0	&	0&		0&	0	&	0	&	0	&	b_{k-2,1}e_{j+1}	\\
			0	&		0&		0&		0&0	&	b_{k-2,1}^{q}e_{j+3}	&b_{k-1,1}e_{j+1}+b_{k-1,1}^{q}e_{j+3}	\\
		\end{matrix}\right)\ . 
		$$
		And so,  $\mathrm{rank}(\boldsymbol{A}_{7\times 7}^{2})=6$, i.e., $\mathrm{rank}\left( \boldsymbol{GG}^\dagger \right)=q+1$. 
		
		(3) If $b_{k-2,0}=b_{k-2,1}=0$ and $b_{k-1,0}\neq 0$, then we have $$
		\boldsymbol{A}_{7\times 7}^{2}=\left(\begin{matrix}
			&		&		&		&		e_{j+4}&		0	&		0	\\
			&		&		&		e_{j+3}&		&	0&		0	\\
			&		&		0&		&		&			0&		b_{k-1,0}^{q}e_{j+1}	\\
			&		e_{j+1}&		&		&		&	0	&	0\\
			e_j&		&		&		&		&	0	&	0	\\
			0	&	0&		0&	0	&	0	&	0	&	0	\\
			0	&		0&		b_{k-1,0}e_{j+3}&		0&0	&	0	&b_{k-1,1}e_{j+1}+b_{k-1,1}^{q}e_{j+3}	\\
		\end{matrix}\right)\ . 
		$$
		And so,  $\mathrm{rank}(\boldsymbol{A}_{7\times 7}^{2})=6$, i.e., $\mathrm{rank}\left( \boldsymbol{GG}^\dagger \right)=q+1$. 
		
		(4) If $b_{k-2,0}=0$, $b_{k-2,1}\neq 0$ and $b_{k-1,0}\neq 0$, then we have $$
		\boldsymbol{A}_{7\times 7}^{2}=\left(\begin{matrix}
			&		&		&		&		e_{j+4}&		0	&		0	\\
			&		&		&		e_{j+3}&		&	0&		0	\\
			&		&		0&		&		&			0&		b_{k-1,0}^{q}e_{j+1}	\\
			&		e_{j+1}&		&		&		&	0	&	0\\
			e_j&		&		&		&		&	0	&	0	\\
			0	&	0&		0&	0	&	0	&	0	&	b_{k-2,1}e_{j+1}	\\
			0	&		0&		b_{k-1,0}e_{j+3}&		0&0	&	b_{k-2,1}^{q}e_{j+3}	&b_{k-1,1}e_{j+1}+b_{k-1,1}^{q}e_{j+3}	\\
		\end{matrix}\right)\ . 
		$$
		And so,  $\mathrm{rank}(\boldsymbol{A}_{7\times 7}^{2})=6$, i.e., $\mathrm{rank}\left( \boldsymbol{GG}^\dagger \right)=q+1$.
		
		(5) If $b_{k-2,0}\neq 0$, then
		for convenience, we denote the vector $\boldsymbol{a}_{u_2,8}$  be the $u$-th row of the matrix $\boldsymbol{A}_{7\times 7}^{2}$ given by the Equation \eqref{A2}, where $1 \le u_2\le 7$. Next, for the matrix $\boldsymbol{A}_{7\times 7}^{2}$ given by  \eqref{A2}, we perform the following $2$ steps of elementary transformations.
		
		\textbf{Step 1.} Replace $\boldsymbol{a}_{7,8}$ with $\boldsymbol{a}_{7,8}-\frac{b_{k-1,0}}{b_{k-2,0}}\boldsymbol{a}_{6,8}$, and replace $\boldsymbol{a}_{7,8}$ with $\boldsymbol{a}_{7,8}-\left(\frac{b_{k-2,1}}{b_{k-2,0}}\right)^q\boldsymbol{a}_{3,8}$, we obtain the matrix $\boldsymbol{A}_2^{(1)}$, and denote the vector  $\boldsymbol{b}_{8,v_2}^{(1)}$ be  the $v_2$-th column of $\boldsymbol{A}_2^{(1)}$, where $1 \le v_2\le 7$.
		
		\textbf{Step 2.} Replace $\boldsymbol{b}_{8,7}^{(1)}$ with $\boldsymbol{b}_{8,7}^{(1)}-\frac{b_{k-2,1}}{b_{k-2,0}}\boldsymbol{b}_{8,3}^{(1)}-\left(\frac{b_{k-1,0}}{b_{k-2,0}}\right)^q\boldsymbol{b}_{8,6}^{(1)}$, we obtain the matrix 
		$$	\boldsymbol{A}_{2}^{(2)}=\left(\begin{matrix}
			&		&		&		&		e_{j+4}&		0	&		0	\\
			&		&		&		e_{j+3}&		&	0&		0	\\
			&		&		0&		&		&			b_{k-2,0}^{q}e_{j+1}&		0	\\
			&		e_{j+1}&		&		&		&	0	&	0\\
			e_j&		&		&		&		&	0	&	0	\\
			0	&	0&		b_{k-2,0}e_{j+3}&	0	&	0	&	0	&	0	\\
			0	&		0&		0&		0&0	&	0	&\left(b_{k-1,1}-\frac{b_{k-1,0}b_{k-2,1}}{b_{k-2,0}}\right)e_{j+1}+\left(b_{k-1,1}^q-\frac{b_{k-1,0}^qb_{k-2,1}^q}{b_{k-2,0}^q}\right)e_{j+3}	\\
		\end{matrix}\right)\ . $$
		Hence,   $\mathrm{rank}\left(\boldsymbol{A}_{7\times 7}^{2}\right) =\mathrm{rank}\left(\boldsymbol{A}_{2}^{(2)}\right)=6$ if and only if  $$\left(b_{k-1,1}-\frac{b_{k-1,0}b_{k-2,1}}{b_{k-2,0}}\right)e_{j+1}+\left(b_{k-1,1}^q-\frac{b_{k-1,0}^qb_{k-2,1}^q}{b_{k-2,0}^q}\right)e_{j+3}=0.$$ 
		Now by Lemma \ref{L 1}, we know that $\mathrm{rank}\left(\boldsymbol{A}_{7\times 7}^{2}\right)=\mathrm{rank}\left(\boldsymbol{A}_{2}^{(2)}\right)=6$ if and only if  $b_{k-2,1}b_{k-1,0}=b_{k-1,1}b_{k-2,0}$ or  $$\left(b_{k-1,1}-\frac{b_{k-2,1}b_{k-1,0}}{b_{k-2,0}}\right)^{q-1}\gamma^{(i-1)(q-1)}=1.$$ 
		Thus
		$$
		\mathrm{rank}\left( \boldsymbol{GG}^\dagger \right) =\begin{array}{c}
			\begin{cases}
				q+1 , &\text{if } b_{k-2,1}b_{k-1,0}=b_{k-1,1}b_{k-2,0} \text{ or }  \left(b_{k-1,1}-\frac{b_{k-2,1}b_{k-1,0}}{b_{k-2,0}}\right)^{q-1}\gamma^{(i-1)(q-1)}=1;\\ 
				q +2, & \text{if }  b_{k-2,1}b_{k-1,0}\neq b_{k-1,1}b_{k-2,0}  \text{ and } \left(b_{k-1,1}-\frac{b_{k-2,1}b_{k-1,0}}{b_{k-2,0}}\right)^{q-1}\gamma^{(i-1)(q-1)}\neq 1. 
			\end{cases}
		\end{array}  
		$$ 
				And so, for $i\neq\frac{q+1}{2}$, we have 
				
				$$
				\mathrm{rank}\left( \boldsymbol{GG}^\dagger \right) =\begin{array}{c}
					\begin{cases}
						q-1,  &\text{if }   b_{k-2,0}=b_{k-2,1}=b_{k-1,0}=b_{k-1,1}=0\\
						&\quad \text{or }  b_{k-2,0}=b_{k-2,1}=b_{k-1,0}=0  \text{ and }   b_{k-1,1}^{q-1}\gamma^{(i-1)(q-1)}=1;\\ 
						q, & \text{if }  b_{k-2,0}=b_{k-2,1}=b_{k-1,0}=0, b_{k-1,1}\neq 0 \text{ and }  b_{k-1,1}^{q-1}\gamma^{(i-1)(q-1)}\neq 1;\\ 
						q+2, &\text{if }  b_{k-2,0}\neq 0 ,  b_{k-2,1}b_{k-1,0}\neq b_{k-1,1}b_{k-2,0} \text{ and }  \left(b_{k-1,1}-\frac{b_{k-2,1}b_{k-1,0}}{b_{k-2,0}}\right)^{q-1}\gamma^{(i-1)(q-1)}\neq 1;\\
						q+1,& \text{otherwise}. 						
					\end{cases}
				\end{array} 
				$$
				
				Now by combining with the discussions of the above $2$ cases, we conclude that
				
				$$
				\mathrm{rank}\left( \boldsymbol{GG}^\dagger \right) = \begin{array}{c}
					\begin{cases}
						q-1, & \text{if }  i=\frac{q+1}{2}, b_{k-1,0}=b_{k-1,1}=0 \\
						&\quad \text{or }  i=\frac{q+1}{2},  b_{k-1,0}= 0 \text{ and }  b_{k-1,1}^{q-1}\gamma^{(i-1)(q-1)}=1\\
						&\quad \text{or } i\neq \frac{q+1}{2}, b_{k-2,0}=b_{k-2,1}=b_{k-1,0}=b_{k-1,1}=0\\ 
						&\quad \text{or }  i\neq \frac{q+1}{2} ,  b_{k-2,0}=b_{k-2,1}=b_{k-1,0}=0 \text{ and }  b_{k-1,1}^{q-1}\gamma^{(i-1)(q-1)}=1;\\ 
						q, & \text{if } i=\frac{q+1}{2},  b_{k-1,0}=0 , b_{k-1,1}\neq 0 \text{ and }  b_{k-1,1}^{q-1}\gamma^{(i-1)(q-1)}\neq 1\\
						&\quad \text{or }   i\neq \frac{q+1}{2} ,  b_{k-2,0}=b_{k-2,1}=b_{k-1,0}=0 , b_{k-1,1}\neq 0  \text{ and }  b_{k-1,1}^{q-1}\gamma^{(i-1)(q-1)}\neq 1;\\ 
						q+2, &\text{if }  i\neq \frac{q+1}{2} ,  b_{k-2,0}\neq 0,  b_{k-2,1}b_{k-1,0}\neq b_{k-1,1}b_{k-2,0}  \text{ and } \left(b_{k-1,1}-\frac{b_{k-2,1}b_{k-1,0}}{b_{k-2,0}}\right)^{q-1}\gamma^{(i-1)(q-1)}\neq 1;\\
						q+1,&\text{otherwise}. 					
					\end{cases}
				\end{array} 
				$$
				
				And so, by Lemma \ref{L2.16 }, we complete the proof of Theorem \ref{T3.2}. 
				
				\quad
				
				\textbf{The proof of Theorem \ref{T3.3}  } 
				
				By $2\nmid i$ and Lemma \ref{L2.5 }, we have $e_{j+2}\neq 0$. Next, depending on the relation between $i$ and $q+1$, we  divide the following $3$ cases to determine the value of $\mathrm{rank}\left(\boldsymbol{A}_{7\times 7}\right)$. 
				
				\textbf{case 1.} If $i=\frac{q+1}{2}$, then  by Lemma \ref{L2.6 }, we have $e_{j+1}=e_{j+3}=0$, and by Lemmas \ref{L2.3 }-\ref{L2.4 }, we can get   $e_{j-1}=e_{j+5}=0$, $e_{j}\in \mathbb{F}_{q^2}^{*}$ and $e_{j+4}\in \mathbb{F}_{q^2}^{*}$. And so, we have 
				
				\begin{equation}
					\label{M3.26}
					\boldsymbol{A}_{7\times 7}=\left(\begin{matrix}
						&		&		&		&		e_{j+4}&		0	&		0	\\
						&		&		&		0&		&	b_{k-2,1}^{q}e_{j+2}&		b_{k-1,1}^{q}e_{j+2}	\\
						&		&		e_{j+2}&		&		&			0&		0	\\
						&		0&		&		&		&	0	&		e_j	\\
						e_j&		&		&		&		&	0	&	0	\\
						0	&	b_{k-2,1}e_{j+2}&		&	0	&	0	&	b_{k-2,0}^{q+1}e_{j+2}&	b_{k-2,0}b_{k-1,0}^{q}e_{j+2}	\\
						0	&		b_{k-1,1}e_{j+2}&	0&		e_{j+4}&0	&b_{k-2,0}^{q}b_{k-1,0}e_{j+2}&	b_{k-1,0}^{q+1}e_{j+2}	\\
					\end{matrix}\right)\ . 
				\end{equation}			
				Now, for convenience, we denote the vectors $\boldsymbol{a}_{u_2,9}$ and $\boldsymbol{b}_{9,v_2}$ be the $u_2$-th row and the $v_2$-th column of the matrix  $\boldsymbol{A}_{7\times 7}$ given by the Equation \eqref{M3.26}, respectively, where $1 \le u_2,v_2\le 7$. Next, for the matrix  $\boldsymbol{A}_{7\times 7}$ given by  \eqref{M3.26}, we perform the following 2 steps of elementary transformations.

				\textbf{Step 1.} Replace $\boldsymbol{a}_{2,9}$ with $\boldsymbol{a}_{2,9}-b_{k-1,1}^{q}\frac{e_{j+2}}{e_{j}}\boldsymbol{a}_{4,9}$, replace $\boldsymbol{a}_{6,9}$ with $\boldsymbol{a}_{6,9}-b_{k-2,0}b_{k-1,0}^{q}\frac{e_{j+2}}{e_{j}}\boldsymbol{a}_{4,9}$, and replace $\boldsymbol{a}_{7,9}$ with $\boldsymbol{a}_{7,9}-b_{k-1,0}^{q+1}\frac{e_{j+2}}{e_{j}}\boldsymbol{a}_{4,9}$; 
				
				\textbf{Step 2.} Replace $\boldsymbol{b}_{9,2}$ with $\boldsymbol{b}_{9,2}-b_{k-1,1}\frac{e_{j+2}}{e_{j+4}}\boldsymbol{b}_{9,4}$, and replace $\boldsymbol{b}_{9,6}$ with $\boldsymbol{b}_{9,6}-b_{k-2,0}^{q}b_{k-1,0}\frac{e_{j+2}}{e_{j+4}}\boldsymbol{b}_{9,4}$, we obtain the matrix  
				$$
				\boldsymbol{A}_{7\times 7}^3=\left(\begin{matrix}
					&		&		&		&		e_{j+4}&		0	&		0	\\
					&		&		&		0&		&	b_{k-2,1}^{q}e_{j+2}&		0	\\
					&		&		e_{j+2}&		&		&			0&		0	\\
					&		0&		&		&		&	0	&		e_j	\\
					e_j&		&		&		&		&	0	&	0	\\
					0	&	b_{k-2,1}e_{j+2}&		&	0	&	0	&	b_{k-2,0}^{q+1}e_{j+2}& 0\\
					0	&	0&	0&		e_{j+4}&0	&0 & 0	\\
				\end{matrix}\right)\ . 
				$$				
				Thus, it is easy to get 			
				$$\mathrm{rank}\left( \boldsymbol{GG}^\dagger \right)=q-5+\mathrm{rank}(\boldsymbol{A}_{7\times 7})=q-5+\mathrm{rank}(\boldsymbol{A}_{7\times 7}^3)=\begin{array}{c}
					\begin{cases}
						q,& \text{if } b_{k-2,1}=b_{k-2,0}=0;\\
						q+1,& \text{if } b_{k-2,1}=0 \text{ and }  b_{k-2,0}\neq 0;\\ 
						q+2,& \text{if }  b_{k-2,1}\neq 0.
					\end{cases}
				\end{array}  $$
				
				\textbf{case 2.} If $i=\frac{q+1}{4}$ or $\frac{3(q+1)}{4}$, then by Lemma \ref{L2.7 }, we have $e_{j}=e_{j+4}=0$, and by Lemma \ref{L2.4 }, $e_{j-1}$, $e_{j+1}$, $e_{j+3}$ and $e_{j+5}$ are all non-zero. Hence,  	
				\begin{equation}
					\label{M3.27}
					\boldsymbol{A}_{7\times 7}=\left(\begin{matrix}
						&		&		&		&		0&		0	&		0	\\
						&		&		&		e_{j+3}&		&	b_{k-2,1}^{q}e_{j+2}&		b_{k-1,1}^{q}e_{j+2}	\\
						&		&		e_{j+2}&		&		&			b_{k-2,0}^{q}e_{j+1}&		b_{k-1,0}^{q}e_{j+1}	\\
						&		e_{j+1}&		&		&		&	0	&		0	\\
						0&		&		&		&		&	e_{j-1}	&	0	\\
						0	&	b_{k-2,1}e_{j+2}&		b_{k-2,0}e_{j+3}&	0	&	e_{j+5}	&	b_{k-2,0}^{q+1}e_{j+2}	&	b_{k-2,1}e_{j+1}+b_{k-2,0}b_{k-1,0}^{q}e_{j+2}	\\
						0	&		b_{k-1,1}e_{j+2}&		b_{k-1,0}e_{j+3}&		0 &0	&	b_{k-2,0}^{q}b_{k-1,0}e_{j+2}+b_{k-2,1}^{q}e_{j+3}	&	b_{k-1,1}e_{j+1}+b_{k-1,0}^{q+1}e_{j+2}+b_{k-1,1}^{q}e_{j+3}	\\
					\end{matrix}\right)\ . 
				\end{equation}
				Now, for convenience, we denote the vectors $\boldsymbol{a}_{u_2,10}$ and $\boldsymbol{b}_{10,v_2}$ be the $u_2$-th row and the $v_2$-th column of the matrix  $\boldsymbol{A}_{7\times 7}$ given by the Equation \eqref{M3.27}, respectively, where $1 \le u_2,v_2\le 7$. Next, for the matrix  $\boldsymbol{A}_{7\times 7}$ given by  \eqref{M3.27}, we perform the following $4$ steps of elementary transformations. 
				
				\textbf{Step 1.} Replace $\boldsymbol{a}_{2,10}$ with $\boldsymbol{a}_{2,10}-b_{k-2,1}^{q}\frac{e_{j+2}}{e_{j-1}}\boldsymbol{a}_{5,10}$, and replace $\boldsymbol{a}_{3,10}$ with $\boldsymbol{a}_{3,10}-b_{k-2,0}^{q}\frac{e_{j+1}}{e_{j-1}}\boldsymbol{a}_{5,10}$; 
				
				\textbf{Step 2.} Replace $\boldsymbol{a}_{6,10}$ with $\boldsymbol{a}_{6,10}-b_{k-2,1}\frac{e_{j+2}}{e_{j+1}}\boldsymbol{a}_{4,10}-b_{k-2,0}^{q+1}\frac{e_{j+2}}{e_{j-1}}\boldsymbol{a}_{5,10}$, and replace $\boldsymbol{a}_{7,10}$ with $\boldsymbol{a}_{7,10}-b_{k-1,1}\frac{e_{j+2}}{e_{j+1}}\boldsymbol{a}_{4,10}-\left(b_{k-2,0}^{q}b_{k-1,0}\frac{e_{j+2}}{e_{j-1}}+b_{k-2,1}^q\frac{e_{j+3}}{e_{j-1}}\right)\boldsymbol{a}_{5,10}$;
				
				\textbf{Step 3.} Replace $\boldsymbol{b}_{10,7}$ with $\boldsymbol{b}_{10,7}-b_{k-1,1}^{q}\frac{e_{j+2}}{e_{j+3}}\boldsymbol{b}_{10,4}-\left(b_{k-2,1}\frac{e_{j+1}}{e_{j+5}}+b_{k-2,0}b_{k-1,0}^q\frac{e_{j+2}}{e_{j+5}}\right)\boldsymbol{b}_{10,5}$, and replace $\boldsymbol{b}_{10,3}$ with $\boldsymbol{b}_{10,3}-b_{k-2,0}\frac{e_{j+3}}{e_{j+5}}\boldsymbol{b}_{10,5}$, we obtain the matrix $\boldsymbol{A}_{3}^{(1)}$, and denote the vector $\boldsymbol{b}_{10,v_2}^{(1)}$ be the $v_2$-th column of $\boldsymbol{A}_{3}^{(1)}$, where $1 \le v_2\le 7$;
				
				\textbf{Step 4.} Replace $\boldsymbol{b}_{10,7}^{(1)}$ with $\boldsymbol{b}_{10,7}^{(1)}-b_{k-1,0}^{q+1}\frac{e_{j+1}}{e_{j+2}}\boldsymbol{b}_{10,3}^{(1)}$, we obtain the matrix 
				$$
				\boldsymbol{A}_{7\times 7}^4=\left(\begin{matrix}
					&		&		&		&		0&		0	&		0	\\
					&		&		&		e_{j+3}&		&  0&	0	\\
					&		&		e_{j+2}&		&		&		0&		0	\\
					&		e_{j+1}&		&		&		&	0	&	0	\\
					0&		&		&		&		&	e_{j-1}	&	0	\\
					0	&	0&		0&	0	&	e_{j+5}	&	0	&	0	\\
					0	&	0&		b_{k-1,0}e_{j+3}&		0&0	&	0	&	b_{k-1,1}e_{j+1}+b_{k-1,0}^{q+1}\left(e_{j+2}-\frac{e_{j+1}}{e_{j+2}}e_{j+3}\right)+b_{k-1,1}^{q}e_{j+3} 	\\
				\end{matrix}\right)\ . 
				$$			
				Hence, $\mathrm{rank}\left(
				\boldsymbol{A}_{7\times 7}\right)=\mathrm{rank}\left(
				\boldsymbol{A}_{7\times 7}^{4}\right)=5$ if and only if 	$$b_{k-1,1}e_{j+1}+b_{k-1,0}^{q+1}\left(e_{j+2}-\frac{e_{j+1}}{e_{j+2}}e_{j+3}\right)+b_{k-1,1}^{q}e_{j+3}=0.$$ 
				Now by Lemma \ref{L2.14 },  we know that $\mathrm{rank}\left(
				\boldsymbol{A}_{7\times 7}\right)=\mathrm{rank}\left(
				\boldsymbol{A}_{7\times 7}^{4}\right)=5$ if and only if  $\varGamma_{1}=0$. 
				Thus 				
				$$\mathrm{rank}\left( \boldsymbol{GG}^\dagger \right)=
				\begin{array}{c}
					\begin{cases}
						q, & \text{if }  \varGamma_{1}=0;\\
						q+1,& \text{if } \varGamma_{1}\neq 0. 
					\end{cases}
				\end{array}  $$
				
				\textbf{case 3.} If $\frac{q+1}{2}\nmid 2i$, then by Lemma \ref{L2.8 }, we know that $e_{j}$, $e_{j+1}$, $e_{j+3}$ and $e_{j+4}$ are all non-zero. 
				
				Now, for convenience, we denote the vector $\boldsymbol{a}_{u_2,11}$ be the $u_2$-th row of the matrix $\boldsymbol{A}_{7\times 7}$ given by the Equation \eqref{M3.20}, where $1 \le u_2\le 7$. Next, for the matrix $\boldsymbol{A}_{7\times 7}$ given by  \eqref{M3.20}, we perform the following $2$ steps of elementary transformations.	
				
				\textbf{Step 1.} Replace $\boldsymbol{a}_{6,11}$ with $\boldsymbol{a}_{6,11}-\frac{e_{j+5}}{e_{j+4}}\boldsymbol{a}_{1,11}-b_{k-2,0}\frac{e_{j+3}}{e_{j+2}}\boldsymbol{a}_{3,11}-b_{k-2,0}\frac{e_{j+2}}{e_{j+1}}\boldsymbol{a}_{4,11}$, and replace $\boldsymbol{a}_{7,11}$ with $\boldsymbol{a}_{7,11}-\frac{e_{j+4}}{e_{j+3}}\boldsymbol{a}_{2,11}-b_{k-1,0}\frac{e_{j+3}}{e_{j+2}}\boldsymbol{a}_{3,11}-b_{k-1,1}\frac{e_{j+2}}{e_{j+1}}\boldsymbol{a}_{4,11}$, we obtain the matrix $\boldsymbol{A}_{4}^{(1)}$, and denote the vector $\boldsymbol{b}_{11,v_2}^{(1)}$ be the $v_2$-th column of $\boldsymbol{A}_{4}^{(1)}$, where $1 \le v_2\le 7$;
				
				\textbf{Step 2.} Replace $\boldsymbol{b}_{11,6}^{(1)}$ with $\boldsymbol{b}_{11,6}^{(1)}-\frac{e_{j-1}}{e_{j}}\boldsymbol{b}_{11,1}^{(1)}-b_{k-2,0}^q\frac{e_{j+1}}{e_{j+2}}\boldsymbol{b}_{11,3}^{(1)}-b_{k-2,1}^q\frac{e_{j+2}}{e_{j+3}}\boldsymbol{b}_{11,4}^{(1)}$, and replace $\boldsymbol{b}_{11,7}^{(1)}$ with $\boldsymbol{b}_{11,7}^{(1)}-\frac{e_{j}}{e_{j+1}}\boldsymbol{b}_{11,2}^{(1)}-b_{k-1,0}^q\frac{e_{j+1}}{e_{j+2}}\boldsymbol{b}_{11,3}^{(1)}-b_{k-1,1}^q\frac{e_{j+2}}{e_{j+3}}\boldsymbol{b}_{11,4}^{(1)}$, we obtain the matrix 
				$$
				\boldsymbol{A}_{7\times 7}^5= \left(\begin{matrix}
					&		&		&		&		e_{j+4}&		0	&		0	\\
					&		&		&		e_{j+3}&		&0&		0	\\
					&		&		e_{j+2}&		&		&			0&	0	\\
					&		e_{j+1}&		&		&		&	0	&	0\\
					e_j&		&		&		&		&	0	&	0	\\
					0	&	0&		0&	0	&0	&	C_1&	C_2	\\
					0	&		0&		0&		0&0	&	C_3&	C_4	\\
				\end{matrix}\right)\ , 
				$$
				where 
				$$
				\begin{aligned}
					&C_1=b_{k-2,0}^{q+1}\left(e_{j+2}-\frac{e_{j+3}}{e_{j+2}}e_{j+1}\right),\\ &C_2=b_{k-2,1}\left(e_{j+1}-\frac{e_{j+2}}{e_{j+1}}e_{j}\right) +b_{k-2,0}b_{k-1,0}^{q}\left(e_{j+2}-\frac{e_{j+3}}{e_{j+2}}e_{j+1}\right), \\
					&C_3=b_{k-2,0}^{q}b_{k-1,0}\left(e_{j+2}-\frac{e_{j+3}}{e_{j+2}}e_{j+1}\right)+b_{k-2,1}^{q}\left(e_{j+3}-\frac{e_{j+4}}{e_{j+3}}e_{j+2}\right)\\
				\end{aligned}
				$$
				and
				$$C_4=b_{k-1,1}\left(e_{j+1}-\frac{e_{j+2}}{e_{j+1}}e_{j}\right)+b_{k-1,0}^{q+1}\left(e_{j+2}-\frac{e_{j+3}}{e_{j+2}}e_{j+1}\right)+b_{k-1,1}^{q}\left(e_{j+3}-\frac{e_{j+4}}{e_{j+3}}e_{j+2}\right).$$
				Thus, we have 
				
				$$\mathrm{rank}\left( \boldsymbol{GG}^\dagger \right)=q-5+\mathrm{rank}\left(\boldsymbol{A}_{7\times 7}\right)=q-5+\mathrm{rank}(\boldsymbol{A}_{7\times 7}^5)=q+\mathrm{rank}\left(\boldsymbol{A}_{2\times 2}\right). $$
				Now by the proof of Theorem \ref{T3.1}, we know that  	
				
				$$
				mathrm{rank}\left( \boldsymbol{GG}^\dagger \right) = \begin{array}{c}
					\begin{cases}
						q,  & \text{if }  b_{k-2,0}=b_{k-2,1}=b_{k-1,0}=b_{k-1,1}=0\\
						&\quad \text{or }  b_{k-2,1}=b_{k-2,0}=0 , b_{k-1,1}\neq 0 \text{ and }  \varGamma_{2}=0;\\
						q+2,& \text{if }  b_{k-2,0}=0 \text{ and }  b_{k-2,1}\neq 0\\
						&\quad \text{or } b_{k-1,0}=b_{k-1,1}=0 \text{ and }  b_{k-2,0}, b_{k-2,1}\neq 0\\ 
						&\quad  \text{or }  b_{k-2,1}=0, b_{k-2,0},b_{k-1,1}\neq 0 \text{ and }  b_{k-1,1}^{q-1}\gamma^{(i-1)(q-1)}\neq -1\\
						&\quad  \text{or } b_{k-2,1},b_{k-2,0}\neq 0,\left(b_{k-1,0},b_{k-1,1}\right)\neq (0,0) \text{ and }  \varGamma\neq 0; \\
						q+1, &\text{otherwise}.
					\end{cases}
				\end{array} 
				$$ 
				
				Thus, by combining with the discussions of the above $3$ cases, we conclude that	
				
				$$
				\mathrm{rank}\left( \boldsymbol{GG}^\dagger \right) =\begin{array}{c}
					\begin{cases}
						q, &\text{if }  i=\frac{q+1}{2}\quad and \quad  b_{k-2,1}=b_{k-2,0}=0\\
						&\quad \text{or }  i=\frac{q+1}{4} \text{ or } \frac{3(q+1)}{4} , b_{k-1,1}=b_{k-1,0}=0 \text{ and }  \left(b_{k-2,0},b_{k-2,1}\right)\neq (0,0)\\
						&\quad \text{or } i=\frac{q+1}{4} \text{ or } \frac{3(q+1)}{4} , \left(b_{k-1,0},b_{k-1,1}\right)\neq (0,0) \text{ and }  \varGamma_{1}=0 \\
						&\quad \text{or }\frac{q+1}{2}\nmid 2i \text{ and }  b_{k-2,0}=b_{k-2,1}=b_{k-1,0}=b_{k-1,1}=0\\
						&\quad \text{or } \frac{q+1}{2}\nmid 2i , b_{k-2,1}=b_{k-2,0}=0, b_{k-1,1}\neq 0 \text{ and }  \varGamma_{2}=0;\\	
						q+2, &\text{if } i = \frac{q+1}{2}  \text{ and }  b_{k-2,1}\neq 0\\
						&\quad \text{or } \frac{q+1}{2}\nmid 2i , b_{k-2,0}=0 \text{ and } b_{k-2,1}\neq 0\\
						&\quad \text{or } \frac{q+1}{2}\nmid 2i, b_{k-1,0}=b_{k-1,1}=0 \text{ and }  b_{k-2,0}, b_{k-2,1}\neq 0\\ 
						&\quad \text{or } \frac{q+1}{2}\nmid 2i, b_{k-2,1}=0, b_{k-2,0},b_{k-1,1}\neq 0 \text{ and }  b_{k-1,1}^{q-1}\gamma^{(i-1)(q-1)}\neq -1\\
						&\quad \text{or }\frac{q+1}{2}\nmid 2i,  b_{k-2,1},b_{k-2,0}\neq 0,\left(b_{k-1,0},b_{k-1,1}\right)\neq (0,0)  \text{ and }  \varGamma\neq 0;	\\
						q+1,& \text{otherwise}. 
					\end{cases}
				\end{array} 
				$$ 				
				
				And so, by Lemma \ref{L2.16 }, we complete the proof of Theorem \ref{T3.3}
				\newpage
				
				\section{Some examples}	 
				\label{section 4}
				
				In this section, we present the corresponding illustrative examples for each case of Theorems \ref{T3.1}-\ref{T3.3} given in Section \ref{section 3}I, which are also checked by the Magma program. For convenience, we adopt the same case-wise classifications in Table I as those used in the proofs of Theorems \ref{T3.1}-\ref{T3.3}.
				\begin{table}[H]
					\centering
					\renewcommand{\arraystretch}{1.5}
					\caption{Some examples for Theorem \ref{T3.1}}			
					\begin{tabular}{|c|c|c|c|c|c|c|c|c|c|c|c|}
						\hline
						\multirow{2}{*}{$q$} & \multirow{2}{*}{$q+1$} & \multirow{2}{*}{$i$}
						& \multirow{2}{*}{$b_{k-2,0}$} & \multirow{2}{*}{$b_{k-2,1}$}
						& \multirow{2}{*}{$b_{k-1,0}$} & \multirow{2}{*}{$b_{k-1,1}$}
						&\multicolumn{2}{c|}{$\operatorname{rank}(\boldsymbol{G}\boldsymbol{G}^\dagger)$} &\multicolumn{2}{c|}{$\operatorname{dim}\left(\mathrm{Hull}_H(\mathcal{C}_k(\boldmath\alpha))\right)$}& \multirow{2}{*}{Our results} \\
						\cline{8-11}
						& & & & & & &Magma output&theoretical value&Magma output&theoretical value& \\
						\hline
						\multirow{12}{*}{$9$} & \multirow{12}{*}{$10$}& \multirow{12}{*}{$3$} & $0$ & $0$ & $0$ & $0$ & $9$ & $9$ & $3$&$3$& case 2.1  \\
						\cline{4-12}
						& & &$1$ & $0$ & $0$ & $0$ & $10$ & $10$ &$2$&$2$& case 2.2 (1) \\
						\cline{4-12}
						&  &  & $0$ & $1$ & $0$ & $0$ & $11$ & $11$ &$1$&$1$& case 2.2 (2) \\
						\cline{4-12}
						&  &  & $0$ & $0$ & $1$ & $0$ & $10$ & $10$ &$2$&$2$& case 2.2 (3) \\
						\cline{4-12}
						&  &  & $0$ & $0$ & $0$ & $\gamma^{3}$ & $9$ & $9$ &$3$&$3$& case 2.2 (4) \\
						\cline{4-12}
						&  &  & $0$ & $0$ & $0$ & $1$ & $10$ & $10$ &$2$&$2$& case 2.2 (4) \\
						\cline{4-12}
						&  &  & $1$ & $0$ & $0$ & $\gamma^3$ & $10$ & $10$ &$2$&$2$& case 2.3 (4) \\
						\cline{4-12}
						&  &  & $1$ & $0$ & $0$ & $2$ & $11$ & $11$ &$1$&$1$& case 2.3 (4) \\
						\cline{4-12}
						&  &  & $1$ & $0$ & $2$ & $0$ & $10$ & $10$ &$2$&$2$& case 2.3 (5) \\
						\cline{4-12}
						&  &  & $1$ & $0$ & $2$ & $\gamma^3$ & $10$ & $10$ &$2$&$2$& case 2.4 (2) \\
						\cline{4-12}
						&  &  & $1$ & $0$ & $2$ & $1$ & $11$ & $11$ &$1$&$1$& case 2.4 (2) \\
						\hline
						\multirow{3}{*}{$25$} & \multirow{3}{*}{$26$}& \multirow{3}{*}{$9$} & $0$ & $1$ & $0$ & $2$ & $27$ & $27$ &$9$&$9$& case 2.3 (2) \\
						\cline{4-12}
						&  &  & $0$ & $1$ & $2$ & $0$ & $27$ & $27$ &$9$&$9$& case 2.3 (3) \\
						\cline{4-12}
						&  &  & $0$ & $1$ & $2$ & $3$ & $27$ & $27$ &$9$&$9$& case 2.4 (1) \\
						\hline
						\multirow{10}{*}{$27$} & \multirow{10}{*}{$28$}& \multirow{6}{*}{$3$} & $1$ & $2$ & $0$ & $1$ & $28$ & $28$ &$11$&$11$& case 2.4 (3) \\
						\cline{4-12}
						&  &  & $1$ & $2$ & $0$ & $2$ & $29$ & $29$ &$10$&$10$& case 2.4 (3) \\
						\cline{4-12}
						&  &  & $1$ & $2$ & $1$ & $0$ & $28$ & $28$ & $11$&$11$&case 2.4 (4) \\
						\cline{4-12}
						&  &  & $1$ & $2$ & $2$ & $0$ & $29$ & $29$ &$10$&$10$& case 2.4 (4) \\
						\cline{4-12}
						&  &  & $1$ & $1$ & $\gamma^3$ & $\gamma^{17}$ & $28$ & $28$ &$11$&$11$& case 2.5 \\
						\cline{4-12}
						&  &  & $1$ & $2$ & $1$ & $2$ & $29$ & $29$ &$10$&$10$& case 2.5 \\
						\cline{3-12}
						&  & \multirow{3}{*}{$9$} & $0$ & $0$ & $\gamma$ & $\gamma^3$ & $27$ & $27$ &$12$&$12$& case 2.3 (1) \\
						\cline{4-12}
						&  &  & $0$ & $0$ & $1$ & $2$ & $28$ & $28$ &$11$&$11$& case 2.3 (1) \\
						\cline{4-12}
						&  &  & $1$ & $2$ & $0$ & $0$ & $29$ & $29$ &$10$&$10$& case 2.3 (6) \\
						\cline{3-12}
						&  & $27$ & $1$ & $2$ & $1$ & $2$ & $27$ & $27$ &$12$&$12$& case 1 \\
						\hline
					\end{tabular}
				\end{table}  
				\begin{table}[H]
					\centering
					\renewcommand{\arraystretch}{1.5}
					\caption{Some examples for Theorem \ref{T3.2}}
					\begin{tabular}{|c|c|c|c|c|c|c|c|c|c|c|c|}
						\hline
						\multirow{2}{*}{$q$} & \multirow{2}{*}{$q+1$} & \multirow{2}{*}{$i$}
						& \multirow{2}{*}{$b_{k-2,0}$} & \multirow{2}{*}{$b_{k-2,1}$}
						& \multirow{2}{*}{$b_{k-1,0}$} & \multirow{2}{*}{$b_{k-1,1}$}
						& \multicolumn{2}{c|}{$\operatorname{rank}(\boldsymbol{G}\boldsymbol{G}^\dagger)$} &\multicolumn{2}{c|}{$\operatorname{dim}\left(\mathrm{Hull}_H(\mathcal{C}_k(\boldmath\alpha))\right)$}& \multirow{2}{*}{Our results} \\
						\cline{8-11}
						& & & & & & &Magma output&theoretical value& Magma output&theoretical value& \\
						\hline
						\multirow{11}{*}{$27$} & \multirow{11}{*}{$28$}& \multirow{8}{*}{$6$} & $0$ & $0$ & $0$ & $\gamma^{23}$ & $26$ & $26$ &$13$&$13$& case 2 (1) \\
						\cline{4-12}
						&  &  & $0$ & $0$ & $0$ & $2$ & $27$ & $27$ &$12$&$12$& case 2 (1) \\
						\cline{4-12}
						&  &  & $0$ & $1$ & $0$ & $2$ & $28$ & $28$ &$11$&$11$& case 2 (2) \\
						\cline{4-12}
						&  &  & $0$ & $0$ & $1$ & $2$ & $28$ & $28$ &$11$&$11$& case 2 (3) \\
						\cline{4-12}
						&  &  & $0$ & $1$ & $2$ & $1$ & $28$ & $28$ &$11$&$11$& case 2 (4) \\
						\cline{4-12}
						&  &  & $1$ & $2$ & $1$ & $2$ & $28$ & $28$ &$11$&$11$& case 2 (5) \\
						\cline{4-12}
						&  &  & $1$ & $1$ & $1$ & $\gamma^{11}$ & $28$ & $28$ &$11$&$11$& case 2 (5) \\
						\cline{4-12}
						&  &  & $1$ & $1$ & $1$ & $2$ & $29$ & $29$ &$10$&$10$& case 2 (5) \\
						\cline{3-12}
						& & \multirow{3}{*}{$14$} & $1$ & $2$ & $0$ & $\gamma^{15}$ & $26$ & $26$ &$13$&$13$& case 1 (1) \\
						\cline{4-12}
						&  &  & $1$ & $2$ & $0$ & $1$ & $27$ & $27$ &$12$&$12$& case 1 (1) \\
						\cline{4-12}
						&  &  & $1$ & $2$ & $1$ & $2$ & $28$ & $28$ &$11$&$11$& case 1 (2) \\
						\hline
						
					\end{tabular}
				\end{table}
				\begin{table}[H]
					\centering
					\renewcommand{\arraystretch}{1.5}
					\caption{Some examples for Theorem \ref{T3.3}}
					\label{tab:examples_T3.3}
					\begin{tabular}{|c|c|c|c|c|c|c|c|c|c|c|c|}
						\hline
						\multirow{2}{*}{$q$} & \multirow{2}{*}{$q+1$} & \multirow{2}{*}{$i$}
						& \multirow{2}{*}{$b_{k-2,0}$} & \multirow{2}{*}{$b_{k-2,1}$}
						& \multirow{2}{*}{$b_{k-1,0}$} & \multirow{2}{*}{$b_{k-1,1}$}
						&\multicolumn{2}{c|}{$\operatorname{rank}(\boldsymbol{G}\boldsymbol{G}^\dagger)$} &\multicolumn{2}{c|}{$\operatorname{dim}\left(\mathrm{Hull}_H(\mathcal{C}_k(\boldmath\alpha))\right)$}& \multirow{2}{*}{Our results} \\
						\cline{8-11}
						& & & & & & &Magma output&theoretical value&Magma output&theoretical value& \\
						\hline
						\multirow{3}{*}{$9$} 
						& \multirow{3}{*}{$10$} 
						& \multirow{3}{*}{$5$} 
						& $0$ & $0$ & $1$ & $2$ & $9$ & $9$ &$3$&$3$ &case 1  \\
						\cline{4-12}
						&  &  
						& $1$ & $0$ & $2$ & $1$ & $10$ & $10$ &$2$&$2$& case 1  \\
						\cline{4-12}
						&  &  
						& $2$ & $1$ & $1$ & $2$ & $11$ & $11$ &$1$&$1$& case 1  \\
						\hline
						\multirow{3}{*}{$27$} 
						& \multirow{3}{*}{$28$} 
						& \multirow{3}{*}{$7$} 
						& $1$ & $2$ & $0$ & $0$ & $27$ & $27$&$12$&$12$ & case 2  \\
						\cline{4-12}
						&  &  
						& $1$ & $2$ & $\gamma^{17}$ & $1$ & $27$ & $27$ &$12$&$12$ & case 2  \\
						\cline{4-12}
						&  &  
						& $1$ & $2$ & $1$ & $2$ & $28$ & $28$ &$11$&$11$ & case 2 \\
						\hline 
						\multirow{22}{*}{$49$} & \multirow{22}{*}{$50$}& \multirow{22}{*}{$15$} & $1$ & $0$ & $0$ & $0$ & $50$ & $50$ &$22$&$22$ & case 3  \\
						\cline{4-12}
						&  &  & $0$ & $1$ & $0$ & $0$ & $51$ & $51$ &$21$&$21$ & case 3  \\
						\cline{4-12}
						&  &  & $0$ & $0$ & $1$ & $0$ & $50$ & $50$ &$22$&$22$ & case 3  \\
						\cline{4-12}
						&  &  & $0$ & $0$ & $0$ & $\gamma^{11}$ & $49$ & $49$ &$23$&$23$ & case 3  \\
						\cline{4-12}
						&  &  & $0$ & $0$ & $0$ & $1$ & $50$ & $50$ &$22$&$22$ & case 3  \\
						\cline{4-12}
						&  &  & $0$ & $0$ & $\gamma^2$ & $\gamma^{4}$ & $49$ & $49$ &$23$&$23$ & case 3 \\
						\cline{4-12}
						&  &  & $0$ & $0$ & $1$ & $2$ & $50$ & $50$ &$22$&$22$ & case 3  \\
						\cline{4-12}
						&  &  & $0$ & $1$ & $0$ & $2$ & $51$ & $51$ &$21$&$21$ & case 3  \\
						\cline{4-12}
						&  &  & $0$ & $1$ & $2$ & $0$ & $51$ & $51$ &$21$&$21$ & case 3  \\
						\cline{4-12}
						&  &  & $1$ & $0$ & $0$ & $\gamma^{11}$ & $50$ & $50$ &$22$&$22$ & case 3  \\
						\cline{4-12}
						&  &  & $1$ & $0$ & $0$ & $2$ & $51$ & $51$ &$21$&$21$ & case 3 \\
						\cline{4-12}
						&  &  & $1$ & $0$ & $2$ & $0$ & $50$ & $50$ &$22$&$22$ & case 3  \\
						\cline{4-12}
						&  &  & $1$ & $2$ & $0$ & $0$ & $51$ & $51$ &$21$&$21$ &  case 3  \\
						\cline{4-12}
						&  &  & $0$ & $1$ & $2$ & $1$ & $51$ & $51$ &$21$&$21$ &  case 3  \\
						\cline{4-12}
						&  &  & $1$ & $0$ & $2$ & $\gamma^{11}$ & $50$ & $50$ &$22$&$22$ & case 3  \\
						\cline{4-12}
						&  &  & $1$ & $0$ & $2$ & $1$ & $51$ & $51$ &$21$&$21$ & case 3  \\
						\cline{4-12}
						&  &  & $1$ & $\gamma^3$ & $0$ & $\gamma^{17}$ & $50$ & $50$ &$22$&$22$ & case 3  \\
						\cline{4-12}
						&  &  & $1$ & $2$ & $0$ & $1$ & $51$ & $51$ &$21$&$21$ & case 3  \\
						\cline{4-12}
						&  &  & $1$ & $\gamma$ & $\gamma^{24}$ & $0$ & $50$ & $50$ &$22$&$22$ & case 3  \\
						\cline{4-12}
						&  &  & $1$ & $2$ & $1$ & $0$ & $51$ & $51$ &$21$&$21$ & case 3 \\
						\cline{4-12}
						&  &  & $1$ & $1$ & $1$ & $\gamma^{20}$ & $50$ & $50$ &$22$&$22$ & case 3 \\
						\cline{4-12}
						&  &  & $1$ & $2$ & $1$ & $1$ & $51$ & $51$ &$21$&$21$ & case 3 \\
						\hline
					\end{tabular}
				\end{table}

				\section{Conclusions}
				\label{section 5}
				In this paper, by taking a special class of the vector  $\boldsymbol{\alpha}$, we obtain the following three main results about a class of $\left(\mathcal{L},\mathcal{P}\right)$-TGRS codes $\mathcal{C}_{q+j}\left(\boldsymbol{\alpha}\right)$. 
				\begin{itemize}
					\item Determine the value of the matrix $\boldsymbol{GG}^\dagger$ (Proposition \ref{P3.1}).
					
					\item Completely determine the Hermitian hull dimension of $\mathcal{C}_{q+j}\left(\boldsymbol{\alpha}\right)$ for three cases(Theorems \ref{T3.1}-\ref{T3.3}). 
					
					\item Construct two classes of EAQECCs.
				\end{itemize}

				\newpage
				\appendices
				\section{Proofs of Several Lemmas}
				\label{appendex A}
				
				\textbf{The proof of Lemma \ref{L2.6 1} } 
				
				By $2\mid i$ and Lemma \ref{L2.5 }, we have $e_{j+2}=0$. 
				
				If $e_{j}=e_{j+2}=0$, then by Lemma \ref{L2.3 }, we have $k_{j+2}\in \mathbb{Z}$ and $k_{j}\in \mathbb{Z}$, i.e., $\frac{(j+2)h}{q+1}\in\mathbb{Z}$ and $\frac{jh}{q+1}\in\mathbb{Z}$, it means  $\frac{2h}{q+1}\in\mathbb{Z}$, i.e., $q+1\mid 2h$. Note that $2\mid q+1$, thus, $\frac{q+1}{2}\mid h$. Now by $h\mid i$, we have $\frac{q+1}{2}\mid i$. And so, by $2\leq i\leq q$, we have $i=\frac{q+1}{2}$. 
				
				Conversely, if $i=\frac{q+1}{2}$, then $h=\frac{q+1}{2}$, thus, from $1\le r \le q-2$, we can get   $$k_{r+2}=\frac{(r+2)h}{q+1}=\frac{2h}{q+1}+k_r=1+k_r,$$
				and so, $k_{r+2}$ and $k_r$ are both integers or not simultaneously. Furthermore, by Lemma \ref{L2.3 }, we know that $e_{r}$ and $e_{r+2}$ are both zeros or not simultaneously. Thus, by $e_{j+2}=0$, we have $e_j=e_{j+4}=0$. 
				
				
				$	\quad $
				
				\textbf{The proof of Lemma \ref{L2.6 } } 
				
				$(1)\Longrightarrow(2)$  On the one hand, by  $e_{j+1}=e_{j+3}=0$ and Lemma \ref{L2.4 }, we have $e_{j+2}\neq 0$, thus, by Lemma \ref{L2.5 }, we have $2\nmid i$. On the other hand, by $e_{j+1}=e_{j+3}=0$ and Lemma \ref{L2.3 }, we have $\frac{(j+1)h}{q+1}\in \mathbb{Z}$ and $\frac{(j+3)h}{q+1}\in\mathbb{Z}$, it means $\frac{2h}{q+1}\in\mathbb{Z}$, i.e., $q+1\mid2h$. Note that  $2\mid q+1$, thus, $\frac{q+1}{2}\mid h$. Now by $h\mid i$, we have  $\frac{q+1}{2}\mid i$.
				And so, by $2\leq i\leq q$, we can get  $i=\frac{q+1}{2}$.
				
				$(2)\Longrightarrow(3)$ On the one hand, by $2 \nmid i$ and $i = \frac{q+1}{2}$, we have $2 \nmid \frac{q+1}{2}$, it means $q \equiv 1 \pmod 4$. On the other hand, by $i = \frac{q+1}{2}$, it's easy to get $\frac{q+1}{2}\mid i$. 
				
				$(3)\Longrightarrow(1)$ On the one hand, by $\frac{q+1}{2}\mid i$ and $2\le i\le q$, we have $i=\frac{q+1}{2}$, thus, $h=\frac{q+1}{2}$.  By Lemma \ref{L2.3 }, we know that $e_{r}=0$ if and only if $k_{r}=\frac{r}{2}\in \mathbb{Z}$, i.e., $2\mid r$, it means that exactly one of  $e_{j+1}$ and $e_{j+2}$ equals zero. 
				On the other hand, by $q \equiv 1 \pmod 4$, we have $4\nmid q+1$. Note that  $2\mid q+1$, thus, $2\nmid \frac{q+1}{2}$. By combining with  $i=\frac{q+1}{2}$, it follows that $2\nmid i$. Now by Lemma \ref{L2.5 }, we have $e_{j+2}\neq 0$, thus, $e_{j+1}=0$. Furthermore, by Lemma \ref{L2.4 }, we can get  $e_{j+1}=e_{j+3}=0$. 
				
				
				
				
				\quad
				
				\textbf{The proof of Lemma \ref{L2.7 } } 
				
				By $e_{j+2}\neq 0$, Lemma \ref{L2.3 } and Lemma \ref{L2.5 }, we have $2\nmid i$ and $k_{j+2}\notin \mathbb{Z}$, i.e., $2\nmid i$ and $\frac{(j+2)h}{q+1}\notin\mathbb{Z}$. 
				
				$(1)\Longrightarrow(2)$ By  $e_{j+4}=e_{j}=0$ and Lemma \ref{L2.3 }, we have $\frac{(j+4)h}{q+1}\in\mathbb{Z}$ and $\frac{jh}{q+1}\in\mathbb{Z}$, then by combining with $\frac{(j+2)h}{q+1}\notin\mathbb{Z}$, it follows that $\frac{4h}{q+1}\in\mathbb{Z}$ and $\frac{2h}{q+1}\notin\mathbb{Z}$, i.e., $q+1\mid 4h$ and $q+1\nmid 2h$. Note that $2\mid q+1$, thus,  $\frac{q+1}{2}\mid 2h$ and $\frac{q+1}{2}\nmid h$, namely,  $2\mid\frac{q+1}{2}$, i.e., $4\mid q+1$, and so, $q\equiv-1\pmod 4$. Hence, $\frac{q+1}{4}\mid h$ and $\frac{q+1}{2}\nmid h$. Now by $h\mid i$, we have $\frac{q+1}{4}\mid i$ and $\frac{q+1}{2}\nmid i$. And so, by $2\le i\le q$, we can get $i=\frac{q+1}{4}$ or $\frac{3(q+1)}{4}$. 
				
				$(2)\Longrightarrow(3)$ By $i=\frac{q+1}{4}$ or $\frac{3(q+1)}{4}$, it is easy to get $\frac{q+1}{4}\mid i$.   
				
				$(3)\Longrightarrow(1)$ By $q\equiv-1\pmod 4$, we have $4\mid q+1$, i.e, $2\mid \frac{q+1}{2}$, then by combining with $2\nmid i$, it follows that $i\neq \frac{q+1}{2}$. While by  $\frac{q+1}{4}\mid i$ and $2\le i\le q$, we have $i=\frac{q+1}{4}$ or $\frac{3(q+1)}{4}$, thus, $h=\frac{q+1}{4}$. By Lemma \ref{L2.3 }, we know that $e_{r}=0$ if and only if $k_{r}=\frac{r}{4}\in \mathbb{Z}$, i.e.,  $4\mid r$, it means that only one of $e_{j}$, $e_{j+1}$, $e_{j+2}$, $e_{j+3}$ equals zero. And by Lemma \ref{L2.4 }, we know that $e_{j+1}$ and $e_{j+3}$ are both zeros or not simultaneously, thus, $e_{j+1}$ and $e_{j+3}$ are both non-zeros. By combining with $e_{j+2}\neq 0$, we have $e_{j}=0$. Furthermore, by Lemma \ref{L2.4 }, it's easy to get $e_{j}=e_{j+4}=0$. 
				

				
				
				

				\quad
				
				\textbf{The proof of Lemma \ref{L2.8 } } 
				
				For any odd prime $p$ and $q=p^m$, it's easy to know that $q\equiv1\pmod 4$ or $q\equiv-1\pmod 4$. 
				
				Then, by Lemma \ref{L2.5 } and Lemmas \ref{L2.6 }-\ref{L2.7 }, the following three statements are true,
				
				$(1)$ $e_{j+2}\neq 0$ if and only if $2\nmid i$; 
				
				$(2)$ Both $e_{j+1}\neq 0$ and $e_{j+3}\neq 0$ if and only if $q\equiv -1\pmod 4$ or $q\equiv 1\pmod 4$ and $\frac{q+1}{2}\nmid i$. 
				
				$(3)$ If $e_{j+2}\neq 0$, then both $e_{j}\neq 0$ and $e_{j+4}\neq 0$ if and only if  $q\equiv 1\pmod 4$ or $q\equiv -1\pmod 4$ and $\frac{q+1}{4}\nmid i$.
				
				Thus, $e_{r}\neq 0$ $\left(r=j, j+1, j+2, j+3, j+4\right)$ if and only if both $2\nmid i$ and $q\equiv 1\pmod 4$ with $\frac{q+1}{2}\nmid i$ or $q\equiv -1\pmod 4$ with $\frac{q+1}{4}\nmid i$. If $q\equiv 1\pmod 4$, then $2\nmid \frac{q+1}{2}$, i.e., $\gcd\left(2,\frac{q+1}{2}\right)=1$, furthermore, $\frac{q+1}{2}\nmid 2i$ if and only if $\frac{q+1}{2}\nmid i$. If $q\equiv -1\pmod 4$, then $4\mid q+1$, i.e., $2\mid \frac{q+1}{2}$, furthermore,  $\frac{q+1}{2}\nmid 2i$ if and only if $\frac{q+1}{4}\nmid i$.

				\quad
				
				\textbf{The proof of Lemma \ref{L2.9 } }
				
				By Corollary \ref{C2.1}, we have $e_1$, $e_q \neq 0$. 
				If $i\equiv 0\pmod p$, then   $i(q-1)=0$, furthermore, $\mathrm{rank}(A)=2$. 
			    If $ i^2 \equiv 1\pmod p$, then  $p\nmid i$, furthermore, $2\le i\le q-1$ and $i(q-1)\neq 0$, thus, $\mathrm{rank}(A)=1$ or $2$. Note that  $\mathrm{rank}(A)=1$ if and only if $\left[i(q-1)\right]^2=e_1e_q$, i.e., 
				\begin{equation}
					\label{E2.1}
					i^2(q-1)^2=(q-1)^2\left(\sum\limits_{t=0}^{i-1}{ \gamma ^{t\left( q-1 \right)}}\right)\left(\sum\limits_{t=0}^{i-1}{ \gamma ^{qt\left( q-1 \right)}}\right). 
				\end{equation}
				By $ i^2 \equiv 1\pmod p$ and $q-1\not\equiv 0\pmod p$, we know that the Equation \eqref{E2.1} is equivalent to 
				\begin{equation}
					\label{E2.2}
					\left(\sum\limits_{t=0}^{i-1}{ \gamma ^{t\left( q-1 \right)}}\right)\left(\sum\limits_{t=0}^{i-1}{ \gamma ^{qt\left( q-1 \right)}}\right)=1.
				\end{equation}
				Note that $\gamma^{q-1}-1$ and $\gamma^{q(q-1)}-1$ are both non-zero, and so, by multiplying $\left(\gamma^{q-1}-1\right)\left(\gamma^{q(q-1)}-1\right) $ on both sides of the Equation \eqref{E2.2}, we have 
				\begin{equation}
					\label{E2.3}
					\left(\gamma^{i(q-1)}-1\right)\left(\gamma^{iq(q-1)}-1\right)=\left(\gamma^{q-1}-1\right)\left(\gamma^{q(q-1)}-1\right). 
				\end{equation}
				Next, by expanding and simplifying the Equation \eqref{E2.3}, we have 
				\begin{equation}
					\label{E2.4}		\gamma^{q-1}-\gamma^{i(q-1)}=\gamma^{iq(q-1)}-\gamma^{q(q-1)},  
				\end{equation}
				i.e., 
				\begin{equation}
					\label{E2.5}		\gamma^{q-1}-\gamma^{i(q-1)}=\gamma^{i(1-q)}-\gamma^{(1-q)}.
				\end{equation}			
				And so, the Equation \eqref{E2.5} is equivalent to 
				\begin{equation}
					\label{E2.6}
					\gamma^{q-1}\left(\gamma^{(i+1)(1-q)}-1\right)\left(1-\gamma^{(i-1)(q-1)}\right)=0.
				\end{equation}
				Note that $\gamma^{q-1}\in \mathbb{F}_{q^2}^{*}$, thus,  $\gamma^{(i+1)(1-q)}=1$ or $\gamma^{(i-1)(q-1)}=1$, i.e., 
				$$(i+1)(1-q)\equiv 0\pmod {q^2-1} \text{ or } (i-1)(q-1)\equiv 0\pmod {q^2-1},$$ 
				namely, 
				\begin{equation}
					\label{E2.6 }
					i+1\equiv 0\pmod {q+1} \text{ or } i-1\equiv 0\pmod {q+1},
				\end{equation}
				which contradicts the assumption $ 2 \le i\le q-1$. Thus $\mathrm{rank}(A)\neq 1$, it means  $\mathrm{rank}(A)=2$.			 
				
				
				\quad
				
				\textbf{The proof of Lemma \ref{L2.10 } } 
				
				\textbf{(1)}	
				By  $2\nmid i$ and Lemma \ref{L2.5 }, we have $e_{j+2}=q-1$ and 
				$\gamma^{(j+2)(q-1)} = -1$, thus, $e_{j+2}^2=e_{j+3}e_{j+1}$ if and only if  $(q-1)^2=e_{j+3}e_{j+1}$, i.e. 
				\begin{equation}
					\label{E2.11}
					( q-1 )^2 =( q-1 )^2 \left( \sum_{t=0}^{i-1}{ \gamma ^{\left( j+3 \right) t\left( q-1 \right)}} \right) \left( \sum\limits_{t=0}^{i-1}{\gamma ^{\left( j+1 \right) t\left( q-1 \right)}} \right) . 
				\end{equation}
				Now by $\gamma^{(j+2)(q-1)} = -1$, we have  $\gamma^{(j+1)(q-1)}=-\gamma^{1-q} \neq 1$  and $\gamma^{(j+3)(q-1)}=-\gamma^{q-1}\neq 1 $, i.e., $\gamma^{(j+1)(q-1)}-1$ and $\gamma^{(j+3)(q-1)}-1$ are both non-zero, thus, by multiplying $\left(\gamma^{(j+1)(q-1)} - 1 \right) \left(\gamma^{(j+3)(q-1)} - 1\right)$ on both sides of the Equation \eqref{E2.11}, we have 
				\begin{equation}
					\left(\gamma^{(j+1)(q-1)} - 1 \right) \left(\gamma^{(j+3)(q-1)} - 1\right)=\left(\gamma^{i(j+1)(q-1)} - 1 \right) \left(\gamma^{i(j+3)(q-1)} - 1\right), 
				\end{equation}
				i.e., 
				\begin{equation}
					\label{E2.13}
					\left(1+\gamma^{(1-q)} \right) \left(1+\gamma^{(q-1)} \right)=\left(1+\gamma^{i(1-q)}  \right) \left(1+\gamma^{i(q-1)} \right). 
				\end{equation}
				Next, by expanding and simplifying the Equation \eqref{E2.13}, we have
				$\gamma ^{\left( 1-q \right)}+\gamma ^{\left( q-1 \right)}=\gamma ^{i\left( 1-q \right)}+\gamma ^{i\left( q-1 \right)}$,  
				i.e., 
				\begin{equation}
					\label{E2.14}
					\gamma^{(1-q)}\left(\gamma^{(i+1)(q-1)}-1\right) \left( \gamma ^{\left( i-1 \right) \left( 1-q \right)}-1 \right)  =0. 
				\end{equation}
				Note that $\gamma^{1-q}\in \mathbb{F}_{q^2}^{*}$, and so,   $\gamma^{(i+1)(q-1)}=1$ or $\gamma^{(i-1)(1-q)}=1$, i.e., 
				$$(i+1)(q-1)\equiv 0\pmod {q^2-1} \text{ or } (i-1)(1-q)\equiv 0\pmod {q^2-1},$$ 
				namely, 
				\begin{equation}
					\label{E2.14 }
					i+1\equiv 0\pmod {q+1} \text{ or } i-1\equiv 0\pmod {q+1}.
				\end{equation}
				Thus, by $2\leq i\leq q$, we know that the Equation \eqref{E2.14 } holds if and only if $i=q$.

				\textbf{(2)}  
				By $2\nmid i$, $\frac{q+1}{2}\nmid 2i$, Lemma \ref{L2.5 } and Lemma \ref{L2.8 }, we have $e_{j+2}=q-1$ and $e_{r}\in \mathbb{F}_{q^2}^{*}$ for $r\in \left\{j,j+1,j+3,j+4\right\}$, thus, $e_{j+1}^2=e_{j+2}e_{j}$ if and only if $e_{j+1}^2=(q-1)e_{j}$, i.e. 
				\begin{equation}
					\label{E2.7}
					\left( q-1 \right)^2\left( \sum_{t=0}^{i-1}{ \gamma ^{\left( j+1 \right) t\left( q-1 \right)}} \right) ^2=(q-1)^2 \left(  \sum_{t=0}^{i-1}{\gamma ^{jt\left( q-1 \right)}} \right) . 
				\end{equation}
				Note that  $\gamma^{j(q-1)}-1\in \mathbb{F}_{q^2}^{*}$ and  $\gamma^{(j+1)(q-1)}-1\in \mathbb{F}_{q^2}^{*}$, and so, by multiplying $\left(\gamma^{j(q-1)} - 1 \right) \left(\gamma^{(j+1)(q-1)} - 1\right)^2$ on both sides of the Equation \eqref{E2.7}, we have 
				\begin{equation}
					\left( \gamma ^{i(j+1)\left( q-1 \right)}-1 \right) ^2\left( \gamma ^{j\left( q-1 \right)} -1\right) =\left( \gamma ^{ (j+1)(q-1) }-1 \right) ^2\left( \gamma ^{ij\left( q-1 \right)}-1 \right) , 
				\end{equation}
				i.e., 
				\begin{equation}
					\label{E2.9}
					\left( 1+\gamma ^{i\left( 1-q \right)} \right) ^2\left( 1+\gamma ^{2\left( 1-q \right)} \right) =\left( 1+\gamma ^{\left( 1-q \right)} \right) ^2\left( 1+\gamma ^{2i\left( 1-q \right)} \right) .
				\end{equation} 
				Next, by expanding and simplifying the Equation \eqref{E2.9}, we have $\gamma ^{i\left( 1-q \right)}+\gamma ^{\left( i+2 \right) \left( 1-q \right)}=\gamma ^{\left( 1-q \right)}+\gamma ^{\left( 2i+1 \right) \left( 1-q \right)}$,  
				i.e., 
				\begin{equation}
					\label{E2.16 }
					\gamma ^{\left( 1-q \right)}\left( \gamma ^{\left( i-1 \right) \left( 1-q \right)}-1 \right) \left( \gamma ^{\left( i+1 \right) \left( 1-q \right)}-1 \right)=0. 
				\end{equation}
				Now by $\gamma^{1-q}\in \mathbb{F}_{q^2}^{*}$, we have  $\gamma^{(i-1)(1-q)}=1$ or $\gamma^{(i+1)(1-q)}=1$, i.e.,  
				$$(i-1)(1-q)\equiv 0\pmod {q^2-1} \text{ or } (i+1)(1-q)\equiv 0\pmod {q^2-1},$$ 
				namely, 
				\begin{equation}
					\label{E2.17 }
					i-1\equiv 0\pmod {q+1} \text{ or } i+1\equiv 0\pmod {q+1}.
				\end{equation}
				Hence, by $2\leq i\leq q$, we know that the Equation \eqref{E2.17 } holds if and only if $i=q$.
				
				In the similar proof as that of $(2)$, Lemma \ref{L2.10 } (3) holds too.
				
				
				\quad
				
				\textbf{The proof of Lemma \ref{L2.11 } } 
				
				By $ 2\nmid i$, $\frac{q+1}{2}\nmid 2i$, Lemma \ref{L2.5 } and Lemma \ref{L2.8 }, we have $e_{j+2}=q-1$ and $e_{r}\in \mathbb{F}_{q^2}^{*}$ for $r\in \left\{j,j+1,j+3,j+4\right\}$, thus, the  Equation \eqref{E2.15} holds if and only if 
				\begin{equation}
					\label{E2.16}
					b_{k-1,1}\left(e_{j+1}^2-e_{j+2}e_{j}\right)e_{j+2}e_{j+3}+b_{k-1,1}^{q}\left(e_{j+3}^2-e_{j+4}e_{j+2}\right)e_{j+1}e_{j+2}+b_{k-1,0}^{q+1}\left(e_{j+2}^2-e_{j+1}e_{j+3}\right)e_{j+1}e_{j+3}=0. 
				\end{equation}
				Now by the proof of Lemma \ref{L2.10 }, we know that for $m\in \left\{0,1,3,4\right\}$,  $\gamma^{(j+m)(q-1)}=-\gamma^{(m-2)(q-1)}\neq 1$ and $\gamma^{i(j+m)(q-1)}=-\gamma^{i(m-2)(q-1)}$, thus 
				\begin{equation}
					\label{A 19}
					e_{j+m}=\sum\limits_{t=0}^{i-1}{\left( q-1 \right) \gamma ^{t(j+m)\left( q-1 \right)}}=\left( q-1 \right) \frac{1-\gamma ^{i(j+m)\left( q-1 \right)}}{1-\gamma ^{(j+m)\left( q-1 \right)}}=(q-1)\frac{1+\gamma^{i(m-2)(q-1)}}{1+\gamma^{(m-2)(q-1)}}.
				\end{equation}
				Hence, the Equation \eqref{E2.16} is equivalent to 
				\begin{equation} \label{E2.17}
					\begin{aligned} 
						&b_{k-1,1}\left[\left(\frac{1+\gamma ^{i\left( 1-q \right)}}{1+\gamma ^{1-q}}\right)^2-\frac{1+\gamma ^{2i\left( 1-q \right)}}{1+\gamma ^{2\left( 1-q \right)}}\right]\cdot\frac{1+\gamma ^{i\left( q-1 \right)}}{1+\gamma ^{ q-1 }}
						+b_{k-1,1}^{q}\left[\left(\frac{1+\gamma ^{i\left( q-1 \right)}}{1+\gamma ^{ q-1}}\right)^2-\frac{1+\gamma ^{2i\left( q-1 \right)}}{1+\gamma ^{2\left( q-1 \right)}}\right]\cdot\frac{1+\gamma ^{i\left( 1-q \right)}}{1+\gamma ^{1-q }}\\
						=&-b_{k-1,0}^{q+1}\left[1-\frac{1+\gamma^{i(1-q)}}{1+\gamma^{1-q}}\cdot \frac{1+\gamma^{i(q-1)}}{1+\gamma^{q-1}}\right]\cdot \frac{1+\gamma^{i(1-q)}}{1+\gamma^{1-q}}\cdot \frac{1+\gamma^{i(q-1)}}{1+\gamma^{q-1}}. 
					\end{aligned}
				\end{equation}
				Now by multiplying $\left(1+\gamma ^{1-q}\right)^2\left(1+\gamma ^{2\left(1-q\right)}\right)\left(1+\gamma ^{\left(q-1\right)}\right)^2\left(1+\gamma ^{2\left(q-1\right)}\right)$ on both sides of the Equation \eqref{E2.17}, we have
				\begin{equation} \label{E2.18}
					\begin{aligned}
						&b_{k-1,1}\left[\left(1+\gamma ^{i\left(1-q\right)}\right)^2\left(1+\gamma ^{2\left(1-q \right)}\right)-\left(1+\gamma ^{2i\left(1-q \right)}\right)\left(1+\gamma ^{1-q}\right)^2\right]\cdot\left(1+\gamma ^{i\left(q-1\right)}\right)\left(1+\gamma ^{q-1}\right)\left(1+\gamma ^{2\left(q-1 \right)}\right)\\
						+&b_{k-1,1}^{q}\left[\left(1+\gamma ^{i\left(q-1\right)}\right)^2\left(1+\gamma ^{2\left(q-1 \right)}\right)-\left(1+\gamma ^{2i\left(q-1 \right)}\right)\left(1+\gamma ^{q-1}\right)^2\right]\cdot\left(1+\gamma ^{i\left(1-q\right)}\right)\left(1+\gamma ^{1-q}\right)\left(1+\gamma ^{2\left(1-q \right)}\right)\\
						=&b_{k-1,0}^{q+1}\left[\left(1+\gamma^{i(q-1)}\right)\left(1+\gamma^{i(1-q)}\right)-\left(1+\gamma^{1-q}\right)\left(1+\gamma^{q-1}\right)\right] \left(1+\gamma^{i(1-q)}\right)\left(1+\gamma^{i(q-1)}\right)\left(1+\gamma ^{2\left(1-q \right)}\right)\left(1+\gamma ^{2\left(q-1 \right)}\right).
					\end{aligned}
				\end{equation}
				Next, by expanding and simplifying the Equation \eqref{E2.18}, we have
				\begin{equation}\label{E2.20}
					\begin{aligned}
						&2b_{k-1,1}\left[-\gamma^{1-q}\left(\gamma^{(i-1)(1-q)}-1\right)\left(\gamma^{(i+1)(1-q)}-1\right)\right]\cdot\left(1+\gamma ^{i\left(q-1\right)}\right)\left(1+\gamma ^{q-1}\right)\left(1+\gamma ^{2\left(q-1 \right)}\right)\\
						+&2b_{k-1,1}^{q}\left[-\gamma^{q-1}\left(\gamma^{(i-1)(q-1)}-1\right)\left(\gamma^{(i+1)(q-1)}-1\right)\right]\cdot\left(1+\gamma ^{i\left(1-q\right)}\right)\left(1+\gamma ^{1-q}\right)\left(1+\gamma ^{2\left(1-q \right)}\right)\\
						=&b_{k-1,0}^{q+1}\left[\gamma^{q-1}\left(\gamma^{(i-1)(q-1)}-1\right)\left(1-\gamma^{(i+1)(1-q)}\right)\right]\cdot \left(1+\gamma^{i(1-q)}\right)\left(1+\gamma^{i(q-1)}\right)\left(1+\gamma ^{2\left(1-q \right)}\right)\left(1+\gamma ^{2\left(q-1 \right)}\right).
					\end{aligned}
				\end{equation}
				Now by multiplying $\gamma ^{(2i+1)\left(q-1\right)}$ on both sides of the Equation \eqref{E2.20}, we have
				\begin{equation}\label{E2.21}
					\begin{aligned}
						&2b_{k-1,1}\left[\left(1-\gamma^{(i-1)(q-1)}\right)\left(1-\gamma^{(i+1)(q-1)}\right)\right]\cdot\left(1+\gamma ^{i\left(q-1\right)}\right)\left(1+\gamma ^{q-1}\right)\left(1+\gamma ^{2\left(q-1 \right)}\right)\\
						+&2b_{k-1,1}^{q}\left[\gamma^{(i-1)(q-1)}\left(\gamma^{(i-1)(q-1)}-1\right)\left(\gamma^{(i+1)(q-1)}-1\right)\right]\cdot\left(1+\gamma ^{i\left(q-1\right)}\right)\left(1+\gamma ^{q-1}\right)\left(1+\gamma ^{2\left(q-1 \right)}\right)\\
						=&-b_{k-1,0}^{q+1}\left[\left(\gamma^{(i-1)(q-1)}-1\right)\left(\gamma^{(i+1)(q-1)}-1\right)\right]\cdot \left(1+\gamma^{i(q-1)}\right)\left(1+\gamma^{i(q-1)}\right)\left(1+\gamma ^{2\left(q-1 \right)}\right)\left(1+\gamma ^{2\left(q-1 \right)}\right)\gamma^{1-q}. 
					\end{aligned}
				\end{equation}
				Note that $\frac{q+1}{2}\nmid 2i$, then we have $\gamma^{i(q-1)}+1\in \mathbb{F}_{q^2}^{*}$. Thus, the Equation \eqref{E2.21} holds if and only if 
				$$2b_{k-1,1}\left(1+\gamma^{q-1}\right)\left(1+b_{k-1,1}^{q-1}\gamma^{(i-1)(q-1)}\right)=-b_{k-1,0}^{q+1}\left(1+\gamma^{2(q-1)}\right)\left(1+\gamma^{i(q-1)}\right)\gamma^{1-q},$$
				namely, $$2b_{k-1,1}\left(1+\gamma^{1-q}\right)\left(1+b_{k-1,1}^{q-1}\gamma^{(i-1)(q-1)}\right)=-b_{k-1,0}^{q+1}\left(1+\gamma^{2(1-q)}\right)\left(1+\gamma^{i(q-1)}\right),$$
				i.e., $\varGamma_{2}=0$. 
				
				
				\quad
				
				\textbf{The proof of Lemma \ref{L2.12 }  } 
				
				By multiplying $b_{k-2,1}b_{k-2,0}^{q+1}\left(e_{j+1}-\frac{e_{j+2}}{e_{j+1}}e_{j}\right)\left(e_{j+2}-\frac{e_{j+3}}{e_{j+2}}e_{j+1}\right)$ on both sides of the Equation \eqref{E2.22},  we have
				\begin{equation}\label{E2.23}
					\begin{aligned}
						&\left[\Delta\cdot\left(e_{j+1}-\frac{e_{j+2}}{e_{j+1}}e_{j}\right)
						+\Delta^q\cdot
						\left(e_{j+3}-\frac{e_{j+4}}{e_{j+3}}e_{j+2}\right)\right] b_{k-2,0}^{q+1}\left(e_{j+2}-\frac{e_{j+3}}{e_{j+2}}e_{j+1}\right)\\
						=&b_{k-2,1}^{q+1}\left(e_{j+3}-\frac{e_{j+4}}{e_{j+3}}e_{j+2}\right)\left(e_{j+1}-\frac{e_{j+2}}{e_{j+1}}e_{j}\right),
					\end{aligned}
				\end{equation}
				where $\Delta=b_{k-1,1}-\frac{b_{k-2,1}b_{k-1,0}}{b_{k-2,0}}$. 
				Note that $ 2\nmid i$, $\frac{q+1}{2}\nmid 2i$, Lemma \ref{L2.5 } and Lemma \ref{L2.8 }, we have $e_{j+2}=q-1$ and $e_{r}\in \mathbb{F}_{q^2}^{*}$ for $r\in \left\{j,j+1,j+3,j+4\right\}$, and then the Equation \eqref{E2.23} holds if and only if 
				\begin{equation}\label{E2.24}
					\begin{aligned}
						&\left[\Delta\cdot\left(e_{j+1}^2-e_{j}e_{j+2}\right)e_{j+3}
						+\Delta^q\cdot
						\left(e_{j+3}^{2}-e_{j+2}e_{j+4}\right)e_{j+1}\right] b_{k-2,0}^{q+1}\left(e_{j+2}^{2}-e_{j+1}e_{j+3}\right)\\
						=&b_{k-2,1}^{q+1}\left(e_{j+3}^2-e_{j+2}e_{j+4}\right)\left(e_{j+1}^2-e_{j}e_{j+2}\right)e_{j+2}.
					\end{aligned}
				\end{equation}
				where $\Delta=b_{k-1,1}-\frac{b_{k-2,1}b_{k-1,0}}{b_{k-2,0}}$.
				
				Now by the proof of Lemma \ref{L2.11 }, we know that for $m\in \left\{0,1,3,4\right\}$,  $\gamma^{(j+m)(q-1)}=-\gamma^{(m-2)(q-1)}\neq 1$
				and
				$ \gamma^{i(j+m)(q-1)}=-\gamma^{i(m-2)(q-1)}$,  
				thus  
				$$e_{j+m}=\sum\limits_{t=0}^{i-1}{\left( q-1 \right) \gamma ^{t(j+m)\left( q-1 \right)}}=\left( q-1 \right) \frac{1-\gamma ^{i(j+m)\left( q-1 \right)}}{1-\gamma ^{(j+m)\left( q-1 \right)}}=(q-1)\frac{1+\gamma^{i(m-2)(q-1)}}{1+\gamma^{(m-2)(q-1)}}.$$
				Hence, the Equation \eqref{E2.24} is equivalent to 
				\begin{equation}\label{E2.25}
					\begin{aligned}
						&\left\{\Delta\cdot\left[\left(\frac{1+\gamma ^{i\left( 1-q \right)}}{1+\gamma ^{\left( 1-q \right)}}\right)^2-\frac{1+\gamma ^{2i\left( 1-q \right)}}{1+\gamma ^{2\left( 1-q \right)}}\right]\frac{1+\gamma ^{i\left( q-1 \right)}}{1+\gamma ^{\left( q-1 \right)}}\right.\\
						&\left.+\Delta^q\cdot\left[\left(\frac{1+\gamma ^{i\left( q-1 \right)}}{1+\gamma ^{\left( q-1 \right)}}\right)^2-\frac{1+\gamma ^{2i\left( q-1 \right)}}{1+\gamma ^{2\left( q-1 \right)}}\right]\frac{1+\gamma ^{i\left( 1-q \right)}}{1+\gamma ^{\left( 1-q \right)}}\right\}\cdot b_{k-2,0}^{q+1}\left[1-\frac{1+\gamma ^{i\left( 1-q \right)}}{1+\gamma ^{\left( 1-q \right)}}\cdot \frac{1+\gamma ^{i\left( q-1 \right)}}{1+\gamma ^{\left( q-1 \right)}}\right]\\
						=&b_{k-2,1}^{q+1}\left[\left(\frac{1+\gamma ^{i\left( q-1 \right)}}{1+\gamma ^{\left( q-1 \right)}}\right)^2-\frac{1+\gamma ^{2i\left( q-1 \right)}}{1+\gamma ^{2\left( q-1 \right)}}\right]\cdot\left[\left(\frac{1+\gamma ^{i\left( 1-q \right)}}{1+\gamma ^{\left(1-q \right)}}\right)^2-\frac{1+\gamma ^{2i\left( 1-q \right)}}{1+\gamma ^{2\left( 1-q \right)}}\right]
					\end{aligned}
				\end{equation}
				Now by multiplying $\left(1+\gamma ^{ 1-q }\right)^3\left(1+\gamma ^{ q-1 }\right)^3\left(1+\gamma ^{2\left( 1-q \right)}\right)\left(1+\gamma ^{2\left( q-1 \right)}\right)$ on both sides of the Equation \eqref{E2.25}, we have 
				\begin{equation}\label{E2.26}
					\begin{aligned}
						&\left\{\Delta\cdot\varPsi\cdot\left(1+\gamma ^{i\left( q-1 \right)}\right)\left(1+\gamma ^{\left( q-1 \right)}\right)\left(1+\gamma ^{2\left( q-1 \right)}\right)\right.\\
						+&\left.\Delta^q\cdot\varUpsilon\cdot\left(1+\gamma ^{i\left( 1-q \right)}\right)\left(1+\gamma ^{\left( 1-q \right)}\right)\left(1+\gamma ^{2\left( 1-q \right)}\right)\right\}\cdot b_{k-2,0}^{q+1}\Omega\\
						=&b_{k-2,1}^{q+1}\cdot\varPsi\cdot\varUpsilon\cdot\left(1+\gamma ^{ 1-q }\right)\left(1+\gamma ^{ q-1 }\right) ,
					\end{aligned}
				\end{equation}
				where
				$$\varPsi=\left(1+\gamma ^{i\left( 1-q \right)}\right)^2\left(1+\gamma ^{2\left( 1-q \right)}\right)-\left(1+\gamma ^{2i\left( 1-q \right)}\right)\left(1+\gamma ^{\left( 1-q \right)}\right)^2,$$
				$$\varUpsilon=\left(1+\gamma ^{i\left( q-1 \right)}\right)^2\left(1+\gamma ^{2\left( q-1 \right)}\right)-\left(1+\gamma ^{2i\left( q-1 \right)}\right)\left(1+\gamma ^{\left( q-1 \right)}\right)^2$$
				and
				$$\varOmega=\left(1+\gamma ^{\left( q-1 \right)}\right)\left(1+\gamma ^{\left( 1-q \right)}\right)-\left(1+\gamma ^{i\left( q-1 \right)}\right)\left(1+\gamma ^{i\left( 1-q \right)}\right).$$
				Note that
				$$\varPsi=-2\gamma^{1-q}\left(\gamma ^{(i-1)\left( 1-q \right)}-1\right)\left(\gamma ^{(i+1)\left( 1-q \right)}-1\right),\varUpsilon=-2\gamma^{(q-1)}\left(\gamma ^{(i-1)\left( q-1 \right)}-1\right)\left(\gamma ^{(i+1)\left( q-1 \right)}-1\right)$$
				and
				$$\Omega=\gamma^{(q-1)}\left(\gamma ^{(i-1)\left( q-1 \right)}-1\right)\left(\gamma ^{(i+1)\left( 1-q \right)}-1\right),$$
				thus, the Equation \eqref{E2.26} is equivalent to
				\begin{equation}\label{E2.27}
					\begin{aligned}
						&\left\{-2\Delta\cdot\gamma^{1-q}\left(\gamma ^{(i-1)\left( 1-q \right)}-1\right)\left(\gamma ^{(i+1)\left( 1-q \right)}-1\right)\cdot\left(1+\gamma ^{i\left( q-1 \right)}\right)\left(1+\gamma ^{\left( q-1 \right)}\right)\left(1+\gamma ^{2\left( q-1 \right)}\right)\right.\\
						&-\left.2\Delta^q\cdot\gamma^{(q-1)}\left(\gamma ^{(i-1)\left( q-1 \right)}-1\right)\left(\gamma ^{(i+1)\left( q-1 \right)}-1\right)\cdot\left(1+\gamma ^{i\left( 1-q \right)}\right)\left(1+\gamma ^{\left( 1-q \right)}\right)\left(1+\gamma ^{2\left( 1-q \right)}\right)\right\}\\
						&\cdot \left[b_{k-2,0}^{q+1}\gamma^{(q-1)}\left(\gamma ^{(i-1)\left( q-1 \right)}-1\right)\left(\gamma ^{(i+1)\left( 1-q \right)}-1\right)\right]\\
						=&4b_{k-2,1}^{q+1}\cdot\gamma^{1-q}\left(\gamma ^{(i-1)\left( 1-q \right)}-1\right)\left(\gamma ^{(i+1)\left( 1-q \right)}-1\right)\cdot\gamma^{(q-1)}\left(\gamma ^{(i-1)\left( q-1 \right)}-1\right)\left(\gamma ^{(i+1)\left( q-1 \right)}-1\right)\\
						&\cdot \left(1+\gamma ^{ 1-q }\right)\left(1+\gamma ^{ q-1 }\right),
					\end{aligned}
				\end{equation}
				i.e., 
				\begin{equation}\label{E2.28}
					\begin{aligned}
						&\Delta\cdot\gamma^{1-q}\left(\gamma ^{(i-1)\left( 1-q \right)}-1\right)\left(\gamma ^{(i+1)\left( 1-q \right)}-1\right)\cdot\left(1+\gamma ^{i\left( q-1 \right)}\right)\left(1+\gamma ^{\left( q-1 \right)}\right)\left(1+\gamma ^{2\left( q-1 \right)}\right)\\
						+&\Delta^q\cdot\gamma^{(q-1)}\left(\gamma ^{(i-1)\left( q-1 \right)}-1\right)\left(\gamma ^{(i+1)\left( q-1 \right)}-1\right)\cdot\left(1+\gamma ^{i\left( 1-q \right)}\right)\left(1+\gamma ^{\left( 1-q \right)}\right)\left(1+\gamma ^{2\left( 1-q \right)}\right)\\
						=&-2\left(\frac{b_{k-2,1}}{b_{k-2,0}}\right)^{q+1}\cdot\gamma^{1-q}\left(\gamma ^{(i-1)\left( 1-q \right)}-1\right)\left(\gamma ^{(i+1)\left( q-1 \right)}-1\right)\left(1+\gamma ^{ 1-q }\right)\left(1+\gamma ^{ q-1 }\right),
					\end{aligned}
				\end{equation}
				Now by multiplying $\gamma^{(2i+1)(q-1)}$ on both sides of the Equation \eqref{E2.28}, we have
				\begin{equation}\label{E2.29}
					\begin{aligned}
						&\Delta\left(1-\gamma ^{(i-1)\left( q-1 \right)}\right)\left(1-\gamma ^{(i+1)\left( q-1 \right)}\right)\cdot\left(1+\gamma ^{i\left( q-1 \right)}\right)\left(1+\gamma ^{\left( q-1 \right)}\right)\left(1+\gamma ^{2\left( q-1 \right)}\right)\\
						+&\Delta^q\cdot\gamma^{(i-1)(q-1)}\left(\gamma ^{(i-1)\left( q-1 \right)}-1\right)\left(\gamma ^{(i+1)\left( q-1 \right)}-1\right)\cdot\left(1+\gamma ^{i\left( q-1 \right)}\right)\left(1+\gamma ^{\left( q-1 \right)}\right)\left(1+\gamma ^{2\left( q-1 \right)}\right)\\
						=&-2\left(\frac{b_{k-2,1}}{b_{k-2,0}}\right)^{q+1}\cdot\gamma^{i(q-1)}\left(1-\gamma ^{(i-1)\left( q-1 \right)}\right)\left(\gamma ^{(i+1)\left( q-1 \right)}-1\right)\left(1+\gamma ^{ q-1 }\right)\left(1+\gamma ^{ q-1 }\right) .
					\end{aligned}
				\end{equation}
				Note that  $1-\gamma ^{(i-1)\left( q-1 \right)}\in \mathbb{F}_{q^2}^{*}$ and $\gamma ^{(i+1)\left( q-1 \right)}-1\in \mathbb{F}_{q^2}^{*}$, and so, the Equation \eqref{E2.29} is equivalent to 
				\begin{equation}\label{E2.30}
					\Delta\left(1+\Delta^{q-1}\cdot\gamma^{(i-1)(q-1)} \right)\left(1+\gamma ^{i\left( q-1 \right)}\right)\left(1+\gamma ^{2\left( q-1 \right)}\right)
					=2\left(\frac{b_{k-2,1}}{b_{k-2,0}}\right)^{q+1}\gamma^{i(q-1)}\cdot\left(1+\gamma ^{\left( q-1 \right)}\right),
				\end{equation}
				i.e., $\varGamma =0$. 
				
				
				\quad
				
				\textbf{The proof of Lemma \ref{L2.13 }  } 
				
				By Lemma \ref{L2.5 }, we have $\gamma^{(j+2)(q-1)}=-1$, and so,  $\gamma^{(j+3)(q-1)}=-\gamma^{q-1}\neq 1$ and  $\gamma^{(j+1)(q-1)}=-\gamma^{1-q}\neq 1$. Note that $2\mid i$, thus,  $\gamma^{i(j+3)(q-1)}=\gamma^{i(q-1)}$ and $\gamma^{i(j+1)(q-1)}=\gamma^{i(1-q)}$. Furthermore, it is easy to know that 
				$$e_{j+1}=\left( q-1 \right)\sum\limits_{t=0}^{i-1}{ \gamma ^{t(j+1)\left( q-1 \right)}}=\left( q-1 \right) \frac{1-\gamma ^{i(j+1)\left( q-1 \right)}}{1-\gamma ^{(j+1)\left( q-1 \right)}}=(q-1)\frac{1-\gamma^{i(1-q)}}{1+\gamma^{(1-q)}}$$
				and	$$e_{j+3}=\left( q-1 \right)\sum\limits_{t=0}^{i-1}{ \gamma ^{t(j+3)\left( q-1 \right)}}=\left( q-1 \right) \frac{1-\gamma ^{i(j+3)\left( q-1 \right)}}{1-\gamma ^{(j+1)\left( q-1 \right)}}=(q-1)\frac{1-\gamma^{i(q-1)}}{1+\gamma^{(q-1)}}.$$
				Hence, by multiplying $\left(1+\gamma^{q-1}\right)\left(1+\gamma^{1-q}\right)$ on both sides of the Equation \eqref{E2.31}, we have 
				\begin{equation}
					\label{E2.33}
					b_{k-1,1}\left(\gamma^{i(1-q)}-1\right)\left(\gamma^{q-1}+1\right)+b_{k-1,1}^{q}\left(\gamma^{i(q-1)}-1\right)\left(\gamma^{1-q}+1\right)=0.
				\end{equation}
				Now by  multiplying $\gamma^{i(q-1)}$ on both sides of Equation \eqref{E2.33}, we have 
				\begin{equation}
					b_{k-1,1}\left(1-\gamma^{i(q-1)}\right)\left(\gamma^{q-1}+1\right)+b_{k-1,1}^{q}\gamma^{(i-1)(q-1)}\left(\gamma^{i(q-1)}-1\right)\left(1+\gamma^{q-1}\right)=0.
				\end{equation}
				i.e.,
				\begin{equation}
					\label{E2.34}
					b_{k-1,1}\left(1-\gamma^{i(q-1)}\right)\left(\gamma^{q-1}+1\right)\left(1-b_{k-1,1}^{q-1}\gamma^{(i-1)(q-1)}\right)=0.
				\end{equation}
				Note that  $1-\gamma^{i(q-1)}\in \mathbb{F}_{q^2}^{*}$ and $\gamma^{q-1}+1\in \mathbb{F}_{q^2}^{*}$, thus the Equation \eqref{E2.34} is equivalent to $b_{k-1,1}=0$ or 
				$b_{k-1,1}^{q-1}\gamma^{(i-1)(q-1)}=1$. 
				
				
				\newpage
				\textbf{The proof of Lemma \ref{L 1} }
				
				In the similar proof as those of Lemmas \ref{L2.13 }-\ref{L 1} is immediately.  
				
				\quad
				
				\textbf{The proof of Lemma \ref{L2.14 }  } 
				
				By $2\nmid i$ and Lemma \ref{L2.5 }, we have $e_{j+2}=q-1\neq 0$. 
				Then by multiplying $e_{j+2}$ on both sides of the Equation \eqref{E2.35},  
				\begin{equation}
					\label{E2.36}
					b_{k-1,1}e_{j+1}e_{j+2}+b_{k-1,0}^{q+1}\left(e_{j+2}^2-e_{j+1}e_{j+3}\right)+b_{k-1,1}^{q}e_{j+2}e_{j+3}=0. 
				\end{equation}
				Now by the proof of Lemma \ref{L2.11 }, we know that 
				$\gamma^{(j+1)(q-1)}=-\gamma^{1-q}\neq 1$,  $\gamma^{(j+3)(q-1)}=-\gamma^{q-1}\neq 1$ and  
				$$e_{j+1}=\left( q-1 \right)\sum\limits_{t=0}^{i-1}{ \gamma ^{t(j+1)\left( q-1 \right)}}=\left( q-1 \right) \frac{1-\gamma ^{i(j+1)\left( q-1 \right)}}{1-\gamma ^{(j+1)\left( q-1 \right)}}=(q-1)\frac{1+\gamma^{i(1-q)}}{1+\gamma^{(1-q)}}, $$
				$$e_{j+3}=\left( q-1 \right)\sum\limits_{t=0}^{i-1}{ \gamma ^{t(j+3)\left( q-1 \right)}}=\left( q-1 \right) \frac{1-\gamma ^{i(j+3)\left( q-1 \right)}}{1-\gamma ^{(j+1)\left( q-1 \right)}}=(q-1)\frac{1+\gamma^{i(q-1)}}{1+\gamma^{(q-1)}}. $$
				Hence, the Equation \eqref{E2.36} is equivalent to 
				\begin{equation}
					\label{E2.37}
					b_{k-1,1}\left(\frac{1+\gamma^{i(1-q)}}{1+\gamma^{1-q}}\right)+b_{k-1,0}^{q+1}\left(1-\frac{1+\gamma^{i(1-q)}}{1+\gamma^{1-q}}\cdot \frac{1+\gamma^{i(q-1)}}{1+\gamma^{(q-1)}}\right)+b_{k-1,1}^{q}\left(\frac{1+\gamma^{i(q-1)}}{1+\gamma^{(q-1)}}\right)=0. 
				\end{equation}
				Note that $\left(1+\gamma^{1-q}\right)\in \mathbb{F}_{q^2}^{*}$ and $\left(1+\gamma^{q-1}\right)\in \mathbb{F}_{q^2}^{*}$, and so, by multiplying $\left(1+\gamma^{1-q}\right)\left(1+\gamma^{q-1}\right)$ on both sides of the Equation \eqref{E2.37}, we have 
				\begin{equation}
					\label{E2.38}
					\begin{aligned}
						&b_{k-1,1}\left(1+\gamma^{i(1-q)}\right)\left(1+\gamma^{q-1}\right)+b_{k-1,1}^q\left(1+\gamma^{i(q-1)}\right)\left(1+\gamma^{1-q}\right)\\
						=&-b_{k-1,0}^{q+1}\left[\left(1+\gamma^{1-q}\right)\left(1+\gamma^{q-1}\right)-\left(1+\gamma^{i(1-q)}\right)\left(1+\gamma^{i(q-1)}\right)\right]. 
					\end{aligned}
				\end{equation}
				Next, by expanding and simplifying the Equation \eqref{E2.38}, we have 
				\begin{equation}
					\label{E2.39}
					b_{k-1,1}\left(1+\gamma^{i(1-q)}\right)\left(1+\gamma^{q-1}\right)+b_{k-1,1}^q\left(1+\gamma^{i(q-1)}\right)\left(1+\gamma^{1-q}\right)
					=-b_{k-1,0}^{q+1}\gamma^{q-1}\left(\gamma^{(i-1)(q-1)}-1\right)\left(\gamma^{(i+1)(1-q)}-1\right). 
				\end{equation}
				Now by multiplying $\gamma^{i(q-1)}$ on both sides of the Equation \eqref{E2.39}, we have
				$$ 
				\begin{aligned}
					&b_{k-1,1}\left(\gamma^{i(q-1)}+1\right)\left(1+\gamma^{q-1}\right)+b_{k-1,1}^{q}\gamma^{(i-1)(q-1)}\left(1+\gamma^{i(q-1)}\right)\left(\gamma^{q-1}+1\right)\\
					=&-b_{k-1,0}^{q+1}\left(\gamma^{(i-1)(q-1)}-1\right)\left(1-\gamma^{(i+1)(q-1)}\right),
				\end{aligned}
				$$
				which is equivalent to  $$b_{k-1,1}\left(\gamma^{i(q-1)}+1\right)\left(1+\gamma^{q-1}\right)\left(1+b_{k-1,1}^{q-1}\gamma^{(i-1)(q-1)}\right)=b_{k-1,0}^{q+1}\left(\gamma^{(i-1)(q-1)}-1\right)\left(\gamma^{(i+1)(q-1)}-1\right),$$
				i.e., $\varGamma_{1}=0$. 
				
				
				\section{The proof of Proposition \ref{P3.1}}
				\label{appendex B}
				We present the detailed proofs of  $6$ cases in Proposition \ref{P3.1}. 
				
				\textbf{case 1.}  For $ u,v\in \left\{1, 2, \ldots , k-2 \right\} $. 
				
				(1) For $u=1$, it is easy to know that  $u+v-2\in \left\{ 0,1,\ldots ,k-3 \right\}$. Now by  $q-1\le k-3=q+j-3<2q-2$, we know that $a_{1v} \neq 0$ if and only if $u+v-2=0$ or $q-1$, i.e.,  $v=1$ or $v=q$. Hence,
				$$a_{11}=q-1, a_{1q}=\left( q-1 \right) \beta ^{q\left( q-1 \right)}.$$
				
				(2) For $2\le u\le q-j+1=j+4$. it is easy to know that  $ u+v-2\in \left\{ u-1, u, \ldots , u+k-4 \right\}$. Now by $0<1\le u-1\le q-j<q-1 $ and $ q-1< q+j-2 \le u+k-4=u+q+j-4\le 2q-3<2q-2$, we know that $a_{uv} \neq 0$ if and only if $u+v-2=q-1$, i.e.,  $v=q-u+1$. Hence, 
				$$a_{u(q-u+1)}=\left( q-1 \right) \beta ^{(u-1)+q(q-u)}=(q-1)\beta^{(q-1)(q-u+1)}.$$
				
				(3) For $j+5=q-j+2\le u \le q-1$, it is easy to know that $ u+v-2\in \left\{ u-1, u, \ldots , u+k-4 \right\}$. Now by  $0<q-j+1\le u-1\le q-2<q-1$ and $2q-2\le u+k-4 \le q+k-5=2q+j-5=3q-j-8<3q-3$, we know that $a_{uv} \neq 0$ if and only if $u+v-2=q-1$ or $2q-2$, i.e.,  $v=q-u+1$ or $2q-u$. Hence, 
				$$a_{u(q-u+1)}=\left( q-1 \right) \beta ^{(u-1)+q(q-u)}=(q-1)\beta^{(q-1)(q-u+1)}, $$
				$$a_{u(2q-u)}=\left( q-1 \right) \beta ^{(u-1)+q(2q-u-1)}=(q-1)\beta^{(q-1)(q-u)}. $$
				
				(4) For $u=q$, it is easy to know that $ u+v-2\in \left\{ q-1, q, \ldots , q+k-4 \right\}$. Now by  $2q-2\le  q+k-4=2q+j-4 =3q-j-7<3q-3$, we know that $a_{qv} \neq 0$ if and only if $q+v-2=q-1$ or $2q-2$, i.e.,  $v=1$ or $q$. Hence, 
				$$a_{q1}=\left( q-1 \right) \beta ^{q-1} 
				\text{ and } a_{qq}=\left( q-1 \right) \beta ^{(q-1)+q(q-1)}=(q-1)\beta^{q^2-1}=q-1. $$
				
				(5) For $q+1\le u\le k-2=q+j-2$, it is easy to know that  $ u+v-2\in \left\{ u-1, u, \ldots , u+k-4 \right\}$. Now by $q-1<q\le u-1\le q+j-3=2q-j-6<2q-2 $ and $ 2q-2< 2q+j-3 \le u+k-4=u+q+j-4\le 2q+2j-6=3q-9<3q-3$, we know that $a_{uv} \neq 0$ if and only if $u+v-2=2q-2$, i.e.,  $v=2q-u$. Hence, 
				$$a_{u(2q-u)}=\left( q-1 \right) \beta ^{(u-1)+q(2q-u-1)}=(q-1)\beta^{(q-1)(2q-u+1)}.$$
				
				\textbf{case 2.} For $u=k-1, v\in \left\{1, 2, \ldots , k-2 \right\}$, it's easy to get   
				\begin{equation}
					\label{E3.5}
					\begin{aligned}
						a_{\left( k-1 \right) v}&=\sum_{s=1}^{q-1}\left[ \left( \alpha _s\beta \right) ^{k-2+\left( v-1 \right) q}+b_{k-2,0}\left( \alpha _s\beta \right) ^{k+\left( v-1 \right) q}+b_{k-2,1}\left( \alpha _s\beta \right) ^{k+1+\left( v-1 \right) q} \right]\\
						&=\sum_{s=1}^{q-1}\left( \alpha _s\beta \right) ^{k-2+\left( v-1 \right) q}+b_{k-2,0}\sum_{s=1}^{q-1}\left( \alpha _s\beta \right) ^{k+\left( v-1 \right) q}+b_{k-2,1}\sum_{s=1}^{q-1}\left( \alpha _s\beta \right) ^{k+1+\left( v-1 \right) q}. 
					\end{aligned}
				\end{equation}
				
				Now by Remark \ref{u+v-2}, we know that $\sum\limits_{s=1}^{q-1}\left( \alpha _s\beta \right) ^{k-2+\left( v-1 \right) q}\neq 0$ if and only if 
				\begin{equation}
					\label{k-3+v}
					k-3+v\equiv 0(\bmod (q-1)),
				\end{equation} 
				$\sum\limits_{s=1}^{q-1}\left( \alpha _s\beta \right) ^{k+\left( v-1 \right) q}\neq 0$ if and only if 
				\begin{equation}
					\label{k-1+v}
					k-1+v \equiv 0(\bmod (q-1))
				\end{equation}
				 and  $\sum\limits_{s=1}^{q-1}\left( \alpha _s\beta \right) ^{k+1+\left( v-1 \right) q}\neq 0$ if and only if
				 \begin{equation}
				 	\label{k+v}
				 	k+v\equiv 0(\bmod (q-1)).
				 \end{equation}
				
				Note that $q\ge 7$, thus, any two congruences among \ref{k-3+v}, \ref{k-1+v} and \ref{k+v} cannot hold simultaneously.
				 Hence, any two of the sums
				 $\sum\limits_{s=1}^{q-1}\left(\alpha_s\beta\right)^{k-2+(v-1)q}$, 
				 $\sum\limits_{s=1}^{q-1}\left(\alpha_s\beta\right)^{k+(v-1)q}$ and 
				 $\sum\limits_{s=1}^{q-1}\left(\alpha_s\beta\right)^{k+1+(v-1)q}$
				 are not non-zeros simultaneously.

				
				(1) For $\sum\limits_{s=1}^{q-1}\left( \alpha _s\beta \right) ^{k-2+\left( v-1 \right) q}$, we have $k+v-3\in \left\{ k-2,k-1,\ldots ,2k-5 \right\} $. Now by  $q-1<k-2=q+j-2=2q-j-5<2q-2$ and $2q-2<2k-5=2q+2j-5=3q-8<3q-3$, we know that $\sum\limits_{s=1}^{q-1}\left( \alpha _s\beta \right) ^{k-2+\left( v-1 \right) q}\neq 0$ if and only if $k+v-3=2q-2$, i.e., $v=q-j+1=j+4 $. Hence, 
				\begin{equation}
					\label{E3.6}
					\begin{aligned}
						a_{\left( k-1 \right) (j+4)}&=\sum\limits_{s=1}^{q-1}\left( \alpha _s\beta \right) ^{k-2+\left( j+3 \right) q}=(q-1)\beta^{k-2+(j+3)q}=(q-1)\beta^{q+j-2+(q-j)q}\\
						&=(q-1)\beta^{(1-q)(j-1)}=(q-1)\beta^{(q-1)(1-j+q+1)}=(q-1)\beta^{(q-1)(j+5)}.
					\end{aligned}
				\end{equation}

				(2) For $\sum\limits_{s=1}^{q-1}\left( \alpha _s\beta \right) ^{k+\left( v-1 \right) q}$, we have $k+v-1\in \left\{ k,k+1,\ldots ,2k-3 \right\} $. Now by  $q-1<k=q+j=2q-j-3<2q-2$ and $2q-2<2k-3=2q+2j-3=3q-6<3q-3$, we know that  $\sum\limits_{s=1}^{q-1}\left( \alpha _s\beta \right) ^{k+\left( v-1 \right) q}\neq 0$ if and only if $k+v-1=2q-2$, i.e., $v=q-j-1=j+2 $. Hence, 
				\begin{equation}
					\label{E3.7}
					\begin{aligned}
						a_{\left( k-1 \right) (j+2)}&=b_{k-2,0}\sum\limits_{s=1}^{q-1}\left( \alpha _s\beta \right) ^{k+(j+1)q} =b_{k-2,0}(q-1)\beta^{k+(j+1)q}=b_{k-2,0}(q-1)\beta^{q+j+(q-j-2)q}\\
						&=b_{k-2,0}(q-1)\beta^{(1-q)(j+1)}=b_{k-2,0}(q-1)\beta^{(q-1)(q-j)}=b_{k-2,0}(q-1)\beta^{(q-1)(j+3)}.
					\end{aligned}
				\end{equation}

				(3) For $\sum\limits_{s=1}^{q-1}\left( \alpha _s\beta \right) ^{k+1+\left( v-1 \right) q}$, we have  $ k+v\in \left\{ k+1,k+2,\ldots ,2k-2 \right\} $. Now by  $q-1<k+1=q+j+1=2q-j-2<2q-2$ and $2q-2<2k-2=2q+2j-2=3q-5<3q-3$, we know that  $\sum\limits_{s=1}^{q-1}\left( \alpha _s\beta \right) ^{k+1+\left( v-1 \right) q}\neq 0$ if and only if $k+v=2q-2$, i.e., $v=q-j-2=j+1 $. Hence, 
				\begin{equation}
					\label{E3.8}
					\begin{aligned}
						a_{\left( k-1 \right) (j+1)}&=b_{k-2,1}\sum\limits_{s=1}^{q-1}\left( \alpha _s\beta \right) ^{k+1+qj} =b_{k-2,1}(q-1)\beta^{k+1+qj}=b_{k-2,1}(q-1)\beta^{q+j+1+(q-j-3)q}\\
						&=b_{k-2,1}(q-1)\beta^{(1-q)(j+2)}=b_{k-2,1}(q-1)\beta^{(q-1)(q-j-1)}=b_{k-2,1}(q-1)\beta^{(q-1)(j+2)}.
					\end{aligned}
				\end{equation}
				
				\textbf{case 3.} For $u=k, v\in \left\{1, 2, \ldots , k-2 \right\}$, it's easy to get 
				\begin{equation}
					\label{E3.9}
					\begin{aligned}
						a_{kv}&=\sum_{s=1}^{q-1}\left[\left( \alpha _s\beta \right) ^{k-1+\left( v-1 \right) q}+b_{k-1,0}\left( \alpha _s\beta \right) ^{k+\left( v-1 \right) q}+b_{k-1,1}\left( \alpha _s\beta \right) ^{k+1+\left( v-1 \right) q}\right]\\
						&=\sum_{s=1}^{q-1}\left( \alpha _s\beta \right) ^{k-1+\left( v-1 \right) q}+
						b_{k-1,0}\sum_{s=1}^{q-1}\left( \alpha _s\beta \right) ^{k+\left( v-1 \right) q}
						+b_{k-1,1}\sum_{s=1}^{q-1}\left( \alpha _s\beta \right) ^{k+1+\left( v-1 \right) q}. 
					\end{aligned}
				\end{equation}
				
				In the similar proof as that of Case 2, we know that  
				$\sum\limits_{s=1}^{q-1}\left( \alpha _s\beta \right) ^{k-1+\left( v-1 \right) q}\neq 0$ if and only if $k-2+v\equiv 0(\bmod (q-1))$, and any two of the sums  $\sum\limits_{s=1}^{q-1}\left( \alpha _s\beta \right) ^{k-1+\left( v-1 \right) q}$, $\sum\limits_{s=1}^{q-1}\left( \alpha _s\beta \right) ^{k+\left( v-1 \right) q}$  and   $\sum\limits_{s=1}^{q-1}\left( \alpha _s\beta \right) ^{k+1+\left( v-1 \right) q}$ are not non-zeros simultaneously. 
				
				(1) For $\sum\limits_{s=1}^{q-1}\left( \alpha _s\beta \right) ^{k-1+\left( v-1 \right) q}$, we have $k+v-2\in \left\{ k-1,k,\ldots ,2k-4 \right\} $. Now by $q-1<k-1=q+j-1=2q-j-4<2q-2$ and $2q-2<2k-4=2q+2j-4=3q-7<3q-3$, we know that  $\sum\limits_{s=1}^{q-1}\left( \alpha _s\beta \right) ^{k-1+\left( v-1 \right) q}\neq 0$ if and only if $k+v-2=2q-2$, i.e., $v=q-j=j+3 $. Hence, 
				$$
				\begin{aligned}
					a_{k (j+3)}&=\sum\limits_{s=1}^{q-1}\left( \alpha _s\beta \right) ^{k-1+\left( j+2 \right) q}=(q-1)\beta^{k-1+(j+2)q}=(q-1)\beta^{q+j-1+(q-j-1)q}\\
					&=(q-1)\beta^{(1-q)j}=(q-1)\beta^{(q-1)(q+1-j)}=(q-1)\beta^{(q-1)(j+4)}.
				\end{aligned}
				$$
				
				(2) For $\sum\limits_{s=1}^{q-1}\left( \alpha _s\beta \right) ^{k+\left( v-1 \right) q}$, by the Equation \eqref{E3.7},  we can get 
				$a_{k (j+2)}=b_{k-1,0}(q-1)\beta^{(q-1)(j+3)}. $
				
				(3) For $\sum\limits_{s=1}^{q-1}\left( \alpha _s\beta \right) ^{k+1+\left( v-1 \right) q}$, by the Equation \eqref{E3.8} , we can get 
				$a_{k (j+1)}=b_{k-1,1}(q-1)\beta^{(q-1)(j+2)}$. \\
				
				\textbf{case 4.} For $u\in {1, 2, \ldots , k-2}, v=k-1$, it's easy to get 
				\begin{equation}
					\label{E3.10}
					\begin{aligned}
						a_{u\left( k-1 \right)}&=\sum_{s=1}^{q-1}\left[\left( \alpha _s\beta \right) ^{q\left( k-2 \right) +u-1}+b_{k-2,0}^{q}\left( \alpha _s\beta \right) ^{qk+u-1}+b_{k-2,1}^{q}\left( \alpha _s\beta \right) ^{q\left( k+1 \right) +u-1}\right]\\
						&=\sum_{s=1}^{q-1}\left( \alpha _s\beta \right) ^{q\left( k-2 \right) +u-1}+b_{k-2,0}^{q}\sum_{s=1}^{q-1}\left( \alpha _s\beta \right) ^{qk+u-1}+b_{k-2,1}^{q}\sum_{s=1}^{q-1}\left( \alpha _s\beta \right) ^{q\left( k+1 \right) +u-1}. 
					\end{aligned}
				\end{equation}
				
				Now by Remark \ref{u+v-2}, we know that  $\sum\limits_{s=1}^{q-1}\left( \alpha _s\beta \right) ^{(k-2)q+ u-1}\neq 0$ if and only if 
				\begin{equation}
					\label{k-3+u}
					k-3+u\equiv 0(\bmod (q-1)),
				\end{equation}
				 $\sum\limits_{s=1}^{q-1}\left( \alpha _s\beta \right) ^{qk+u-1}\neq 0$ if and only if 
				 \begin{equation}
				 	\label{k-1+u}
				 	k-1+u \equiv 0(\bmod (q-1))
				 \end{equation}
				 and $\sum\limits_{s=1}^{q-1}\left( \alpha _s\beta \right) ^{(k+1)q+u-1}\neq 0$ if and only if 
				 \begin{equation}
				 	\label{k+u}
				 	k+u\equiv 0(\bmod (q-1)). 
				 \end{equation}
				Note that $q\ge 7$, thus, any two congruences among \ref{k-3+u}, \ref{k-1+u} and \ref{k+u} cannot hold simultaneously.
				Hence, any two of the sums
				$\sum\limits_{s=1}^{q-1}\left( \alpha _s\beta \right) ^{(k-2)q+ u-1}$, 
				$\sum\limits_{s=1}^{q-1}\left( \alpha _s\beta \right) ^{qk+u-1}$ and  
				$\sum\limits_{s=1}^{q-1}\left( \alpha _s\beta \right) ^{(k+1)q+u-1}$
				are not non-zeros simultaneously.

				(1) For $\sum\limits_{s=1}^{q-1}\left( \alpha _s\beta \right) ^{q\left( k-2 \right) +u-1}$, we have $k+u-3\in \left\{ k-2,k-1,\ldots ,2k-5 \right\} $. Now by $q-1<k-2=q+j-2=2q-j-5<2q-2$ and $2q-2<2k-5=2q+2j-5=3q-8<3q-3$, we know that $\sum\limits_{s=1}^{q-1}\left( \alpha _s\beta \right) ^{q\left( k-2 \right) +u-1}\neq 0$ if and only if $k+u-3=2q-2$, i.e., $u=q-j+1=j+4 $. Hence, 
				\begin{equation}
					\label{E3.11}
					a_{(j+4)\left( k-1 \right) }=\sum\limits_{s=1}^{q-1}\left( \alpha _s\beta \right) ^{(k-2)q+j+3 }=(q-1)\beta^{(k-2)q+j+3}=(q-1)\beta^{(q+j-2)q+q-j}=(q-1)\beta^{(q-1)(j-1)}. 
				\end{equation}
				
				(2) For $\sum\limits_{s=1}^{q-1}\left( \alpha _s\beta \right) ^{qk+ u-1 }$, we have $k+u-1\in \left\{ k,k+1,\ldots ,2k-3 \right\} $. Now by $q-1<k=q+j=2q-j-3<2q-2$ and $2q-2<2k-3=2q+2j-3=3q-6<3q-3$, we know that $\sum\limits_{s=1}^{q-1}\left( \alpha _s\beta \right) ^{qk+u-1}\neq 0$ if and only if $k+u-1=2q-2$, i.e., $u=q-j-1=j+2 $. Hence, 
				\begin{equation}
					\label{E3.12}
					a_{(j+2)\left( k-1 \right) }=b_{k-2,0}^q\sum\limits_{s=1}^{q-1}\left( \alpha _s\beta \right) ^{qk+j+1} =b_{k-2,0}^q(q-1)\beta^{qk+j+1}=b_{k-2,0}^q(q-1)\beta^{q(q+j)+q-j-2}=b_{k-2,0}^q(q-1)\beta^{(q-1)(j+1)}.
				\end{equation}
				
				(3) For $\sum\limits_{s=1}^{q-1}\left( \alpha _s\beta \right) ^{q(k+1)+u-1  }$, we have  $ k+u\in \left\{ k+1,k+2,\ldots ,2k-2 \right\} $. Now by $q-1<k+1=q+j+1=2q-j-2<2q-2$ and $2q-2<2k-2=2q+2j-2=3q-5<3q-3$, we know that  $\sum\limits_{s=1}^{q-1}\left( \alpha _s\beta \right) ^{q(k+1)+u-1  }\neq 0$ if and only if $k+u=2q-2$, i.e., $u=q-j-2=j+1 $. Hence, 
				\begin{equation}
					\label{E3.13}
					a_{(j+1)\left( k-1 \right) }=b_{k-2,1}^q\sum\limits_{s=1}^{q-1}\left( \alpha _s\beta \right) ^{(k+1)q+j} =b_{k-2,1}^q(q-1)\beta^{(k+1)q+j}=b_{k-2,1}^q(q-1)\beta^{(q+j+1)q+q-j-3}=b_{k-2,1}^q(q-1)\beta^{(q-1)(j+2)}.
				\end{equation}
				
				\textbf{case 5.} When $u\in \left\{1, 2, \ldots , k-2 \right\}, v=k$, it is easy to know that 
				\begin{equation}
					\label{E3.14}
					\begin{aligned}
						a_{uk}&=\sum_{s=1}^{q-1}\left[\left( \alpha _s\beta \right) ^{q\left( k-1 \right) +u-1}+b_{k-1,0}^{q}\left( \alpha _s\beta \right) ^{qk+u-1}+b_{k-1,1}^{q}\left( \alpha _s\beta \right) ^{q\left( k+1 \right) +u-1}\right]\\
						&=\sum_{s=1}^{q-1}\left( \alpha _s\beta \right) ^{q\left( k-1 \right) +u-1}+b_{k-1,0}^{q}\sum_{s=1}^{q-1}\left( \alpha _s\beta \right) ^{qk+u-1}+ b_{k-1,1}^{q}\sum_{s=1}^{q-1}\left( \alpha _s\beta \right) ^{q\left( k+1 \right) +u-1} . 
					\end{aligned}
				\end{equation}
				
				In the similar proof as that of Case 4, we know that  
				$\sum\limits_{s=1}^{q-1}\left( \alpha _s\beta \right) ^{q\left( k-1 \right) +u-1}\neq 0$ if and only if $k-2+u\equiv 0(\bmod (q-1))$, and any two of the sums  $\sum\limits_{s=1}^{q-1}\left( \alpha _s\beta \right) ^{q\left( k-1 \right) +u-1}$, $\sum\limits_{s=1}^{q-1}\left( \alpha _s\beta \right) ^{qk+u-1}$ and  $\sum\limits_{s=1}^{q-1}\left( \alpha _s\beta \right) ^{q\left( k+1 \right) +u-1}$ are not non-zeros simultaneously. 
				
				(1) For $\sum\limits_{s=1}^{q-1}\left( \alpha _s\beta \right) ^{q\left( k-1 \right) +u-1}$, we have $k+u-2\in \left\{ k-1,k,\ldots ,2k-4 \right\} $. Now by $q-1<k-1=q+j-1=2q-j-4<2q-2$ and $2q-2<2k-4=2q+2j-4=3q-7<3q-3$, we know that $\sum\limits_{s=1}^{q-1}\left( \alpha _s\beta \right) ^{q\left( k-1 \right) +u-1}\neq 0$ if and only if $k+u-2=2q-2$, i.e., $u=q-j=j+3 $. Hence, 
				$$a_{ (j+3)k}=\sum\limits_{s=1}^{q-1}\left( \alpha _s\beta \right) ^{(k-1)q+ j+2}=(q-1)\beta^{(k-1)q+ j+2}=(q-1)\beta^{(q+j-1)q+q-j-1}=(q-1)\beta^{(q-1)j}. $$
				
				(2) For $\sum\limits_{s=1}^{q-1}\left( \alpha _s\beta \right) ^{qk+u-1}$, by the Equation \eqref{E3.12},  we can get
				$a_{(j+2) k}=b_{k-1,0}^q(q-1)\beta^{(q-1)(j+1)}$. 
				
				(3) For $\sum\limits_{s=1}^{q-1}\left( \alpha _s\beta \right) ^{q\left( k+1 \right) +u-1}$, by the Equation \eqref{E3.13}, we can get
				$a_{(j+1) k}=b_{k-1,1}^q(q-1)\beta^{(q-1)(j+2)}$. 
				
				\textbf{case 6.} For  $u,v \in \{k-1,k\}$, by directly calculating the  Equation \eqref{a uv}, we have
				$$
				\begin{aligned}
					\label{E3.15}
					A_1&=\left( q-1 \right) b_{k-2,0}^{q+1}\beta ^{\left( q+1 \right) k}
					=\left( q-1 \right) b_{k-2,0}^{q+1}\beta ^{qk+q}
					=\left( q-1 \right) b_{k-2,0}^{q+1}\beta ^{q(q+j)+2q-j-3}
					=\left( q-1 \right) b_{k-2,0}^{q+1}\beta ^{(q-1)(j+2)};\\
					A_2&=\left( q-1 \right) \left[ b_{k-2,0}b_{k-1,0}^{q}\beta ^{\left( q+1 \right) k}+b_{k-2,1}\beta ^{q\left( k-1 \right) +k+1} \right] 
					=\left( q-1 \right) \left[ b_{k-2,0}b_{k-1,0}^{q}\beta ^{(q-1)(j+2)}+b_{k-2,1}\beta ^{q\left( q+j-1 \right) +2q-j-2} \right]\\
					&=\left( q-1 \right) \left[ b_{k-2,0}b_{k-1,0}^{q}\beta ^{(q-1)(j+2)}+b_{k-2,1}\beta ^{(q-1)(j+1)} \right];\\
					A_3&=\left( q-1 \right) \left[ b_{k-2,1}^{q}\beta ^{q\left( k+1 \right) +k-1}+b_{k-2,0}^{q}b_{k-1,0}\beta ^{\left( q+1 \right) k} \right] 
					=\left( q-1 \right) \left[ b_{k-2,1}^{q}\beta ^{q\left( q+j+1 \right) +2q-j-4}+b_{k-2,0}^{q}b_{k-1,0}\beta ^{(q-1)(j+2)} \right] \\
					&=\left( q-1 \right) \left[ b_{k-2,1}^{q}\beta ^{(q-1)(j+3)}+b_{k-2,0}^{q}b_{k-1,0}\beta ^{(q-1)(j+2)} \right] ; \\
					A_4&=\left( q-1 \right) \left[ b_{k-1,1}^{q}\beta ^{q\left( k+1 \right) +k-1}+b_{k-1,0}^{q+1}\beta ^{\left( q+1 \right) k}+b_{k-1,1}\beta ^{q\left( k-1 \right) +k+1} \right]\\
					&=\left( q-1 \right) \left[ b_{k-1,1}^{q}\beta ^{(q-1)(j+3)}+b_{k-1,0}^{q+1}\beta ^{(q-1)(j+2)}+b_{k-1,1}\beta ^{(q-1)(j+1)} \right].
				\end{aligned}
				$$

			\end{document}